\documentclass[a4paper,USenglish,cleveref, autoref, thm-restate]{lipics-v2021}
\hideLIPIcs
\nolinenumbers
\makeatletter 
\def\@oddfoot{} 
\makeatother

\title{An Infinitary and a Cyclic Sequent Calculus\\ for Non-Monotone Inductive Definitions} 
\titlerunning{An Infinitary and a Cyclic Sequent Calculus for Non-Monotone Inductive Definitions} 

\author{Robbe {Van den Eede}}{KU Leuven, Department of Computer Science, Belgium \and Vrije Universiteit
	Brussel, Department of Computer Science, Belgium} 
	{robbe.vandeneede@kuleuven.be}
	{https://orcid.org/0000-0002-2579-9053}
	{This work was supported by Fonds Wetenschappelijk Onderzoek Vlaanderen (FWO Flanders) under the project 11A2R26N.}%

\authorrunning{R. Van den Eede} %

\Copyright{Robbe Van den Eede} %



\acknowledgements{I would like to thank Bart Bogaerts and Marc Denecker for their insightful feedback and valuable discussions concerning the ideas presented in this paper.}%

\usepackage{mathtools} %
\usepackage{ebproof} %
\usepackage{xcolor}
\usepackage{xspace}
\usepackage{amsmath}
\usepackage{amsthm}
\usepackage{amsfonts}
\usepackage{amssymb}
\usepackage{tikz}\usetikzlibrary{arrows,fit,shapes.symbols,patterns}\usetikzlibrary{positioning}
\tikzset{just/.style={draw,->,thick}}
\usepackage{thmtools}

\newcommand{\longonly}[1]{\ifhmode\unskip\fi\ignorespaces}
\newcommand{\shortonly}[1]{#1}

\newcommand{\vsp}{\vspace{2mm}}

\newcommand{\nat}{\mathbb{N}}

\newcommand{\fo}{{\normalfont \textrm{FO}}\xspace}
\newcommand{\foid}{{\normalfont \textrm{FO(ID)}}\xspace}

\newcommand{\lk}{{\normalfont \textrm{LK}}\xspace}
\newcommand{\lj}{{\normalfont \textrm{LJ}}\xspace}
\newcommand{\lkid}{{\normalfont \textrm{LKID}}\xspace}
\newcommand{\lkidinf}{{\normalfont \textrm{LKID}$^\omega$}\xspace}
\newcommand{\clkid}{{\normalfont \textrm{CLKID}$^\omega$}\xspace}

\newcommand{\scfo}{{\normalfont\textrm{SC}$_{\fo}$}\xspace}
\newcommand{\scfoid}{{\normalfont\textrm{SC}$_{\foid}$}\xspace}
\newcommand{\scinf}{{\normalfont\textrm{SC}$_{\scriptsize\foid}^\infty$}\xspace}
\newcommand{\logic}{$\mathcal{L}$\xspace}
\newcommand{\scginf}{{\normalfont\textrm{SC}$^\infty$}\xspace}
\newcommand{\sccyc}{{\normalfont\textrm{SC}$_{\scriptsize\foid}^\circlearrowleft$}\xspace}
\newcommand{\scgcyc}{{\normalfont\textrm{SC}$^\circlearrowleft$}\xspace}

\newcommand{\voc}[1]{\mathsf{voc}(#1)}
\newcommand{\vocab}{\Sigma}
\newcommand{\nlsym}{\gamma}
\newcommand{\ar}[1]{\mathsf{ar}(#1)}
\newcommand{\form}{\varphi}
\newcommand{\formtwo}{\psi}
\newcommand{\formthree}{\chi}
\newcommand{\free}[1]{\normalfont\textrm{Free}(#1)}

\newcommand{\rul}{\leftarrow}
\newcommand{\defrul}{\forall \bar{x}: P(\bar{t}) \rul \form}
\newcommand{\defn}{\Phi}
\newcommand{\defns}{\defn_1, \dots, \defn_n}
\newcommand{\defin}[1]{\left\{ 
	\begin{array}{l} #1 \end{array} 
	\right\}}
\newcommand{\defintight}[1]{\left\{ 
	\hspace{-5pt} 
	\begin{array}{l} #1 \end{array} 
	\hspace{-5pt} 
	\right\}}
\newcommand{\defp}[1]{\mathrm{def}(#1)}
\newcommand{\pars}[1]{\mathrm{pars}(#1)}
\newcommand{\sym}[1]{\mathrm{sym}(#1)}

\newcommand{\dom}[1]{\mathsf{dom}(#1)}
\newcommand{\true}{\mathbf{t}}
\newcommand{\false}{\mathbf{f}}
\newcommand{\unknown}{\mathbf{u}}
\newcommand{\str}{\mathcal{I}}
\newcommand{\strtwo}{\mathcal{J}}
\newcommand{\leqt}{\leq_{\mathrm{t}}}

\newcommand{\leqp}{\leq_{\mathrm{p}}}
\newcommand{\tvstr}{\mathfrak{A}}
\newcommand{\tvstrtwo}{\mathfrak{B}}
\newcommand{\cont}{\mathcal{O}}
\newcommand{\wfind}{(\tvstr_i)_{0 \leq i \leq \beta}}
\newcommand{\modelswf}{\models_{\mathrm{wf}}}
\newcommand{\wfmod}{\mathcal{W}}

\newcommand{\consts}[1]{C_{#1}} %
\newcommand{\facts}[2]{\normalfont\textrm{Facts}(#1,#2)}
\newcommand{\fact}{\alpha}

\newcommand{\labfun}{\ell}
\newcommand{\branch}{\mathbf{b}}
	
\newcommand{\seq}{\Gamma \vdash \Delta}
\newcommand{\seqshort}{\mathcal{S}}
\newcommand{\seqs}{\mathsf{Seqs}}
\newcommand{\regseq}{\defns, \Gamma \vdash \Delta}
\newcommand{\sregseq}{\defn, \Gamma \vdash \Delta}
\newcommand{\deriv}{\mathcal{D}}
\newcommand{\seqfun}{\lambda}
\newcommand{\rulfun}{\rho}
\newcommand{\prf}{\mathcal{P}}
	
\newcommand{\pth}{\pi}
\newcommand{\trace}{\tau}
\newcommand{\repfun}{\mathcal{R}}
\newcommand{\graph}[1]{\mathcal{G}_{#1}}
\newcommand{\trunf}[1]{\mathcal{T}_{#1}}
\newcommand{\comp}{\hspace{2mm} (*)}
\newcommand{\comptwo}{\hspace{2mm} (\dagger)}
\newcommand{\eqp}{\approx}
\newcommand{\peq}{\simeq}
\newcommand{\branchset}[1]{\normalfont\textrm{Branch}(#1)}

\newcommand{\intset}{\mathbb{I}}
\newcommand{\tset}{\mathcal{T}}
\newcommand{\tval}{\normalfont\textrm{TVal}}
\newcommand{\seqsabst}{\normalfont\textrm{Seqs}}
\newcommand{\tpair}{\normalfont\textrm{TPair}}
\newcommand{\rules}{\normalfont\textrm{Rules}}
\newcommand{\otf}{\sigma}
\newcommand{\ord}{\normalfont\textrm{Ord}}
\newcommand{\seqabst}{S}

\newcommand{\structord}{\leq_\prf}
\newcommand{\indord}{\lhd}
\newcommand{\budset}{\mathcal{B}}
\newcommand{\bacyc}[1]{\mathcal{C}_{#1}}

\newcommand{\sregfix}{\defn, \Gamma_0 \vdash \Delta_0}
\newcommand{\cut}[1]{\langle \form, \normalfont\textrm{cut} \rangle}
\newcommand{\mset}{\mathcal{M}}
\newcommand{\cntmod}{\mathcal{I}_\omega}
\newcommand{\sched}{(\xi_i)_i}
\newcommand{\terms}{\normalfont\textrm{Terms}(\vocab)}
\newcommand{\eqrel}{\sim}

\newcommand{\limseq}{\defn, \Gamma_\omega \vdash \Delta_\omega}
\newcommand{\schtree}{\deriv_\omega}
\newcommand{\untbranch}{\pth^{\normalfont\textrm{u}}}

\newcommand{\tvstrseq}{(\tvstrtwo_i)_{0 \leq i \leq \beta}}
\newcommand{\varset}{\mathcal{V}}
\newcommand{\subpth}{(n_{s}, \dots, n_{e})}
\newcommand{\concsubpth}{(n_{e +1}, \dots, n_{r+1})}
\newcommand{\tsubpth}{(\trace_{s}, \dots, \trace_{e})}
\newcommand{\pthjust}{\mathbf{p}}
\newcommand{\subpthcor}[1]{(n_{s_{#1}}, \dots, n_{e_{#1}})}
\newcommand{\tsubpthcor}[1]{(\trace^{#1}_{s_{#1}}, \dots, \trace^{#1}_{e_{#1}})}

\newcommand{\ih}{\normalfont\textrm{IH}}
\newcommand{\ihs}[1]{\normalfont\textrm{IH}_{#1}}
\newcommand{\ihf}[2]{\normalfont\textrm{IH}_{#1}(#2)}
\newcommand{\setexp}[2]{\{ #1 \mid #2 \}}
\newcommand{\setexpg}{\setexp{\bar{z}}{\form}}
\newcommand{\substpos}[2]{[#1 \, {\slash^+} \, #2]}
\newcommand{\substposih}[1]{\substpos{\ihs{#1}}{#1}}
\newcommand{\substposihpi}{\substposih{\Pi}}
\newcommand{\substneg}[2]{[#1 \, {\slash^-} \, #2]}
\newcommand{\substnegih}[1]{\substneg{\ihs{#1}}{#1}}
\newcommand{\substnegihpi}{\substnegih{\Pi}}

\begin{document}

\maketitle

\begin{abstract}
	Inductive definitions are an important form of knowledge in mathematics and computer science.
	Two common techniques to prove theorems about inductive definitions are the principle of mathematical induction and the principle of infinite descent.
	To formalize these principles, Brotherston and Simpson introduced the sequent calculus proof systems \lkid, for mathematical induction, and \lkidinf and \clkid, for infinite descent.
	\lkidinf is an infinitary system, in which proofs are infinite trees, and \clkid a cyclic system, in which proofs are finite graphs.
	However, these calculi restrict to monotone definitions, while inductive definitions are generally non-monotone.
	The logic \foid extends classical first-order logic with non-monotone inductive definitions.
	In earlier work, we provided a formalization of the principle of mathematical induction for non-monotone definitions by extending \lkid to a sequent calculus \scfoid for \foid.
	In this paper, we provide a formalization of the principle of infinite descent for non-monotone definitions by extending \lkidinf and \clkid to sequent calculi \scinf resp.~\sccyc for \foid.
	Furthermore, we extend several proof-theoretic results for \lkidinf and \clkid to \scinf and \sccyc regarding soundness, completeness, cut-elimination and the relation with \scfoid.
\end{abstract}

\section{Introduction} \label{sec:introduction}

\subsection{Non-Monotone Inductive Definitions}

In the field of Knowledge Representation and Reasoning (KRR), first-order logic (\fo) is commonly used as a modeling language.
However, \fo is incapable of expressing some elementary and useful notions such as the set of natural numbers or the transitive closure of a graph.
To enhance expressivity, language constructs can be added to \fo.
A particularly useful language constructs is that of inductive definitions.

Inductive definitions constitute an important form of knowledge, specifying a wide range of notions in mathematics and computer science.
A prototypical example of an inductive definition the following definition of the natural numbers, which consists of two rules:
\begin{enumerate}
	\item $0$ is a natural number.
	\item If $n$ is a natural number, then so is its successor $s(n)$.
\end{enumerate}
An inductive definition specifies how to construct its defined set through the process of \emph{iterated rule application}, which starts from the empty set and consequently applies rules until saturation.
The natural number definition is a \emph{monotone} definition, meaning that once a rule is applicable, it remains applicable.
Not all inductive definitions are monotone, as demonstrated by the following definition of the satisfaction relation $\models$ for propositional logic:
\begin{enumerate}
	\item $\str \models P$ if $P \in \str$.
	\item $\str \models \form \land \formtwo$ if $\str \models \form$ and $\str \models \formtwo$.
	\item $\str \models \lnot \form$ if not $\str \models \form$.
\end{enumerate} 
(Here we view structures $\str$ as sets of propositional symbols $P$ and assume formulae $\form$ to be composed of propositional symbols with conjunctions $\land$ and negations $\neg$.)
This definition is non-monotone because of the third rule. 
At the start of the construction process, the defined relation is empty; hence, the third rule applies to all pairs $(\str, \form)$.
It is crucial, however, not to apply this rule prematurely, for the construction process may later derive $\str \models \form$, thereby invalidating the rule.
For non-monotone definitions, %
rules are to be applied \emph{safely},
i.e., when it is certain that the condition of a rule cannot be falsified later in the process.

The logic \foid \cite{tocl/DeneckerT08} is an extension of \fo with a language construct to express non-monotone inductive definitions.
Syntactically, \foid represents inductive definitions as sets of \emph{definitional rules}.
These are expressions of the form $\defrul$, where $\form$ can be any \fo-formula.
For instance, the natural number definition can be formalized through the following \foid-definition:
\begin{equation*}
	\defin{
		\mathit{Nat}(\mathit{zero}) \rul \top \\
		\forall n: \mathit{Nat}(\mathit{succ(n)}) \rul \mathit{Nat}(n)
	}
\end{equation*}

\foid has historical roots in Logic Programming, as it grew from research to assign declarative meaning to logic programs with negation as non-monotone inductive definitions \cite{tocl/DeneckerBM01}.
Hence, it is no coincidence that the rule-based \foid-definitions are suitable to capture (fragments of) logic programs.
Furthermore, the \emph{well-founded semantics} of \foid is strongly based on the homonymous semantics for logic programs \cite{GelderRS91}.
Denecker and Vennekens \cite{KR/DeneckerV14,Denecker98} argued that the well-founded semantics captures the construction processes behind non-monotone inductive definitions and the essential principle of safe rule application.

\foid serves as a formal scientific study of inductive definitions as they occur in mathematical texts.
It provides a more general account of inductive definitions than alternative logics, which commonly impose syntactic restrictions on their definitions such as \emph{positivity} \cite{ajm/Post43,Spector61,Moschovakis74,Aczel77} \shortonly{\,} and \emph{stratification} \cite{Kreisel63,Feferman70,MartinLoef71,BuchholzFPS81}. 
Positivity forbids negation in bodies of definitions rules, yielding in particular a notion of monotone definition.
Stratification allows for negation, under the condition that the definition can be split into a hierarchy of monotone definitions, in which each defined predicate depends negatively only on defined predicates of a lower level.
While less restrictive than positivity, stratification still excludes many examples of non-monotone inductive definitions \cite{lpnmr/VandenEedeVD24,arxiv}.
\foid does not impose any such syntactic constraints on its definitions, thereby capturing a more general class of inductive definitions.

Discarding stratification allows for \emph{non-total} definitions, for which the construction process terminates in a state where rules are applicable but none safely.
Intuitively, non-total definitions correspond to nonsensical or paradoxical definitions.
A simple example is $\defintight{P \rul \lnot P}$, which defines the propositional symbol $P$ as its own negation.
Non-total definitions are anomalies, to be avoided in practice.
But they are interesting anomalies, especially from the perspective of \foid as a formal scientific study of inductive definitions.
Natural language permits nonsensical definitions, and \foid offers tools to formally study these definitions and to distinguish them from sensible, i.e., total definitions.

\subsection{Proofs with Inductive Definitions}

As pointed out by Brotherston and Simpson \cite{jlc/BrotherstonS11}, two common principles to prove theorems about inductive definitions are \emph{mathematical induction} and \emph{infinite descent}.
To formalize these principles, they introduced the \emph{sequent calculi} \lkid, for mathematical induction, and \lkidinf and \clkid, for infinite descent.
A \emph{sequent calculus} is a style of proof system developed by Gentzen \cite{Gentzen35}, well-known for its theoretical elegance and goal-directed approach to theorem proving.
In the \emph{finitary} calculus \lkid, proofs take the form of finite trees, in the \emph{infinitary} calculus \lkidinf, they take the form of possibly infinite trees, and in the \emph{cyclic} calculus \clkid, they take the form of finite graphs.
Cyclic proofs constitute an active research area with recent work in several settings, including 
verification \cite{ccs/LinkerSCB25,tcs/Ikebuchi25,arxiv/BrotherstonLDO25,ijcar/RooduijnKS24,acmpl/TsukadaU22,pldi/JonesOR22,cpldi/ItzhakyPPRS21,jar/TellezB20,cpp/RoweB17}, 
arithmetic \cite{csl/CurziM26,fscd/DasM23,fscd/Das21,lmcs/Das20,fsscs/Simpson17,lics/BerardiT17}, 
fixpoint logics \cite{ljiglp/AfshariEL24,llc/AfshariL22,aratrm/Ashari21,csl/KoriTK21} 
and foundations \cite{mscs/AfshariW24}.
The sequent calculi by Brotherston and Simpson \cite{jlc/BrotherstonS11} restrict to positive definitions.
In previous work \cite{lpnmr/VandenEedeVD24,arxiv}, we extended \lkid to a sequent calculus \scfoid for \foid, thereby obtaining a formalization of %
mathematical induction for non-monotone %
definitions.

The principle of \emph{infinite descent} has first been rigorously described by Fermat \cite{MahoneyM94}, even though its usage dates back at least to the ancient Greeks.\footnote{
	Notably, Euclid's proof for the irrationality of $\sqrt{2}$ is essentially a proof by infinite descent.
}
A proof by infinite descent is a special kind of proof by contradiction, relying on the well-foundedness of sets.
It is commonly used to prove that a property $P$ holds for all natural numbers, by assuming the existence of a number $n$ that does not satisfy $P$, and deriving the existence of a smaller number $m$ that does not satisfy $P$ either.
Repeating this reasoning leads to an infinitely descending sequence of the natural numbers, which does not exist.
Even though infinite descent has historically been used in the setting of natural numbers, the principle applies to any inductively defined set.
Brotherston and Simpson formally showed this to be the case for positive definitions \cite{jlc/BrotherstonS11} and in this paper, we formally show it for general, i.e., non-monotone definitions.

\subsection{Contributions}

In this work, we extend \lkidinf and \clkid to sequent calculi \scinf resp.~\sccyc for \foid, thus 
obtaining a formalization of 
the principle of 
infinite descent for non-monotone inductive definitions.
Besides inductive definitions, our calculi \scinf and \sccyc can %
prove theorems about logic programs under the well-founded semantics.
We establish several proof-theoretic results for \scinf and \sccyc, extending results for \lkidinf and \clkid:
\begin{itemize}
	\item \scinf and \sccyc are sound.
	\item \scinf is complete.
	\item It is decidable whether a proof in \sccyc satisfies the \emph{trace condition}, ensuring soundness.
	\item Each proof in \sccyc can be transformed to a proof in %
	\emph{cycle normal form}.
	\item For proofs in cycle normal form, finitary alternatives to the trace condition exist.
	\item In \scinf, the cut rule can be eliminated on a fragment of positive definitions, and its instances can be \emph{restricted} on a fragment of stratified definitions.
	\item \scinf is stronger than \sccyc and \sccyc is stronger than \scfoid.
\end{itemize}

\subsection{Overview}

In Section \ref{sec:foid}, we introduce the syntax and semantics of \foid.
In Section \ref{sec:infinite-proofs}, we introduce the infinitary calculus \scinf and in Section \ref{sec:cyclic-proofs} the cyclic calculus \sccyc.
In Section \ref{sec:generic-calculus}, we show that \scinf is an instance of a generic infinitary sequent calculus studied by Brotherston \cite{Brotherston2006PhD}.
As a consequence, we obtain the results on soundness, decidability, cycle normalization and the finitary soundness conditions.
In Sections \ref{sec:completeness}, \ref{sec:cut-elimination} and \ref{sec:relation-calculi}, we establish the results on completeness, cut-elimination and the relation between the three calculi, respectively.
We conclude in Section \ref{sec:conclusion} and point to avenues for further research.
Further conceptual details and proofs are included in the appendix.

\section{FO(ID)} \label{sec:foid}

Large parts of this section appear in earlier papers \cite{tocl/DeneckerBM01,tocl/DeneckerT08,KR/DeneckerV14,lpnmr/VandenEedeVD24,arxiv}.

\subsection{Syntax} \label{sec:syntax}

	The syntax of \foid extends the syntax of first-order logic (\fo), which is defined as usual. 
	A \emph{vocabulary} is a set of non-logical symbols, which are subdivided into predicate symbols and function symbols. 
	We write $\ar{\nlsym}$ to refer to the arity of a predicate or function symbol $\nlsym$, and $\nlsym/n$ to indicate that $\nlsym$ has arity $n$.
	A \emph{propositional symbol} is a $0$-ary predicate symbol, and an \emph{object symbol} is a $0$-ary function symbol.\footnote{
		We do not explicitly distinguish between variables and constants.
		Such distinction can be made implicitly by seeing free occurrences of objects symbols %
		as constants and bound occurrences as variables.
	}
	An occurrence of an object symbol $x$ in a formula $\form$ is \emph{bound} if it is in the scope of a quantified subformula ${\exists x: \formtwo}$ or ${\forall x: \formtwo}$ of $\form$, and is \emph{free} otherwise.
	The set of freely occurring object symbols in a formula $\form$ is denoted by $\free{\form}$. 
	For a set of formulae $\Gamma$, we write $\free{\Gamma}$ for $\cup_{\form \in \Gamma} \free{\form}$. 
	Given a term $t$, an object symbol $x$ and a formula $\form$, we denote by $\form[t/x]$ the formula obtained from $\form$ by replacing all free occurrences of $x$ in $\form$ by $t$. 
	The substitution is \emph{safe} if $t$ contains no object symbol $y$ that gets bound in $\form[t/x]$ due to the existence of a free occurrence of $x$ in the scope of a quantification $\forall y$ or $\exists y$ within $\form$. 
	We will silently assume that such quantifications in $\form$ are renamed prior to the substitution, so that substitution is always safe. 
	Given a set of formulae $\Gamma$, we write $\Gamma[t/x]$ for the set $\cup_{\form \in \Gamma} \form[t/x]$.
	An occurrence of a subformula $\formtwo$ in a formula $\form$ is said to be \emph{positive} if it is in the scope of an even number of negation symbols $\lnot$ (after unwinding all %
	material implications $\Rightarrow$ in terms of $\land$, $\lor$ and $\lnot$), and is said to be \emph{negative} otherwise.
	We use the symbol $\doteq$ for syntactic equality.

\begin{definition} \label{def:definitions}
		A \emph{definitional rule} (or \emph{rule} for short) is an expression of the form $\forall \bar{x} : P(\bar{t}) \rul \form$, 
		where $\bar{x}$ is a tuple of object symbols, $P$ a predicate symbol, $\bar{t}$ an $\ar{P}$-tuple of terms, and $\form$ an \fo-formula. 
		The atom $P(\bar{t})$ is called the \emph{head} and the formula $\form$ the \emph{body} of the definitional rule. 
		An \emph{(\foid-)definition} is a finite set $\defn$ of definitional rules.
		A predicate symbol $P$ that occurs in the head of a rule of a definition $\defn$ is a \emph{defined predicate} of $\defn$.
		A \emph{defined atom} of $\defn$ is an atom $P(\bar{t})$ where $P$ is a defined predicate of $\defn$.
		All non-logical symbols in $\defn$ that are not defined predicates are called \emph{parameters} of $\defn$. 
		The set of defined predicates of $\defn$ is denoted by $\defp{\defn}$, the set of parameters of $\defn$ by $\pars{\defn}$, and their union by $\sym{\defn}$. 
		We extend the notions of free and bound object symbols to %
		rules and definitions.
\end{definition}

Other frameworks of inductive definitions often restrict to \emph{positive} or \emph{stratified} definitions.

\begin{definition}
	Let $\defn$ be a definition.
	We say that $\defn$ is \emph{positive} if the body $\form$ of every definitional rule $\defrul$ in $\defn$ is composed of atoms by conjunctions and disjunctions (and hence, contains no negations or material implications).
	A \emph{stratification} of $\defn$ is a function $\ell : \defp{\defn} \to \nat$ such that for all $P, Q \in \defp{\defn}$ and every %
	rule $\defrul$ in $\defn$:
	\begin{itemize}
		\item if $Q$ occurs in $\form$, then $\ell(Q) \leq \ell(P)$;
		\item if $Q$ occurs negatively in $\form$, then $\ell(Q) < \ell(P)$.
	\end{itemize}
	We say that $\defn$ is \emph{stratified} if there exists a stratification of $\defn$.
\end{definition}

\begin{example} \label{ex:def_even}
The set of even numbers can be defined by the following \foid-definition of a predicate $\mathit{Even}/1$, in terms of the constant $\mathit{zero}$, the function $\mathit{succ}/1$ and the natural number predicate $\mathit{Nat}/1$:
\[ 
\defin{
	\mathit{Even}(\mathit{zero}) \rul \top \\
	\forall n : \mathit{Even}(\mathit{succ}(n)) \rul \mathit{Nat}(n) \land \lnot \mathit{Even}(n)
} 
\]
This definition admits no stratification $\ell$,
as the second rule would impose $\ell(\mathit{Even}) < \ell(\mathit{Even})$.
\end{example}

	The \emph{definitional implication} $\rul$ should not be confused with material implication $\Rightarrow$.
	An important difference is that definitional rules are not assigned truth values in structures. 
	Definitions, on the other hand, \emph{are} assigned truth values in structures,
	as we will see in Section \ref{sec:well-founded-semantics}. 
	Intuitively, a definition $\defn$ is satisfied in a structure $\str$ if $\str$ interprets the defined predicates of $\defn$ as specified by the rules of $\defn$. 

\foid-formulae are defined as \fo-formulae, but with the addition of the following rule:
\begin{itemize}
	\item $\defn$ is an \foid-formula if $\defn$ is an \foid-definition.
\end{itemize}

\subsection{Well-Founded Semantics} \label{sec:well-founded-semantics}

The well-founded semantics was originally developed for logic programs with negation \cite{GelderRS91}.
Denecker and Vennekens \cite{KR/DeneckerV14,Denecker98} argued that the well-founded semantics captures the semantics of non-monotone inductive definitions by formalizing the \emph{induction} or \emph{construction processes} behind inductive definitions.
We restrict to an intuitive discussion of the well-founded semantics in this section, and refer to Appendix \ref{app:well-founded-semantics} for a more detailed exposition.

The key notion in formalizing the construction process behind inductive definitions is that of \emph{well-founded induction}.
Let $\defn$ be a definition and $\cont$ a $\pars{\defn}$-structure, which we call a \emph{$\defn$-context}.
A \emph{well-founded induction} of $\defn$ in $\cont$ is a (possibly transfinite) sequence $\wfind$ of \emph{three-valued structures}.
Three-valued structures assign a \emph{truth value} to each atom, which is either true ($\true$), false ($\false$) or unknown ($\unknown$).
The three-valued structures $\tvstr_i$ have vocabulary $\sym{\defn}$ and expand $\cont$ (when viewed as a three-valued structure).
The first structure $\tvstr_0$ maps every defined atom to unknown, reflecting that at the start of the construction process, nothing is known yet about the defined predicates.
For every $i$, $\tvstr_{i+1}$ \emph{refines} $\tvstr_i$ by setting some unknown atoms to true or false, given that this can be derived safely from the rules of $\defn$.
Once no more refinements apply, the \emph{well-founded model} $\tvstr_\beta$ of $\defn$ in $\cont$ is reached.
The \emph{well-founded satisfaction relation} is a relation between structures $\str$ and $\foid$-formulae $\form$.
It is defined as the satisfaction relation $\models$ for \fo, but with addition of the following rule:
\begin{itemize}
	\item $\str \modelswf \defn$ if $\str|_{\sym{\defn}}$ is the well-founded model of $\defn$ in $\str|_{\pars{\defn}}$.
\end{itemize}
An \foid-formula $\form$ is said to be \emph{valid} (under the well-founded semantics) if $\str \modelswf \form$ for all structures $\str$ that interpret $\form$, and \emph{invalid} otherwise.

We say that a definition $\defn$ is \emph{(non-)total} in a $\defn$-context $\cont$ if the well-founded model $\wfmod$ of $\defn$ in $\cont$ is (not) two-valued, i.e., if (not) every defined atom of $\defn$ is interpreted as $\true$ or $\false$ by $\wfmod$.
Intuitively, $\defn$ is non-total in $\cont$ if $\defn$ is a nonsensical definition given $\cont$.

\section{An Infinitary Sequent Calculus for Non-Monotone Definitions} \label{sec:infinite-proofs}

In a sequent calculus, proofs take the form of trees or graphs, in which the nodes are labeled with \emph{sequents}.
We define an (\foid-)\emph{sequent} %
to be an expression $\sregseq$, where $\defn$ is an \foid-definition and $\Gamma$ and $\Delta$ are finite sets of \fo-formulae.\footnote{
	In previous work \cite{lpnmr/VandenEedeVD24,arxiv}, we defined an \foid-sequent as an expression of the form $\seq$ such that $\Gamma$ and $\Delta$ are sets of \foid-formulae, which is more general than the notion in this paper.
	In particular, an \foid-sequent in \cite{lpnmr/VandenEedeVD24,arxiv} may contain multiple definitions, definitions may occur on both sides of the sequent, and they may occur as components of complex formulae.
	The sequent calculi presented in this paper support a more general notion of \foid-sequent than we consider.
	For instance, we may allow for sequents of the form $\regseq$ with multiple definitions $\defns$, given that they can be merged into a single definition $\defn = \defn_1 \cup \dots \cup \defn_n$ in a model-preserving way.
	Denecker and Ternovska \cite{tocl/DeneckerT08} investigated under which conditions this is the case.
	For simplicity, we restrict to sequents $\sregseq$ with a single definition $\defn$.
	It is currently unknown how general of a notion of \foid-sequent our calculi support.
}
Semantically, an \foid-sequent $\sregseq$ says that whenever $\defn$ holds (under the well-founded semantics) and all formulae in $\Gamma$ hold, then at least one formula in $\Delta$ holds.
We extend the well-founded satisfaction relation $\modelswf$ to \foid-sequents by letting $\str \modelswf \, \sregseq$ if $\str \modelswf \defn$, $\str \models \bigwedge \Gamma$ and $\str \not\models \bigvee \Delta$.
An \emph{\fo-sequent} $\seq$, in which $\Gamma$ and $\Delta$ are finite sets of \fo-formulae, corresponds to an \foid-sequent $\sregseq$ in which $\defn$ is empty.

A sequent calculus contains \emph{inference rules}, deriving a \emph{conclusion} from one or more \emph{premises}, all of which are sequents.
Figure \ref{fig:scfo} presents the inference rules of the sequent calculus \scfo for \fo. %
Some rules are only applicable under certain conditions, which are written next to them.
The \emph{logical rules} in \scfo consist of \emph{left} and \emph{right introduction rules} for the connectives $\lnot$, $\lor$, $\land$, $\Rightarrow$, the quantifiers $\forall$ and $\exists$, and for equality $=$, introducing them in the left resp.~right side of the conclusion.
In the introduction rules for connectives and quantifiers, we refer to the formula in the conclusion that is not in $\Gamma$ or $\Delta$ as the \emph{active} formula, and to the formulae in the premises that are not in $\Gamma$ or $\Delta$ as the \emph{auxiliary} formulae.

\begin{figure}[t]
	\centering
	
	\noindent \textbf{Structural rules}
	
	\vsp
	
	\begin{tabular}{l r}
		$
		\begin{prooftree}[center=false]
			\hypo{ }
			\infer1[$\Gamma \cap \Delta \neq \emptyset$ or $\bot \in \Gamma$ or $\top \in \Delta$ (ax)]{\Gamma \vdash \Delta}
		\end{prooftree}
		$ & $
		\begin{prooftree}[center=false]
			\hypo{\Gamma' \vdash \Delta'}
			\infer1[$\Gamma' \subseteq \Gamma$, $\Delta' \subseteq \Delta$ (wk)]{\Gamma \vdash \Delta}
		\end{prooftree}
		$
		\\
		$
		\begin{prooftree}[center=false]
			\hypo{\Gamma \vdash \Delta}
			\infer1[(subst)]{\Gamma[t / x] \vdash \Delta[t/x]} %
		\end{prooftree}
		$ & $
		\begin{prooftree}[center=false]
			\hypo{\Gamma \vdash \form, \Delta}
			\hypo{\Gamma, \form \vdash \Delta}
			\infer2[(cut)]{\Gamma \vdash \Delta} %
		\end{prooftree}
		$
	\end{tabular}
	
	\vsp
	
	\noindent \textbf{Logical rules}
	
	\vsp
	
	\begin{tabular}{l r}
		$
		\begin{prooftree}[center=false]
			\hypo{\Gamma \vdash \form, \Delta}
			\infer1[($\lnot$L)]{\Gamma, \lnot \form \vdash \Delta}
		\end{prooftree}
		$ & $ 
		\begin{prooftree}[center=false]
			\hypo{\Gamma, \form \vdash \Delta}
			\infer1[($\lnot$R)]{\Gamma \vdash \lnot \form, \Delta}
		\end{prooftree}
		$
		\\
		$
		\begin{prooftree}[center=false]
			\hypo{\Gamma, \form \vdash \Delta}
			\hypo{\Gamma, \formtwo \vdash \Delta}
			\infer2[($\lor$L)]{\Gamma, \form \lor \formtwo \vdash \Delta}
		\end{prooftree}
		$ & $ 
		\begin{prooftree}[center=false]
			\hypo{\Gamma \vdash \form, \formtwo, \Delta}
			\infer1[($\lor$R)]{\Gamma \vdash \form \lor \formtwo, \Delta}
		\end{prooftree}
		$
		\\
		$
		\begin{prooftree}[center=false]
			\hypo{\Gamma, \form, \formtwo \vdash \Delta}
			\infer1[($\land$L)]{\Gamma, \form \land \formtwo \vdash \Delta}
		\end{prooftree}
		$ & $ 
		\begin{prooftree}[center=false]
			\hypo{\Gamma \vdash \form, \Delta}
			\hypo{\Gamma \vdash \formtwo, \Delta}
			\infer2[($\land$R)]{\Gamma \vdash \form \land \formtwo, \Delta}
		\end{prooftree}
		$
		\\
		$
		\begin{prooftree}[center=false]
			\hypo{\Gamma \vdash \form, \Delta}
			\hypo{\Gamma, \formtwo \vdash \Delta}
			\infer2[($\Rightarrow$L)]{\Gamma, \form \Rightarrow \formtwo \vdash \Delta}
		\end{prooftree}
		$ & $ 
		\begin{prooftree}[center=false]
			\hypo{\Gamma, \form \vdash \formtwo, \Delta}
			\infer1[($\Rightarrow$R)]{\Gamma \vdash \form \Rightarrow \formtwo, \Delta}
		\end{prooftree}
		$
		\\
		$
		\begin{prooftree}[center=false]
			\hypo{\Gamma, \form[t/x] \vdash \Delta}
			\infer1[($\forall$L)]{\Gamma, {\forall x : \form} \vdash \Delta}
		\end{prooftree}
		$ & $ 
		\begin{prooftree}[center=false]
			\hypo{\Gamma \vdash \form, \Delta}
			\infer1[$x \not\in \free{\Gamma \cup \Delta}$ ($\forall$R)]{\Gamma \vdash {\forall x: \form}, \Delta}
		\end{prooftree}
		$
		\\
		$
		\begin{prooftree}[center=false]
			\hypo{\Gamma, \form \vdash \Delta}
			\infer1[$x \not\in \free{\Gamma \cup \Delta}$ ($\exists$L)]{\Gamma, {\exists x:  \form} \vdash \Delta}
		\end{prooftree}
		$ & $ 
		\begin{prooftree}[center=false]
			\hypo{\Gamma \vdash \form[t/x], \Delta}
			\infer1[($\exists$R)]{\Gamma \vdash {\exists x: \form}, \Delta}
		\end{prooftree}
		$
		\\
		$
		\begin{prooftree}[center=false]
			\hypo{\Gamma[s/x, t/y] \vdash \Delta[s/x, t/y]}
			\infer1[($=$L)]{\Gamma[t/x, s/y], {t=s} \vdash \Delta[t/x, s/y]}
		\end{prooftree}
		$ & $ 
		\begin{prooftree}[center=false]
			\hypo{\phantom{\Gamma \vdash \Delta}}
			\infer1[($=$R)]{\Gamma \vdash{ t=t}, \Delta}
		\end{prooftree}
		$
	\end{tabular}

	\caption{
		Inference rules of \scfo, based on \cite{jlc/BrotherstonS11}.
	 	\scfo can be seen as an extension of Gentzen's sequent calculus \lk for classical first-order logic \cite{Gentzen35} with introduction rules for equality.
	}
	\label{fig:scfo}
\end{figure}

The inference rules of the infinitary sequent calculus \scinf extend the inference rules of \scfo with a left and a right introduction rule for defined atoms.
The \textbf{right introduction rule for defined atoms} takes the following form:
\begin{equation*} 
	\begin{prooftree}
		\hypo{\defn, \Gamma \vdash \form[\bar{s}/\bar{x}], \Delta}
		\infer1[(def R)]{\defn, \Gamma \vdash P(\bar{t}[\bar{s}/\bar{x}]), \Delta}
	\end{prooftree}
\end{equation*}
Here $\defn$ is a definition with a rule $\forall \bar{x}: P(\bar{t}) \rul \form$ and $\bar{s}$ is a tuple of object symbols with the same length as $\bar{x}$.
Intuitively, (def R) allows deriving heads of definitional rules from bodies.

The \textbf{left introduction rule for defined atoms} in \scinf is a case distinction rule:\footnote{
	We name the rule (case) instead of (def L) to distinguish it from the induction rule (ind) in \scfoid, which serves as the left introduction rule for defined atoms in \scfoid.
}
\begin{equation*}
	\begin{prooftree}
		\hypo{\text{case distinctions}}
		\infer1[(case)]{\defn, \Gamma, P(\bar{v}) \vdash \Delta}
	\end{prooftree}
\end{equation*}
Here $P(\bar{v})$ is a defined atom of $\defn$.
The rule has a \emph{case distinction} for every rule $\defrul$ in $\defn$ defining $P$.
This is a premise $\defn, \Gamma, {\bar{v} = \bar{t}[\bar{y}/\bar{x}]}, \form[\bar{y}/\bar{x}] \vdash \Delta$, where $\bar{y}$ is a tuple of object symbols with the same length as $\bar{x}$, none of which occurs freely in $\defn$, $\Gamma$, $P(\bar{v})$ or $\Delta$.

\begin{definition} \label{def:inf-derivation}  %
		Let $(N, E, r, \seqfun)$ be a labeled, rooted, directed tree, where $N$ is the set of \emph{nodes}, $E \subseteq N \times N$ the set of \emph{edges}, $r \in N$ the \emph{root}, and $\seqfun: N \to \seqs$ the \emph{labeling function}.
		Here $\seqs$ denotes the set of \foid-sequents. %
		Let $\seqshort = \seqfun(r)$, and let $\rulfun$ be a partial function sending nodes to inference rules in \scinf.
		Suppose that for every $n \in N$:
		\begin{itemize}
			\item if $\rulfun(n)$ is defined, then $\seqfun(n)$ is the conclusion and $\{ \seqfun(m) \mid (n,m) \in E \}$ the set of premises of an instance of the rule $\rulfun(n)$;
			\item if $\rulfun(n)$ is undefined, then $n$ has no children, %
			and we call $n$ a \emph{bud}. 
		\end{itemize}
		Then we call $\deriv = (N, E, r, \seqfun, \rulfun)$ an \emph{\scinf-derivation} of $\seqshort$.
		If $\deriv$ has no buds, we call it an \emph{\scinf-pre-proof} of $\seqshort$.
\end{definition}

Note that Definition \ref{def:inf-derivation} does not impose finiteness on derivations. 
Hence, an \scinf-derivation may have \emph{infinite branches}.
As illustrated in Example \ref{ex:bad-pre-proof}, not all infinite branches reflect a sound reasoning.
Soundness is guaranteed by the \emph{trace condition}, which intuitively says that along every infinite branch, a definition is ``unfolded'' infinitely often.

\begin{definition} \label{def:trace}
	Let $\deriv = (N, E, r, \seqfun, \rulfun)$ be an \scinf-derivation %
	and $\pth = (n_i)_{i}$ a path in $\deriv$.
	A \emph{trace} along $\pth$ is a sequence $(\trace_i)_{i}$ of \fo-formulae
	such that for all $i$:
	\begin{itemize}
		\item if $\seqfun(n_i) = \Gamma_i \vdash \Delta_i$, then $\trace(n_i) \in \Gamma_i$ or $\trace(n_i) \in \Delta_i$;
		\item if $\rulfun(n_i) = \textup{(subst)}$, then $\trace_i = \trace_{i+1}[t/x]$, where $[t/x]$ is the substitution associated with the instance of \textup{(subst)} applied at $n_i$;
		\item if $\rulfun(n_i) = \textup{($=$L)}$ and ${t=s}$ is the equality associated with the instance of \textup{($=$L)}, then there exists an  \fo-formula $\form$ and object symbols $x$ and $y$ such that $\trace_i = \form[t/x, s/y]$ and $\trace_{i+1} = \form[s/x, t/y]$;
		\item if $\rulfun(n_i) \notin \{\textup{(subst)}, \textup{($=$L)}\}$ and $\trace_i$ is the active formula of the instance of $\rulfun(n_i)$, then $\trace_{i+1}$ is an auxiliary formula of this instance;
		if furthermore $\rulfun(n_i) = \textup{(case)}$, %
		then $i$ is said to be a \emph{progression point}; %
		\item if $\rulfun(n_i) \notin \{\textup{(subst)}, \textup{($=$L)}\}$ and $\trace_i$ is not the active formula of the instance, then $\trace_{i} = \trace_{i+1}$.
	\end{itemize}
	
	An \emph{(infinite) branch} in $\deriv$ is an (infinite) path in $\deriv$ starting at the root $r$ of $\deriv$.
	We say that $\deriv$ satisfies the \emph{trace condition} if for every infinite branch $\pth$ in $\deriv$, %
	there exists an \emph{infinitely progressing trace}, i.e., a trace with infinitely many progression points, along some tail of $\pth$.
	An \emph{\scinf-proof} is an \scinf-pre-proof that satisfies the trace condition.
	If there exists an \scinf-proof of an \foid-sequent $\seqshort$, we call $\seqshort$ an \emph{\scinf-theorem}.
\end{definition}

The trace condition entails soundness w.r.t.~the well-founded semantics by appealing to the well-foundedness of the ordinals, through an argument by contradiction.
On a high level, the reasoning goes as follows.
Suppose that $\deriv$ is an \scinf-proof of a sequent $\sregseq$, but that $\sregseq$ is invalid.
Then we can construct an infinite path $\pth= (n_i)_{i}$ in $\deriv$, a structure $\str$ such that $\str$ is a countermodel of $\seqfun(n_i)$ %
for all $i$,\footnote{
	A structure $\str$ is a \emph{countermodel} for an \foid-sequent $\sregseq$ if $\str \modelswf \defn$, $\str \models \bigwedge \Gamma$ and $\str \not\models \bigvee \Delta$.
}  
and a sequence of three-valued structures $\tvstrseq$ such that for all $i$,
$\tvstrtwo_i$ expands $\tvstr_i$ for a well-founded induction $\wfind$ of $\defn$ in $\str|_{\pars{\defn}}$.
Furthermore, we can construct $\pth$, $I$ and $\tvstrseq$ in such way that if $(\trace_i)_i$ is a trace along $\pth$, then the sequence of ordinals $(\alpha_i)_i$ defined by 
\begin{equation*}
	\alpha_i \text{ is the least } \alpha \text{ such that } \trace_i^{\tvstrtwo_\alpha} \in \{\true, \false\}
\end{equation*}
has the property that for all $i$: $\alpha_{i+1} \leq \alpha_{i}$ and $\alpha_{i+1} < \alpha_{i}$ if $i$ is a progression point of $(\trace_i)_i$.
Since $\deriv$ satisfies the trace condition, there exists an infinitely progressing trace $(\trace_i)_i$ along $\pth$.
This implies the existence of an infinite sequence of ordinals $(\alpha_i)_i$ such that $\alpha_{i+1} \leq \alpha_{i}$ for all $i$ and $\alpha_{i+1} < \alpha_{i}$ for infinitely many $i$.
However, this contradicts the well-foundedness of the ordinals.

In Section \ref{sec:generic-calculus}, we will rigorously establish the soundness of \scinf.

\begin{example}
	Let $\defn$ be the definition of $\mathit{Even}$ from Example \ref{ex:def_even}, and let $\Gamma$ denote the Peano axioms $\forall x: \forall y: {\mathit{succ}(x) = \mathit{succ}(y)} \Rightarrow {x=y}$ and $\forall x: {\mathit{succ}(x) \neq \mathit{zero}}$.\footnote{
		We use $t \neq s$ as shorthand notation for $\lnot {(t=s)}$.
	}
	Consider the sequent $\defn, \Gamma \vdash \lnot \mathit{Even}(\mathit{succ}(\mathit{zero}))$, which says that, given $\defn$ and $\Gamma$, the number one is not even.
	Below, we display an \scinf-derivation $\deriv$ of this sequent, rendered in a typical proof format, using horizontal bars for rule applications.
	For space reasons, we denote $\mathit{zero}$ by $0$, abbreviate every other non-logical symbol to its first letter, and split $\deriv$ into two parts.
	\begin{equation*} \small
		\begin{prooftree}
			\infer0[(ax)]{\defn, s(0) = s(n) \vdash s(0) = s(n), E(n)}
			
			\infer0[(ax)]{\defn \vdash \top}
			\infer1[(def R)]{\defn \vdash E(0)}
			\infer1[($=$L)]{\defn, {0=n} \vdash E(n)}
			\infer1[(wk)]{\defn, {0=n}, {s(0) = s(n)} \vdash E(n)}
			
			\infer2[($\Rightarrow$L)]{\defn, {s(0) = s(n)} \Rightarrow {0=n}, {s(0) = s(n)} \vdash E(n)}
		\end{prooftree}
	\end{equation*}

	\begin{equation*} \small
	\begin{prooftree}%
		\infer0[(ax)]{{s(0) = 0} \vdash {s(0) = 0} }
		\infer1[($\lnot$L)]{{s(0) \neq 0}, {s(0) = 0} \vdash}
		\infer1[($\forall$L)]{\forall x: {s(x) \neq 0}, {s(0) = 0} \vdash}
		\infer1[(wk)]{\defn, \Gamma, {s(0) = 0}, \top \vdash}

		\hypo
		{\defn, {s(0) = s(n)} \Rightarrow {0=n}, {s(0) = s(n)} \vdash E(n)}
		\infer1[($\forall$L)]{\defn, \forall y: {s(0) = s(y)} \Rightarrow {0=y}, {s(0) = s(n)} \vdash E(n)}
		\infer1[($\forall$L)]{\defn,\forall x: \forall y: {s(x) = s(y)} \Rightarrow {x=y}, {s(0) = s(n)} \vdash E(n)}
		\infer1[(wk)]{\defn, \Gamma,{s(0) = s(n)}, N(n) \vdash E(n)}
		\infer1[($\lnot$L)]{\defn, \Gamma,{s(0) = s(n)}, N(n), \lnot E(n) \vdash}
		\infer1[($\land$L)]{\defn, \Gamma,{s(0) = s(n)}, N(n) \land \lnot E(n) \vdash}
		
		\infer2[(case)]{\defn, \Gamma, E(s(0)) \vdash}
		\infer1[($\lnot$R)]{\defn, \Gamma \vdash \lnot E(s(0))}
	\end{prooftree}
	\end{equation*}
	Since $\deriv$ has no buds, it is an \scinf-pre-proof.
	Since it has no infinite branches, $\deriv$ trivially satisfies the trace condition and hence, it is an \scinf-proof of $\defn, \Gamma \vdash \lnot \mathit{Even}(\mathit{succ}(\mathit{zero}))$.
	
	Reading the proof bottom-up, it corresponds to the following informal reasoning. 
	If one were an even number, it would have to be equal to zero or to the successor of a non-even natural number $n$.
	Since one is not zero, the latter case must hold.
	As one is only the successor of zero, zero must be a non-even natural number.
	However, zero is even.
\end{example}

\begin{example} \label{ex:bad-pre-proof}
Let $\defn \doteq \defintight{P \rul P}$.
The well-founded model of $\defn$ (in any $\defn$-context) interprets $P$ as false.
Consider the following \scinf-pre-proof:
\begin{equation*}
	\begin{prooftree}
		\hypo{\vdots}
		\infer[no rule]1[]{\defn \vdash P}
		\infer1[(def R)]{\defn \vdash P}
		\infer1[(def R)]{\defn \vdash P}
	\end{prooftree}
\end{equation*}
Since it has no applications of (case), this \scinf-pre-proof does not satisfy the trace condition.
Therefore, it is not an \scinf-proof.
\end{example}

\begin{example} \label{ex:inf-proof-nat-even-odd}
	Let $\defn_N$ be the definition of $\mathit{Nat}$ from Section \ref{sec:introduction}, and consider the following (mutual) definition of $\mathit{Even}$ and $\mathit{Odd}$:
	\begin{equation*}
		\defn_{E, O} \doteq \defin{
			\mathit{Even}(\mathit{zero}) \rul \top\\
			\forall n: \mathit{Even}(\mathit{succ}(n)) \rul \mathit{Odd}(n)\\
			\forall n: \mathit{Odd}(\mathit{succ}(n)) \rul \mathit{Even}(n)
		}
	\end{equation*}
	Let $\defn = \defn_N \cup \defn_{E, O}$.
	Consider the sequent $\defn, \mathit{Nat}(n) \vdash \mathit{Even}(n), \mathit{Odd}(n)$, which says that, under $\defn$, any natural number is even or odd (as $n$ can be interpreted arbitrarily).
	Below, we display a segment of an infinite \scinf-pre-proof $\deriv$ of this sequent.
	As before, we denote $\mathit{zero}$ by $0$ and abbreviate every other non-logical symbol to its first letter.\footnote{
		This example comes from \cite{jlc/BrotherstonS11}.
	}
	\begin{equation*}
		\begin{prooftree}
						\infer0[(ax)]{\Phi \vdash \top}
						\infer1[(def R)]{\Phi \vdash E(0)}
						\infer1[($=$L)]{\Phi, n=0 \vdash E(n)}
						\infer1[(wk)]{\Phi, n=0, \top \vdash E(n), O(n)}
						\hypo{\vdots \hspace{6mm}}
						\infer[no rule]1[]{\Phi, \underline{N(m)} \vdash O(m), E(m)}
						\infer1[(def R)]{\Phi, \underline{N(m)} \vdash O(m), O(s(m))}
						\infer1[(def R)]{\Phi, \underline{N(m)} \vdash E(s(m)), O(s(m))}
						\infer1[($=$L)]{\Phi, n=s(m), \underline{N(m)} \vdash E(n), O(n)}
						\infer2[(case)]{\Phi, \underline{N(n)} \vdash E(n), O(n)}
					\end{prooftree}
	\end{equation*}
	We underlined a trace along %
	the single infinite branch in $\deriv$.
	Since this trace has infinitely many progression points, %
	$\deriv$ satisfies the trace condition, and hence, it 
	is an \scinf-proof.
	
	The above \scinf-proof corresponds to the following informal proof by infinite descent.
	Suppose that $n$ is a natural number that is neither even nor odd.
	Since $n$ is a natural number, it is either zero or the successor of a smaller natural number $m$.
	As zero is even, $n$ must be the successor of a natural number $m$.
	Note that $m$ is neither even nor odd, since otherwise its successor $n$ would be even or odd.
	Repeating this reasoning, we obtain an infinite sequence of decreasing natural numbers, which does not exist.
\end{example}

\begin{example} \label{ex:inf-proof-even-implies-double}
	Consider the definition
	\begin{equation*}
		\defn_D \doteq \defin{
			\mathit{Double}(\mathit{zero}) \rul \top \\
			\forall n: \mathit{Double}(\mathit{succ}(\mathit{succ}(n))) \rul \mathit{Double}(n)
		} 
	\end{equation*}
	and note that it is an alternative, monotone definition of the even numbers.
	Let $\defn_N$ be the definition of $\mathit{Nat}$ from Section \ref{sec:introduction} and $\defn_E$ the definition of $\mathit{Even}$ from Example \ref{ex:def_even}, and let %
	$\defn = \defn_N \cup \defn_E \cup \defn_D$.
	Below, we display a segment of an infinite \scinf-pre-proof $\deriv$ of the sequent %
	$\defn, \mathit{Even}(n) \vdash \mathit{Double}(n)$.
	As before, we denote $\mathit{zero}$ by $0$, abbreviate every other non-logical symbol to its first letter, and split $\deriv$ into two parts.
	\begin{equation*}\small
		\begin{prooftree}
			\infer0[(ax)]{\defn \vdash \top}
			\infer1[(def R)]{\defn \vdash E(0)}
			\infer1[($=$L)]{\defn, m=0 \vdash E(m)}
			\infer1[(wk)]{\defn, m=0, \top \vdash E(m), D(s(m))}
			
			\infer0[(ax)]{\defn, N(p) \vdash N(p), D(p)}
			\hypo{\hspace{3mm} $\vdots$}
			\infer[no rule]1[]{\defn, \underline{E(p)} \vdash D(p)}
			\infer1[(wk)]{\defn, N(p), \underline{E(p)} \vdash D(p)}
			\infer1[($\lnot$R)]{\defn, N(p) \vdash \underline{\lnot E(p)}, D(p)}
			\infer2[($\land$R)]{\defn, N(p) \vdash \underline{N(p) \land \lnot E(p)}, D(p)}
			\infer1[(def R)]{\defn, N(p) \vdash \underline{E(s(p))}, D(p)}
			\infer1[(def R)]{\defn, N(p) \vdash \underline{E(s(p))}, D(s(s(p)))}
			\infer1[($=$L)]{\defn, m=s(p), N(p) \vdash \underline{E(m)}, D(s(m))}
			\infer2[(case)]{\defn, N(m) \vdash \underline{E(m)}, D(s(m))}
		\end{prooftree}
	\end{equation*}

	\begin{equation*}\small
	\begin{prooftree}
		\infer0[(ax)]{\defn, \top \vdash \top}
		\infer1[(def R)]{\defn, \top \vdash D(0)}
		\infer1[($=$L)]{\defn, n=0, \top \vdash D(n)}
		\hypo{\defn, N(m) \vdash \underline{E(m)}, D(s(m))}
		\infer1[($\lnot$L)]{\defn, N(m), \underline{\lnot E(m)} \vdash D(s(m))}
		\infer1[($\land$L)]{\defn, \underline{N(m) \land \lnot E(m)} \vdash D(s(m))}
		\infer1[($=$L)]{\defn, n = s(m), \underline{N(m) \land \lnot E(m)} \vdash D(n)}
		\infer2[(case)]{\defn, \underline{E(n)} \vdash D(n)} %
	\end{prooftree}
	\end{equation*}
	We underlined a trace along %
	the single infinite branch in $\deriv$.
	Since this trace has infinitely many progression points, %
	$\deriv$ satisfies the trace condition, and hence, it 
	is an \scinf-proof.
\end{example}

\section{A Cyclic Sequent Calculus for Non-Monotone Definitions} \label{sec:cyclic-proofs}

Many infinite proofs can be represented in a finite way.
Our cyclic calculus \sccyc can be seen as the restriction of \scinf to \emph{regular trees}, which are trees with only finitely many distinct subtrees.
It is well-known that regular trees are precisely the tree unravelings of finite graphs \cite{lp/Colmerauer82,panjoc/Mauborgne00,acm/TurbakW01}.

\begin{definition}
	An \emph{\sccyc-derivation} $\deriv = (N, E, r, \seqfun, \rulfun)$ of an \foid-sequent $\seqshort$ is defined as an \scinf-derivation of $\seqshort$, with the difference that $(N, E)$ is a finite graph instead of a (potentially infinite) tree.
	A node $n \in N$ is called a \emph{companion} of a node $m \in N$ in $\deriv$ if %
	$\seqfun(n) = \seqfun(m)$.
	A \emph{repeat function} $\repfun$ for $\deriv$ maps every bud of $\deriv$ to a companion.
	An \emph{\sccyc-pre-proof} $\prf$ of $\seqshort$ is a pair $(\deriv, \repfun)$ of an \sccyc-derivation $\deriv$ of $\seqshort$ and a repeat function $\repfun$ for $\deriv$.
	The \emph{graph} $\graph{\prf}$ of an \sccyc-pre-proof $\prf$ is obtained by identifying every bud $n$ with its companion $\repfun(n)$.
	An \emph{\sccyc-proof} of $\seqshort$ is an \sccyc-pre-proof of $\seqshort$ for which $\graph{\prf}$ satisfies the trace condition.
	If there exists an \sccyc-proof of an \foid-sequent $\seqshort$, we call $\seqshort$ an \emph{\sccyc-theorem}.
\end{definition}

\begin{example}
	The \scinf-proof $\deriv$ from Example \ref{ex:inf-proof-nat-even-odd} can easily be transformed into an \sccyc-proof, by replacing the top node of the displayed segment with the derivation
	\begin{equation*}
		\begin{prooftree}
			\hypo{\Phi, \underline{N(n)} \vdash O(n), E(n)}
			\infer1[(subst)]{\Phi, \underline{N(m)} \vdash O(m), E(m)}
		\end{prooftree}
	\end{equation*}
	and letting $\repfun$ map the upper node to its (unique) companion in $\deriv$.
	In a similar way, the \scinf-proof $\deriv$ from Example \ref{ex:inf-proof-even-implies-double} can be transformed into an \sccyc-proof.
\end{example}

\scinf and \sccyc allow proving non-totality of definitions $\defn$ by deriving sequents of the form $\defn \vdash$.
Indeed, since $\defn \vdash$ has an empty right side, it is semantically equivalent to $\lnot \defn$.
Validity of $\lnot \defn$ means that $\defn$ has no two-valued models; 
in other words, the well-founded model of $\defn$ in any $\defn$-context must be strictly three-valued.

\begin{example} \label{ex:proof-liar}
	Let $\defn \doteq \defintight{P \rul \lnot P}$.
	Consider the following \sccyc-pre-proof $\prf$ of the sequent $\defn \vdash$, where we represent the repeat function $\repfun$ by the annotations $(*)$ and $(\dagger)$:
	\begin{equation*}
		\begin{prooftree}
			\hypo{\defn, \underline{P} \vdash \comp}
			\infer1[($\lnot$R)]{\defn \vdash \underline{\lnot P}}
			\infer1[(def R)]{\defn \vdash  \underline{P} \comptwo}
			\hypo{\defn \vdash \underline{P} \comptwo}
			\infer1[($\lnot$L)]{\defn,  \underline{\lnot P} \vdash}
			\infer1[(case)]{\defn,  \underline{P} \vdash \comp}
			\infer2[(cut)]{\defn \vdash}
		\end{prooftree}
	\end{equation*}
	The graph $\graph{\prf}$ of $\prf$ has two infinite branches, sharing a tail.
	We underlined an infinitely progressing trace along this tail, %
	showing that $\prf$ is an \sccyc-proof.
\end{example}

\begin{example} \label{ex:proof-choice-def}
	Let $\defn \doteq \defintight{
		P \rul \lnot Q \\
		Q \rul \lnot P
	}$. 
	Consider the following \sccyc-pre-proof $\prf$ of $\defn \vdash$ :%
	\begin{equation*}
		\begin{prooftree}
			\hypo{\defn \vdash \underline{P} \comp}
			\infer1[($\lnot$L)]{\defn, \underline{\lnot P} \vdash}
			\infer1[(case)]{\defn, \underline{Q} \vdash}
			\infer1[($\lnot$R)]{\defn \vdash \underline{\lnot Q}}
			\infer1[(def R)]{\defn \vdash \underline{P} \comp}
			\hypo{\defn, \underline{P} \vdash \comptwo}
			\infer1[($\lnot$R)]{\defn \vdash \underline{\lnot P}}
			\infer1[(def R)]{\defn \vdash \underline{Q}}
			\infer1[($\lnot$L)]{\defn, \underline{\lnot Q} \vdash}
			\infer1[(case)]{\defn, \underline{P} \vdash \comptwo}
			\infer2[(cut)]{\defn \vdash}
		\end{prooftree}
	\end{equation*}
	The graph $\graph{\prf}$ of $\prf$ has two infinite branches.
	We underlined an infinitely progressing trace along tails of these branches, showing that $\prf$ is an \sccyc-proof.
\end{example}

\section{A Generic Infinitary Calculus} \label{sec:generic-calculus}

In his PhD thesis \cite{Brotherston2006PhD}, Brotherston studied a generic infinitary calculus \scginf for a generic logic \logic.
	The only requirements on \logic and \scginf are that: (1) \logic has a notion of \emph{sequent},\footnote{
		These \emph{sequents} need not be of the form $\seq$ as defined earlier.
		They could also take the form of formulae $\form$, for instance.
	} 
	a notion of \emph{interpretation}, and a notion of \emph{satisfaction relation} $\models$ between interpretations and sequents; and (2) the inference rules in \scginf are of the form 
	\[
	\begin{prooftree}
		\hypo{\seqabst_1}
		\hypo{\dots}
		\hypo{\seqabst_n}
		\infer3[(R)]{\seqabst}
	\end{prooftree}
	\]
	for sequents $\seqabst_1$, $\dots$, $\seqabst_n$ and $\seqabst$ in \logic.
Brotherston defined a generic notion of trace condition for \scginf-pre-proofs, stating that every infinite branch must have infinitely many \emph{progression points} in some general sense.

Our infinitary calculus \scinf is an instance of the generic infinitary calculus \scginf, and (consequently) our cyclic calculus \sccyc is an instance of the corresponding generic cyclic calculus \scgcyc.
We refer to Appendix \ref{app:generic-calculus} for the proof.
As a consequence, \scinf and \sccyc inherit the properties proven by Brotherston about \scginf and \scgcyc:

\begin{theorem}[Soundness] \label{thm:soundness}
	Let $\seqshort$ be an \foid-sequent.
	If there exists an \scinf-proof or an \sccyc-proof of $\seqshort$, then $\seqshort$ is valid.
\end{theorem}

\begin{theorem}
	It is decidable %
	whether an \sccyc-pre-proof is an \sccyc-proof.
\end{theorem}

Furthermore, we obtain \emph{cycle normalization} and the existence of \emph{trace manifolds} for \sccyc \cite{tableaux/Brotherston05,Brotherston2006PhD}.
For space reasons, we only discuss these properties informally in the main text and refer to Appendix \ref{app:generic-calculus} for a more detailed exposition.

An \sccyc-(pre-)proof $\prf = (\deriv, \repfun)$ is in \emph{cycle normal form} if for every bud $n$ of $\prf$, the companion $\repfun(n)$ of $n$ is an ancestor of $n$ in $\deriv$.
\emph{Cycle normalization} for \sccyc says that every \sccyc-(pre-)proof $\prf$ can be transformed into an \emph{equivalent} \sccyc-(pre-)proof $\prf'$ in cycle normal form; equivalent in the sense that $\prf$ and $\prf'$ have the same tree unraveling.
Cycle normalization is valuable in proof search, as it assures that one can restrict to ancestors of buds when searching for companions in pre-proofs.\footnote{
	On the other hand, proofs in cycle normal form may in the worst case be exponentially larger than the smallest proof of the same sequent that is not in cycle normal form \cite{Brotherston2006PhD}.
}

Pre-proofs in cycle normal form admit notions of \emph{trace manifold}, which provide finitary alternatives to the trace condition.
Intuitively, a trace manifold for $\prf$ consists of a collection of traces along finite paths in $\graph{\prf}$, together with conditions that these traces can be ``glued together'' to form traces along all infinite paths in $\graph{\prf}$.
Brotherston introduced two equivalent notions of trace manifolds: one in terms of \emph{strongly connected subgraphs} and one in terms of an \emph{induction order} \cite{Brotherston2006PhD}.
While these trace manifold conditions are more restrictive than the trace condition, they have the benefit of being finitary and more explicit.

\section{Completeness} \label{sec:completeness}

Like the infinitary calculus \lkidinf of Brotherston and Simpson \cite{jlc/BrotherstonS11}, \scinf is complete.

\begin{theorem}[restate=Completeness, name=Completeness] \label{thm:completeness}
	Let $\seqshort$ be an \foid-sequent.
	If $\seqshort$ is valid, then there exists an \scinf-proof of $\seqshort$.
\end{theorem}

Our completeness proof is based on the proof of cut-free completeness of \lkidinf by Brotherston and Simpson \cite{jlc/BrotherstonS11}.
It proceeds along the following lines:
\begin{itemize}
	\item We fix an \foid-sequent $\seqshort$ and construct a pre-proof $\schtree$ of $\seqshort$ called the \emph{search tree} for $\seqshort$. 
	Intuitively, $\schtree$ corresponds to an exhaustive search for a proof of $\seqshort$.
	Our notion of search tree differs from the one in \cite{jlc/BrotherstonS11}, as we incorporate the cut rule in it.
	\item If $\schtree$ is a proof of $\seqshort$, then the statement holds for $\seqshort$.
	If not, there must be a branch in $\schtree$ along which no infinitely progressing trace exists. %
	We fix such branch and refer to it as the \emph{untraceable branch} $\untbranch$.
	\item Based on $\untbranch$, we construct a structure $\cntmod$ that we show the be a countermodel of $\seqshort$.
	\item In conclusion, this shows that every \foid-sequent either has a proof or a countermodel.
\end{itemize}

The most innovative part of the proof lies in showing that $\cntmod$ is a countermodel of $\seqshort$, and in particular that 
$\cntmod$ satisfies the definition $\defn$ in $\seqshort$.
We do so by showing that for every defined atom $P(\bar{t})$ of $\defn$, $P(\bar{t})$ is true in the well-founded model $\wfmod$ of $\defn$ in $\cntmod|_{\pars{\defn}}$ iff $P(\bar{t})$ is true in $\cntmod$.
Intuitively, the untraceable branch $\untbranch$ in $\schtree$ provides enough structural information to extract a derivation of the truth of $P(\bar{t})$ in $\wfmod$ in case $\cntmod \models P(\bar{t})$ and, using the untraceability of $\untbranch$, the absence of such a reason in case $\cntmod \not\models P(\bar{t})$.

On a technical level, we achieve this %
by employing the semantic framework of \emph{justification theory} \cite{DeneckerS93,DeneckerPhD93}, which characterizes the well-founded semantics through tree-like (or graph-like in some more recent versions \cite{lpnmr/DeneckerBS15,phd/Marynissen22}) objects called \emph{justifications}. 
Justifications formalize the construction processes behind non-monotone inductive definitions, and as such, they provide an alternative to the well-founded inductions from Section \ref{sec:well-founded-semantics}.
We refer to Appendix \ref{app:justifications} for an exposition on justification theory and to Appendix \ref{app:completeness} for a proof of Theorem \ref{thm:completeness}.

As a consequence of G\"odel's first incompleteness theorem, the set of \scinf-theorems is not recursively enumerable.
Since \sccyc-proofs are finite objects, the set of \sccyc-theorems is recursively enumerable, and hence, by G\"odel's first incompleteness theorem, an analogous completeness result cannot hold for \sccyc.

\section{Cut-Elimination} \label{sec:cut-elimination}

Brotherston and Simpson showed that \lkidinf is \emph{cut-free complete}, meaning that every valid sequent admits an \lkidinf-proof without (cut).
Together with soundness, this implies \emph{cut-elimination} for \lkidinf, which says that every provable sequent in \lkidinf is also provable in \lkidinf without (cut).
Cut-elimination is a fundamental result in proof theory, originally shown by Gentzen for his sequent calculi \lk and \lj for classical resp.~intuitionistic first-order logic \cite{Gentzen35}.
Since the cut rule can intuitively be seen as the application of a lemma, cut-elimination intuitively says that every theorem can be proven without the need to introduce lemmas.
This lemma is formalized by the \emph{cut formula} $\form$ in (cut).
As the cut formula can be any \fo-formula, the cut rule is an obstacle for proof search.
Cut-elimination guarantees that the cut rule can be avoided in proof search.\footnote{
	However, restricting to cut-free proofs may significantly enlarge the size of proofs \cite{Boolos1984}.
}

Cut-elimination does {not} hold for \scinf or \sccyc.
This is shown by the proof in Example \ref{ex:proof-liar}, for instance, as a simple inspection of the inference rules of \scinf and \sccyc reveals that the sequent $\defn \vdash$ cannot be proven with any rule other than (cut).
However, cut is eliminable in Brotherston and Simpson's infinitary calculus \lkidinf \cite{jlc/BrotherstonS11}.
Slightly extending their proof of cut-elimination for \lkidinf, we obtain a proof of cut-elimination for \scinf on \foid-sequents $\sregseq$ with positive definitions $\defn$.\footnote{
	The extension of Brotherston and Simpson's proof of cut-free completeness for \lkidinf to \scinf is similar to the extension of their proof of cut-free completeness for \lkid to \scfoid \cite{arxiv}.
}

\begin{theorem}[Cut-Free Completeness] \label{thm:cut-free-completeness}
	Let $\sregseq$ be an \foid-sequent with a positive definition $\defn$.
	If $\sregseq$ is valid, then there exists a cut-free \scinf-proof of $\sregseq$.
\end{theorem}

\begin{theorem}[Cut-Elimination] \label{thm:cut-elimination}
Let $\sregseq$ be an \foid-sequent with a positive definition $\defn$.
If $\sregseq$ is provable in \scinf, then it %
is cut-free provable in \scinf.
\end{theorem}

Cut-elimination does not hold for \sccyc, as it does not hold for \clkid either.
This was shown by Oda et al.~\cite{jlc/OdaBT25}, confirming a conjecture by Brotherston \cite{Brotherston2006PhD}.

For \scfoid, we also proved a \emph{cut-restriction} result for %
sequents with stratified definitions, saying that these sequents can be proven with cuts of a very specific form \cite{arxiv}.
By the same reasoning as in \cite{arxiv}, this result extends to \scinf.

\begin{definition}
	An \foid-sequent $\sregseq$ is \emph{EC-provable} (EC standing for elementary cut) in \scinf if there exists an \scinf-proof of $\sregseq$ in which %
	every cut formula
	is of the form $\forall \bar{y}: Q(\bar{y}) \lor \lnot Q(\bar{y})$ such that $Q$ appears negatively in the body of a rule of $\defn$.
\end{definition}

\begin{theorem}[EC-Completeness]
	\label{thm:ec-completeness}
	Let $\sregseq$ be an \foid-sequent with a stratified definition $\defn$.
	If $\sregseq$ is valid, then it is EC-provable in \scinf.
\end{theorem}

\begin{theorem}[Cut-Restriction] 
	\label{thm:cut-restriction}
	Let $\sregseq$ be an \foid-sequent with a stratified definition $\defn$.
	If $\sregseq$ is provable in \scinf, then it is EC-provable in \scinf.
\end{theorem}

Note that Theorems \ref{thm:ec-completeness} and \ref{thm:cut-restriction} generalize Theorems \ref{thm:cut-free-completeness} and \ref{thm:cut-elimination}, respectively, as for positive definitions, EC-provability %
comes down to (regular) provability. %

Cut-restriction is valuable for proof search, as it narrows the %
cut formulae down to a finite set, bounded by the defined predicates of a definition $\defn$.%

It is currently unknown whether Theorem \ref{thm:cut-restriction} can be strengthened to full cut-elimination.
Cut-elimination may thus hold for the broader fragment of \foid-sequents with stratified definitions and, as far as we know, even for \foid-sequents with total definitions.

\section{Relation between \texorpdfstring{\scfoid}{SCFO(ID)}, \texorpdfstring{\scinf}{SCFO(ID)-inf} and \texorpdfstring{\sccyc}{SCFO(ID)-cyc}} \label{sec:relation-calculi}

Since every \sccyc-proof unravels into an \scinf-proof, every \sccyc-theorem is also an \scinf-theorem.
Conversely, not every \scinf-theorem is an \sccyc-theorem, as \scinf is complete whereas \sccyc is not.

Brotherston and Simpson showed that every \lkid-theorem is also a \clkid-theorem \cite{jlc/BrotherstonS11,tableaux/Brotherston05}, and they conjectured the converse to hold as well \cite{jlc/BrotherstonS11}.
Their conjecture was later disproven by Berardi and Tatsuta \cite{Berardi2019}, who provided a counterexample involving only the natural number definition from Section \ref{sec:introduction}.
The same counterexample shows that not every \sccyc-theorem is an \scfoid-theorem.
Additionally, the sequent $\defn \vdash$ from Example \ref{ex:proof-choice-def} is an \sccyc-theorem but not an \scfoid-theorem.
This sequent is not provable in \scfoid, as \scfoid is sound w.r.t.~the \emph{stable (model) semantics} (which also originates from Logic Programming \cite{iclp/GelfondL88}) \cite{arxiv}, and $\defn \vdash$ is not valid under the stable semantics.

Extending Brotherston and Simpson's proof of the fact that every \lkid-theorem is a \clkid-theorem, we can show that  every \scfoid-theorem is an \sccyc-theorem.
The crux of the proof lies in showing that every instance of the \emph{induction rule} (ind), which is the left introduction rule for defined atoms in \scfoid \cite{lpnmr/VandenEedeVD24,arxiv}, is \emph{derivable} in \sccyc:

\begin{lemma}[restate=IndDerivableInCyc, name=] 
	\label{lem:ind-derivable-in-cyc}
	For every instance of \emph{(ind)} with conclusion $\seqshort$ and premises $\seqshort_1, \dots, \seqshort_n$, there exists an \sccyc-derivation of $\seqshort$ satisfying the global trace condition, the buds of which are labeled by $\seqshort_1, \dots, \seqshort_n$.
\end{lemma}

\begin{theorem}[restate=IndThmIsCycThm, name=] 
	\label{thm:ind-thm-is-cyc-thm}
	Every \scfoid-theorem is an \sccyc-theorem.
\end{theorem}

Appendix \ref{app:relation-calculi} contains proofs of Lemma \ref{lem:ind-derivable-in-cyc} and Theorem \ref{thm:ind-thm-is-cyc-thm}.

The extensions of \lkidinf and \clkid to \scinf resp.~\sccyc are arguably more natural than the extension of \lkid to \scfoid.
To extend the induction rule (ind) to non-monotone definitions, we introduced an asymmetry between positive and negative occurrences of defined predicates \cite{lpnmr/VandenEedeVD24,arxiv}, deviating further from the %
principle of mathematical induction as commonly used in pen-and-paper proofs.
The case distinction rule (case) in \lkidinf and \clkid, on the other hand, is essentially\footnote{
	Since definitions are implicit in the logic of Brotherston and Simpson, they do not appear in the inference rules of their sequent calculi.
} identical to the corresponding rule in \lkidinf and \clkid.
To accommodate for non-monotone definitions in \lkidinf and \clkid, we only needed to refine the notion of trace condition.
While our notion of trace condition is more involved than Brotherston and Simpson's notion,
the underlying intuition of ``unfolding'' a definition infinitely often along infinite paths remains.

\scinf and \sccyc 
are 
more tailored towards the well-founded semantics than \scfoid.
Indeed, since \scfoid is sound w.r.t.~the stable semantics, it cannot prove sequents that are invalid under the stable semantics, such as the sequent from Example \ref{ex:proof-choice-def}.
\scfoid, on the other hand, 
is
more broadly applicable, as it can additionally be used to prove theorems about logic programs under the stable semantics.

\section{Conclusion} \label{sec:conclusion}

In this paper, we provided a formalization of the principle of infinite descent for general, i.e., non-monotone inductive definitions. 
We accomplished this by extending the infinitary sequent calculus \lkidinf and the cyclic sequent calculus \clkid by Brotherston and Simpson \cite{jlc/BrotherstonS11} to the sequent calculi \scinf resp.~\sccyc for \foid.
Furthermore, we extended various results about \lkidinf and \clkid to \scinf resp.~\sccyc, thereby providing a solid proof-theoretic evaluation of our sequent calculi.

This work opens several avenues for future research:
\begin{itemize}
	\item In \cite{lpnmr/VandenEedeVD24,arxiv}, we introduced the sequent calculus \scfoid for a more general notion of \foid-sequent $\seq$, in which $\Gamma$ and $\Delta$ are sets of $\foid$-formulae.
	It is currently unknown whether our calculi \scinf and \sccyc can be refined to calculi that are sound w.r.t.~this more general notion of \foid-sequent.
	\item Our cut-elimination result for \scinf can potentially be strengthened to include sequents with stratified definitions or even total definitions.
	\item In Section \ref{sec:cyclic-proofs}, we saw that \scinf and \sccyc can be used to prove non-totality of definitions (as is the case for \scfoid \cite{lpnmr/VandenEedeVD24,arxiv}).
	Proving totality of definitions is currently impossible, however, a first obstacle being that we cannot express totality in the syntax of \foid.
	Therefore, one may investigate how to extend the logic \foid and its proof systems to enable proofs of totality of definitions.
	\item Due to structural similarities, we expected justification theory to play a role in the soundness criterion for pre-proofs in \scinf and \sccyc.
	Since the trace condition does not involve justifications, the question remains whether there exist notions of infinitary and cyclic calculus that are explicitly based on justification theory.
	Since justification theory has grown into a uniform framework for capturing the semantics of several non-monotonic logics (among which the stable semantics) \cite{lpnmr/DeneckerBS15,phd/Marynissen22}, such notions could provide sequent calculi for all these logics.
	\item Since the formalization of mathematical proofs is an important motivation of our work, it would be sensible to develop \emph{natural deduction} counterparts to our sequent calculi.
	Natural deduction is another style of proof system by Gentzen, aiming to come ``as close as possible to actual reasoning'' \cite{Gentzen35}.
	\item One could extend our calculi to extensions of \foid with additional language constructs that are relevant in KRR, such as 
	aggregates, partial functions and modal operators.
	\item One could implement an automated theorem prover for \sccyc by instantiating the generic cyclic theorem prover \textsc{Cyclist} by Brotherston et al.~\cite{pls/BrotherstonGP12}.
\end{itemize}

Together with the previous work \cite{lpnmr/VandenEedeVD24,arxiv}, this paper shows that existing proof systems for monotone inductive definitions can be extended to general non-monotone definitions.

\bibliography{krrlib.bib,refs_not_in_krrlib.bib}

\appendix

\section{Well-Founded Semantics} \label{app:well-founded-semantics}

In this appendix, we rigorously define the well-founded semantics for \foid.
Large parts of this appendix appear in earlier papers \cite{tocl/DeneckerBM01,tocl/DeneckerT08,KR/DeneckerV14,lpnmr/VandenEedeVD24,arxiv}.
We start with an intuitive discussion of the semantics of monotone definitions and how it can be refined to extend to non-monotone definitions.

	For monotone inductive definitions, the standard construction process starts from the assumption that all defined facts are false, and then gradually revises some of these assumptions through iterated rule application.
	Once saturation is reached, i.e., once there are no more applicable rules, it concludes that all defined atoms that have not been derived to be true must be false.
	At each point of the construction process, partial information about the defined atoms is available: atoms that are true at an intermediate stage have reached their defined value, but atoms that are false at an intermediate stage not necessarily; these may still be derived to be true at a later stage.
	
	This construction process for monotone definitions does not work for non-monotone definitions, since non-monotone rules can only be applied safely when it is certain that the negative subformulae $\neg \formtwo$ in their bodies are true, i.e., that $\formtwo$ will not be derived later in the construction process.
	Based on the well-founded semantics for logic programs \cite{GelderRS91}, Denecker and Vennekens introduced a new formalization of the construction process \cite{KR/DeneckerV14}.
	By re-formalizing the construction process as a sequence of three-valued structures of increasing precision, it was made explicit at each stage of the process whether a defined fact was derived to be true ($\true$), derived to be false ($\false$), or not yet derived, i.e., still unknown ($\unknown$).
	As in the monotone case, a fact is derived to be true if the body of one of its rules is true in the current three-valued structure.
	Deriving an unknown fact to be (certainly) false is more subtle.
	It uses the notion of \emph{unfounded set} \cite{GelderRS91}, which is a set of unknown atoms such that, if all of these atoms are set to false, then the bodies of all the rules deriving these facts become false as well.
	Intuitively, the atoms in an unfounded set can only make each other true, and therefore, they can be safely derived to be false.
	We will formalize the aforementioned notions in the remainder of this appendix.
	
	(Two-valued) structures $\str$ are defined as usual, consisting of a set $D$ called the \emph{domain} of $\str$, also denoted by $\dom{\str}$, and a mapping from non-logical symbols $\nlsym$ to appropriate values $\nlsym^{\str}$ in $D$, called the \emph{interpretation} of $\nlsym$ in $\str$. 
	In particular, $\nlsym^{\str} \subseteq D^n$ if $\nlsym$ is an $n$-ary predicate symbol, and $\nlsym^{\str} : D^n \to D$ if $\nlsym$ is an $n$-ary function symbol.
	By abuse of notation, we will view the interpretation of a propositional symbol as a Boolean value, $\false$ or $\true$ (standing for `false' and `true', respectively),  
	and the interpretation of an object symbol as an element of $D$.
	The set of interpreted symbols of $\str$ is called the \emph{vocabulary} $\voc{\str}$ of $\str$.
	If $\voc{\str} = \vocab$, we cal $\str$ a \emph{$\vocab$-structure}.
	We say that a structure $\str$ \emph{interprets} an expression $\epsilon$ (i.e., a term or a formula) if all freely occurring non-logical symbols in $\epsilon$ are in $\voc{\str}$. 
	In this case, we call $\epsilon$ an expression \emph{over} $\voc{\str}$.
	Given a tuple of terms $\bar{t} = (t_1, \dots, t_n)$ interpreted by a structure $\str$, we write $\bar{t}^{\str}$ for $(t_1^{\str}, \dots, t_n^{\str})$.
	
	We say that a structure $\str$  \emph{expands} a structure $\strtwo$ if $\voc{\strtwo} \subseteq \voc{\str}$, $\dom{\str} = \dom{\strtwo}$, and $\nlsym^{\str}$ = $\nlsym^{\strtwo}$ for all $\nlsym \in \voc{\strtwo}$. 
	Given a structure $\str$ with domain $D$, 
	a non-logical symbol $\nlsym$, and a corresponding value $\alpha$ in $D$, we denote by $\str[\nlsym : \alpha]$ the expansion of $\str$ such that $\voc{\str[\nlsym : \alpha]} = \voc{\str} \cup \{\nlsym\}$ and $\nlsym^\str = \alpha$.
	If $\Sigma$ is a subset of $\voc{\str}$, %
	then $\str|_{\Sigma}$ denotes the structure obtained from $\str$ by restricting its vocabulary (and interpretation mapping) 
	to $\Sigma$.
	
	Denecker and Vennekens formalize the construction processes behind inductive definitions as \emph{well-founded inductions}, which are (possibly transfinite) sequences of increasingly precise \emph{three-valued structures} \cite{KR/DeneckerV14}.
	These structures have three different \textit{truth values}: $\false$, $\unknown$ and $\true$, where $\unknown$ stands for `unknown'.
	We consider two orders on the set of truth values: the \emph{truth order} $\leqt$ and the \emph{precision order} $\leqp$, 
	defined by $\false \leqt \unknown \leqt \true$ and $\unknown \leqp \false$, $\unknown \leqp \true$, respectively. 
	We also write $\lnot \false \coloneq \true$, $\lnot \unknown \coloneq \unknown$ and $\lnot \true \coloneq \false$.
	A three-valued structure $\tvstr$ is defined similarly as a two-valued structure, except that $n$-ary predicate symbols are interpreted as functions from $D^n$ to $\{ \false, \unknown, \true\}$, where $D$ is the domain of $\tvstr$. 
	We extend the truth and precision order pointwisely to functions with codomain $\{ \false, \unknown, \true\}$, i.e., given $f,g : X \to \{ \false, \unknown, \true\}$, we let $f \leqt g$ iff $f(x) \leqt g(x)$ for all $x \in X$, and similarly for $\leqp$.
	Furthermore, we extend these orders to three-valued structures $\tvstr, \tvstrtwo$ by letting $\tvstr \leqt \tvstrtwo$ iff $\tvstr$ and $\tvstrtwo$ have the same domain and vocabulary, the same interpretation for non-predicate symbols, and $P^{\tvstr} \leqt P^{\tvstrtwo}$ for every predicate symbol $P$, and similarly for $\leqp$.
	Note that a two-valued structure can be seen as a special case of a three-valued structure, by identifying a subset $S$ of $D^n$ with its characteristic function, sending every element of $S$ to $\true$ and every non-element to $\false$. 
	In this way, two-valued structures correspond to maximally precise three-valued structures. 
	In the sequel, we will frequently abuse this correspondence, for instance by writing $P^\str(\bar{a}) = \true$ instead of $\bar{a} \in P^\str$ for two-valued structures $\str$.
	
	The value $t^\tvstr$ of a term $t$ in a three-valued structure $\tvstr$ is defined as for two-valued structures, via the rule $f(t_1, \dots, t_n)^{\tvstr} = f^{\tvstr}(t_1^{\tvstr}, \dots, t_n^{\tvstr})$.
	Given an \fo-formula $\form$ and a three-valued structure $\tvstr$, we define the \emph{truth value} or \emph{truth assignment} $\form^{\tvstr}$ of $\form$ in $\tvstr$ by structural induction on $\form$, via the following rules:\footnote{
		This truth assignment is known as Kleene's truth assignment.
		In \cite{KR/DeneckerV14}, Denecker and Vennekens also allow other truth assignments.
		The only conditions that such an assignment must satisfy are $\leqp$-monotonicity (i.e., $\tvstr \leqp \tvstrtwo$ implies $\form^{\tvstr} \leqp \form^{\tvstrtwo}$) and restriction to the standard truth assignment on two-valued structures.
		Another sensible truth assignment, for instance, is \emph{supervaluation}, which defines $\form^{\tvstr}$ as $\mathsf{glb}_{\leqp} \{ \form^{\str} \mid \str \text{ is two-valued and } \tvstr \leqp \str\}$ (where $\mathsf{glb}$ stands for greatest lower bound) \cite{jp/vanFraassen66}.
		Supervaluation is more precise than Kleene's truth assignment.
		For instance, for $P^{\tvstr} = \unknown$, Kleene's truth assignment of $P \lor \lnot P$ is $\unknown$, while the supervaluation of $P \lor \lnot P$ is $\true$.
		Different truth assignments lead to different notions of \emph{well-founded model}.
		Kripke's truth assignments leads to the \emph{standard} well-founded model, while supervaluation leads to the \emph{ultimate} well-founded model \cite{DeneckerMT04}. 
		We restrict to Kleene's truth assignment for simplicity, but since our proof of soundness w.r.t.\ the well-founded semantics does not rely on the specifics of Kleene's truth assignment, our sequent calculi are also sound w.r.t.\ other notions of the well-founded semantics, corresponding to different truth assignments.
	}
	\begin{itemize}
		\item $(t=s)^{\tvstr} = \true$ if $t^{\tvstr} = s^{\tvstr}$ and $(t=s)^{\tvstr} = \false$ if $t^{\tvstr} \neq s^{\tvstr}$; %
		\item $P(\bar{t})^{\tvstr} = P^{\tvstr}\!(\bar{t}^{\tvstr})$; %
		\item $(\lnot \form)^{\tvstr} = \lnot \form^{\tvstr}$;
		\item $(\form \land \formtwo)^{\tvstr} = \min_{\leqt}\{\form^{\tvstr}, \formtwo^{\tvstr}\}$;
		\item $(\forall x: \form)^{\tvstr} = \min_{\leqt} \{ \form^{\tvstr[x:a]} \mid a \in \dom{\tvstr} \}$.
	\end{itemize}
	The rules for $\lor$, $\Rightarrow$ %
	and $\exists$ can de derived from the above rules through reformulation in terms of $\lnot$, $\land$ and $\forall$.
	It is straightforward to verify that this truth assignment is $\leqp$-monotone, i.e., $\tvstr \leqp \tvstrtwo$ implies $\form^{\tvstr} \leqp \form^{\tvstrtwo}$; %
	and that if $\tvstr$ is two-valued, then $\form^{\tvstr}$ takes the standard truth value of $\form$ in $\tvstr$.
	We extend the satisfaction relation $\models$ for \fo to three-valued structures by letting $\tvstr \models \form$ iff $\form^{\tvstr} = \true$. 
	
	Given a definition $\defn$, the interpretation of the defined predicates of $\defn$ generally depends on an interpretation of the parameters of $\defn$.
	A structure with vocabulary $\pars{\defn}$ is called a $\defn$-\emph{context}.
	For the remainder of the section, we fix a definition $\defn$ and a $\defn$-context $\cont$, and come to the notion of \emph{well-founded model} of $\defn$ in $\cont$.

	Given a tuple of non-logical symbols $\bar{\nlsym}$ (possibly) occurring in a formula $\form$, together with a tuple of corresponding values $\bar{\alpha}$ in $D$, we write $\form(\bar{\alpha}/\bar{\nlsym})^{\tvstr}$, or shortly $\form(\bar{\alpha})^{\tvstr}$, to refer to $\form^{\tvstr[\bar{\nlsym} : \bar{\alpha}]}$. 
	A \emph{domain atom} is a pair $(P, \bar{a})$, where $P$ is a predicate symbol and $\bar{a} \in D^{\ar{P}}$. 
	With an abuse of notation, we will often write $P(\bar{a})$ for $(P, \bar{a})$, and $P(\bar{a})^{\tvstr}$ for $P^{\tvstr}\!	(\bar{a})$.
	
	Given a definition $\defn$, a defined predicate $P$ of $\defn$, and an $\ar{P}$-tuple of object symbols $\bar{y}$ not occurring freely in the body $\form$ of any definitional rule of $\defn$ of the form $\defrul$, we write $\form_P(\bar{y})$ to denote the disjunction of all formulae $\exists \bar{x}: \bar{y} = \bar{t} \land \form$ such that $\defrul$ is a rule in $\defn$.
	Intuitively, $\form_P(\bar{y})$ merges all bodies of instantiations of rules in $\defn$ that could derive $P(\bar{y})$.
	For instance, for the even number definition from Example \ref{ex:def_even}, $\form_{\mathit{Even}}(y)$ is equal to $(y = \mathit{zero} \land \top) \lor (\exists n: y=\mathit{succ}(n) \land \mathit{Nat}(n) \land \lnot \mathit{Even}(n))$.
	Given $\bar{a} \in D^{\ar{P}}$, we write $\form_{P(\bar{a})}^{\tvstr}$ for $\form_{P}(\bar{y})^{\tvstr[\bar{y} : \bar{a}]}$.
	Given a set of predicate symbols $\Pi$, we denote by $At_D^\Pi$ the set of domain atoms $(P, \bar{a})$ for which $P \in \Pi$.
	
	\begin{definition} \label{def:refinement}
		Let $\tvstr$ be a three-valued $\sym{\defn}$-structure with domain $D$. 
		A three-valued $\sym{\defn}$-structure $\tvstr'$ is said to be a \emph{$\defn$-refinement} of $\tvstr$ if there exists a non-empty set $U \subseteq At_{D}^{\defp{\defn}}$ such that $A^{\tvstr} = \unknown$ for all $A \in U$, and either
		\begin{itemize}
			\item $\tvstr' = \tvstr[U : \true]$ and for all $A \in U$, $\form_A^{\tvstr} = \true$, or
			\item $\tvstr' = \tvstr[U : \false]$ and for all $A \in U$, $\form_A^{\tvstr'} = \false$.
		\end{itemize}
		Here $\tvstr[U : \true]$ refers to the structure identical to $\tvstr$, except that its interpretation of $P$ maps $\bar{a}$ to $\true$ for all $(P,\bar{a}) \in U$, and similarly for $\tvstr[U : \false]$.
	\end{definition}
	
	Note the asymmetry between deriving truth and falsity of defined atoms: the truth of defined atoms can only be derived if the corresponding bodies are true in $\tvstr$, while the falsity of defined atoms can be derived if this implies the falsity of the corresponding bodies in $\tvstr'$. 
	The second kind of refinement relates to the notion of \emph{unfounded set} in the well-founded semantics for logic programs \cite{GelderRS91}.
	
	\begin{definition} \label{def:well-founded-induction}
		A \emph{well-founded induction} of $\defn$ in $\cont$ is a sequence $(\tvstr_i)_{0 \leq i \leq \beta}$ of three-valued $\sym{\defn}$-structures extending $\cont$, such that:
		\begin{itemize}
			\item $P^{\tvstr_0}(\bar{a}) = \unknown$ for all $P \in \defp{\defn}$ and all $\bar{a} \in \dom{\tvstr_0}^{\ar{P}}$;
			\item $\tvstr_{i+1}$ is a $\defn$-refinement of $\tvstr_i$ for every ordinal $0 \leq i < \beta$; and 
			\item $\tvstr_\lambda$ is the $\leqp$-limit of $(\tvstr_i)_{0 \leq i < \lambda}$ for every limit ordinal $0 < \lambda \leq \beta$. 
		\end{itemize}
		A well-founded induction $(\tvstr_i)_{0 \leq i \leq \beta}$ is \emph{terminal} if its \emph{limit} $\tvstr_\beta$ has no $\defn$-refinement. 
		The \emph{well-founded model} of $\defn$ in $\cont$ is the limit of any terminal well-founded induction of $\defn$ in $\cont$. 
	\end{definition}
	
	\begin{remark} \label{rem:props-wf-inductions}
		Note that well-founded inductions are $\leqp$-increasing, which guarantees that the $\leqp$-limits are well-defined.\footnote{
			The existence of a $\leqp$-limit follows from a generalization of the Knaster-Tarski fixpoint theorem \cite{au/Markowsky76}, as the $\leqp$-relation defines a \emph{chain-complete partial order} on the set of three-valued structures.
		} 
		Since the truth function is $\leqp$-monotone, well-founded inductions $(\tvstr_i)_{0 \leq i \leq \beta}$ have the property that once $\form^{\tvstr_i} \in \{\false, \true\}$ for some $i$, then $\form^{\tvstr_i} = \form^{\tvstr_j}$ for all $j \geq i$.
	\end{remark}
	
	The following result %
	guarantees that the notion of well-founded model %
	is well-defined.
	
	\begin{theorem}[Denecker and Vennekens %
		\cite{KR/DeneckerV14}]
		Any definition has a terminal well-founded induction in any context. 
		Furthermore, all terminal well-founded inductions of a given definition in a given context have the same limit.
	\end{theorem}
	
	Well-founded models are not always two-valued.
	
	\begin{definition}
		We say that a definition $\defn$ is \emph{total} in a $\defn$-context $\cont$ is the well-founded model of $\defn$ in $\cont$ is two-valued.
		Otherwise, we say that $\defn$ is \emph{non-total} in $\cont$.
		If $\defn$ is total in any $\defn$-context, %
		then $\defn$ is \emph{total}.
	\end{definition}
	
	Non-totality of $\defn$ in $\cont$ means that the well-founded inductions of $\defn$ in $\cont$ leave at least one defined atom unknown. 
	This usually indicates some sort of flaw in %
	$\defn$ or $\cont$.
	
		\begin{example}
			Let $\defn$ be the even number definition from Example \ref{ex:def_even}, and let $\cont$ be the $\defn$-context with $ \dom{\defn} = \nat$ and with the standard interpretation for $\mathit{zero}$ and $\mathit{succ}$, i.e., $\mathit{zero}^\cont = 0$ and $\mathit{succ}^\cont(n) = n+1$ for all $n \in \nat$. 
			Recall that $\form_{\mathit{Even}}(y)$ is equal to $(y = \mathit{zero} \land \top) \lor \exists x: y = s(x) \land \lnot \mathit{Even}(x)$.
			As can be checked, the assignments 
			\begin{align*}
				\mathit{Even}^{\tvstr_0}(n) &= \unknown \, \text{  for } n \in \nat\\
				\mathit{Even}^{\tvstr_1}(0) &= \true, \ \ \mathit{Even}^{\tvstr_1}(n) = \unknown  \, \text{  for  } n \in \nat \setminus \{0\}\\
				\mathit{Even}^{\tvstr_2}(0) &= \true, \ \ \mathit{Even}^{\tvstr_2}(1) = \false, \ \ \mathit{Even}^{\tvstr_2}(n) = \unknown \, \text{  for } n \in \nat \setminus \{0,1\}\\
				\mathit{Even}^{\tvstr_3}(0) &= \true, \ \ \mathit{Even}^{\tvstr_3}(1) = \false, \ \ \mathit{Even}^{\tvstr_3}(2) = \true, \ \ \mathit{Even}^{\tvstr_3}(n) = \unknown \, \text{  for } n \in \nat \setminus \{0,1, 2\}\\
				& \; \; \vdots \\
				\mathit{Even}^{\tvstr_\omega}(n) &= \true \, \text{  for } n \in 2\nat ,\ \ \mathit{Even}^{\tvstr_\omega}(n) = \false \, \text{  for } n \in 2\nat+1
			\end{align*}
			specify a terminal well-founded induction $(\tvstr_i)_{0 \leq i \leq \omega}$ of $\defn$ in $\cont$ (in fact, the only one). 
			Here $\omega$ denotes the first infinite ordinal.  
			Since the well-founded model $\tvstr_\omega$ is two-valued, $\defn$ is total in $\cont$. 
			The definition $\defn$ is not total in general, however.
			Indeed, let $\cont'$ be the $\defn$-context such that $ \dom{\defn} = \{0,1\}$, $\mathit{zero}^{\cont'} = 0$, $\mathit{succ}^{\cont'}\!(0) = 1$ and $\mathit{succ}^{\cont'}\!(1) = 1$.
			As can be checked, the assignments 
			\begin{align*}
				\mathit{Even}^{\tvstr_0}(0) &= \unknown, \ \ \mathit{Even}^{\tvstr_0}(1) = \unknown\\
				\mathit{Even}^{\tvstr_1}(0) &= \true, \ \ \mathit{Even}^{\tvstr_1}(1) = \unknown
			\end{align*}
			specify a terminal well-founded induction $(\tvstr_i)_{0 \leq i \leq \omega}$ of $\defn$ in $\cont'$ (in fact, the only one). 
			Indeed, $\tvstr_1$ has no $\defn$-refinements, since $\form_{\mathit{Even}(1)}^{\tvstr_1} = \false$ and $\form_{\mathit{Even}(1)}^{\tvstr_1[\mathit{Even}(1):\false]} = \true$.
			Thus, the well-founded model of $\defn$ in $\cont'$ is strictly three-valued, and therefore, $\defn$ is non-total in $\cont'$.
		\end{example}
	
	The \emph{well-founded satisfaction relation} $\modelswf$ is a binary relation between two-valued structures $\str$ and \foid-formulae $\form$ of which all non-logical symbols are interpreted by $\str$, defined by the following rules:
	\begin{itemize}
		\item $\str \modelswf {t=s}$ if $t^{\str} = s^{\str}$;
		\item $\str \modelswf P(\bar{t})$ if $\bar{t} \in P^{\str}$;
		\item $\str \modelswf \neg \form$ if not $\str \modelswf \form$;
		\item $\str \modelswf \form \land \formtwo$ if $\str \modelswf \form$ and $\str \modelswf \formtwo$;
		\item $\str \modelswf {\forall x: \form}$ if $\str[x:a] \modelswf \form$ for all $a \in \dom{\str}$;
		\item $\str \modelswf \defn$ if $\str|_{\sym{\defn}}$ is the well-founded model of $\defn$ in $\str|_{\pars{\defn}}$.
	\end{itemize}
	The rules for $\lor$, $\Rightarrow$ %
	and $\exists$ can de derived from the above rules through reformulation in terms of $\lnot$, $\land$ and $\forall$.
	An \foid-formula $\form$ is said to be \emph{valid (under the well-founded semantics)} if $\str \modelswf \form$ for all structures $\str$ interpreting the non-logical symbols of $\form$.

\section{A Generic Infinitary Calculus} \label{app:generic-calculus}

In this appendix, we define the generic infinitary calculus \scginf from Section \ref{sec:generic-calculus}, show that our infinitary calculus \scinf is an instance of \scginf, and present the results about cycle normalization and trace manifolds mentioned in Section \ref{sec:generic-calculus}.
The conceptualization in this appendix is strongly based on \cite{Brotherston2006PhD,tableaux/Brotherston05}.

\subsection{The Generic Infinitary Calculus} \label{app:the-generic-calculus}

We define the generic calculus \scginf for the logic \logic, subject to the following conditions:
\begin{enumerate}
	\item \logic has a notion of \emph{sequent}, a notion of \emph{interpretation}, and a notion of \emph{satisfaction relation} $\models$ between interpretations and sequents; 
	\item the inference rules in \scginf are of the form 
	\[
	\begin{prooftree}
		\hypo{\seqabst_1}
		\hypo{\dots}
		\hypo{\seqabst_n}
		\infer3[(R)]{\seqabst}
	\end{prooftree}
	\]
	for sequents $\seqabst_1$, $\dots$, $\seqabst_n$ and $\seqabst$ in \logic.
\end{enumerate}
We denote the set of sequents in \logic by $\seqsabst$, the set of interpretations in \logic by $\intset$, and the set of inference rules in \scginf by $\rules$.

	\begin{definition} \label{def:generalized-trace}
		Let $\tset$ be a set and let $\tval \subseteq \tset \times \seqsabst$ be such that for any $\seqabst \in \seqsabst$, there are only finitely many $\trace \in \tset$ such that $(\trace, \seqabst) \in \tval$.
		Let $\tpair: \tset \times \tset \to (\seqsabst \times \rules \times \seqsabst \to \{0,1,2\})$ be a computable function such that for all $\trace, \trace' \in \tset$, all $\seqabst, \seqabst' \in \seqsabst$, and all $R \in \rules$: if $(\trace, \seqabst) \notin \tval$ or $(\trace', \seqabst') \notin \tval$, then $\tpair(\trace,\trace')(\seqabst ,R , \seqabst')=0$.
		Suppose that there exists a function $\otf: \tset \times \intset \to \ord$, where $\ord$ is some initial fragment of the ordinals, such that for every $I \in \intset$, every \scginf-pre-proof $\deriv = (N, E, r, \seqfun, \rulfun)$ (defined similarly as an \scinf-pre-proof) and every $n \in N$: if $I \not\models \seqfun(n)$, then there exists an $n' \in N$ and an $I' \in \intset$ such that $(n, n') \in E$ and $I' \not\models \seqfun(n')$, and for all $\trace, \trace' \in \tset$:
		\begin{enumerate}
			\item if $\tpair(\trace, \trace')(\seqfun(n), \rulfun(n), \seqfun(n'))=1$, then $\otf(\trace',I') \leq \otf(\trace,I)$;
			\item if $\tpair(\trace, \trace')(\seqfun(n), \rulfun(n), \seqfun(n'))=2$, then $\otf(\trace',I') < \otf(\trace,I)$.
		\end{enumerate}
		Then we call $\tval$ a \emph{trace value relation} and $\tpair$ a \emph{trace pair function} for \scginf.
		We call $\otf$ the \emph{ordinal trace function} associated with $\tval$ and $\tpair$. 
		A pair $(\trace, \trace') \in \tset \times \tset$ is said to be a \emph{valid trace pair} on a pair of nodes $(n, n') \in N \times N$ if $\tpair(\trace, \trace')(\seqfun(n), r(n), \seqfun(n')) \neq 0$, and is said to be a \emph{progressing trace pair} on $(n, n')$ if $\tpair(\trace, \trace')(\seqfun(n), r(n), \seqfun(n')) = 2$. 
		A sequence $(\trace_i)_i$ of elements of $\tset$ is said to be a \emph{generalized trace} along a path $(n_i)_i$ in $\deriv$ if for all $i$, $(\trace_i, \trace_{i+1})$ is a valid trace pair on $(n_i, n_{i+1})$. 
		We call $j$ a \emph{progression point} of $(\trace_i)_i$ if $(\trace_j, \trace_{j+1})$ is a progressing trace pair.
		If a generalized trace has infinitely many progression points, we call it \emph{infinitely progressing}.
		We say that $\deriv$ satisfies the \emph{generalized trace condition} w.r.t.~$\tval$, $\tpair$ and $\otf$ if for every infinite branch $\pth$ in $\deriv$, there exists an infinitely progressing generalized trace along a tail of $\pth$.
		An \emph{\scginf-proof} is an \scginf-pre-proof that satisfies the generalized trace condition.
	\end{definition}

	\begin{proposition} \label{prop:instance-generic-inf-calculus}
		There exists a trace value relation $\tval$, a trace pair function $\tpair$ and an associated ordinal trace function $\otf$ such that an \scinf-pre-proof satisfies the generalized trace condition w.r.t.~$\tval$, $\tpair$ and $\otf$ (according to Definition \ref{def:generalized-trace}) iff it satisfies the trace condition (according to Definition \ref{def:trace}).
	\end{proposition}
	
	\begin{proof}
		We start by specifying the trace value relation $\tval$, the trace pair function $\tpair$ and the associated ordinal trace function $\otf$.
		Let %
		\begin{equation*}
			\tset \coloneq \{ (\defn, \form) \mid \text{$\defn$ is an \foid-definition and $\form$ an \fo-formula} \}
		\end{equation*}
		and 
		\begin{equation*}
			\tval \coloneq \{ ((\defn, \form), (\sregseq)) \in \tset \times \seqsabst \mid \form \in \Gamma \cup \Delta\}
		\end{equation*}
		(where $\seqsabst$ is the set of \foid-sequents).
		Since \foid-sequents contain only finitely many formulae, there are only finitely many $\trace \in \tset$ such that $(\trace, \seqabst) \in \tval$ for any $\seqabst \in \seqsabst$.
		
		We define the trace pair function $\tpair: \tset \times \tset \to (\seqsabst \times \rules \times \seqsabst \to \{0,1,2\})$ as follows.
		Take $\trace, \trace' \in \tset$, $\seqabst, \seqabst' \in \seqsabst$ and $R \in \rules$.
		Write $\trace = (\defn, \form)$ and $\trace' = (\defn', \form')$.
		If $(\trace, \seqabst) \notin \tval$ or $(\trace', \seqabst') \notin \tval$, or if $\defn \neq \defn'$, then we set $\tpair(\trace,\trace')(\seqabst,R,\seqabst')= 0$.
		Now assume that $(\trace, \seqabst) \in \tval$, $(\trace', \seqabst') \in \tval$ and $\defn = \defn'$.
		If there is no instance of $R$ with conclusion $\seqabst$ and with (one of the) premise(s) $\seqabst'$, we also set $\tpair(\trace,\trace')(\seqabst,R,\seqabst')= 0$.
		If there is such an instance, then consider an \scinf-derivation $\deriv$ representing this instance.
		Concretely, the root $r$ of $\deriv$ has label $\seqabst$ and has one child per premise; in particular, $r$ has a child $n$ with label $\seqabst'$.
		If $(\form, \form')$ is not a trace along the path $(r,n)$ in $\deriv$, we set $\tpair(\trace,\trace')(\seqabst,R,\seqabst')= 0$.
		Otherwise, if $(\form, \form')$ has a progression point, %
		we set $\tpair(\trace,\trace')(\seqabst,R,\seqabst')= 2$, and if not, we set $\tpair(\trace,\trace')(\seqabst,R,\seqabst')= 1$.
		The function $\tpair$ is computable since each of its defining conditions is computable, 
		as is straightforward to check.
		In our definition of $\tpair$, we guaranteed that $\tpair(\trace,\trace')(\seqabst,R,\seqabst')= 0$ whenever $(\trace, \seqabst) \notin \tval$ or $(\trace', \seqabst') \notin \tval$.
		
		We define the ordinal trace function $\otf$ as follows.
		First, we set 
		\begin{equation*}
			\ord = \bigcup \{ \beta \mid \wfind \text{ is a well-founded induction of a definition } \defn \text{ in a $\defn$-context} \} .
		\end{equation*}
		Stated differently, $\ord$ is the supremum of all lengths of well-founded inductions.
		
		To define $\otf$, we take $\trace \in \tset$ and $I \in \intset$ (where $\intset$ is the set of two-valued structures), and we write $\trace = (\defn, \form)$.
		If $I$ is not a model of $\defn$ (under the well-founded semantics) or if $I$ does not interpret $\form$, we let $\otf(\trace, I)$ be arbitrary.
		Now assume that $I$ is a model of $\defn$ and that $I$ interprets $\form$.
		Let $\vocab$ be the vocabulary of $I$ and $\cont$ the $\defn$-context $I|_{\pars{\defn}}$.
		For every three-valued $\sym{\defn}$-structure $\tvstr$, we denote by $\tvstr^I$ the expansion of $\tvstr$ with vocabulary $\vocab$ such that $\nlsym^{\tvstr^I} = \nlsym^{\tvstr}$ for all $\nlsym \in \sym{\defn}$ and $\nlsym^{\tvstr^I} = \nlsym^{I}$ for all $\nlsym \in \vocab \setminus \sym{\defn}$.
		We set
		\begin{equation*}
			\otf(\trace, I) = \bigcap \{ i \mid \form^{\tvstr^I_i} \in \{\true, \false\} \text{ for a well-founded induction } \wfind \text{ of } \defn \text{ in } \cont \} .
		\end{equation*}
		Stated differently, $\otf(\trace, I)$ is the earliest point at which the truth or falsity of $\form$ is derived in any sequence $(\tvstr^I_i)_{0 \leq i \leq \beta}$ corresponding to a well-founded induction $\wfind$ of $\defn$ in $\cont$.
		Note that $\otf(\trace, I)$ is well-defined.
		Indeed, since $I$ interprets $\form$, $\form$ must be either true or false in $I$.
		Since $I \modelswf \defn$, there must be at least one terminal well-founded induction $\wfind$ of $\defn$ in $\cont$.
		This well-founded induction has the property that $\form^{\tvstr^I_i} \in \{\true, \false\}$ at a certain point $i$.
		
		Now take $I \in\intset$, let $\deriv = (N, E, r, \seqfun, \rulfun)$ be an \scinf-pre-proof, take $n \in N$, and suppose that $I \not\models \seqfun(n)$.
		We need to find an $n' \in N$ and an $I' \in \intset$ such that $(n, n') \in E$ and $I' \not\models \seqfun(n')$, and for all $\trace, \trace' \in \tset$:
		\begin{enumerate}
			\item if $\tpair(\trace, \trace')(\seqfun(n), \rulfun(n), \seqfun(n'))=1$, then $\otf(\trace',I') \leq \otf(\trace,I)$;
			\item if $\tpair(\trace, \trace')(\seqfun(n), \rulfun(n), \seqfun(n'))=2$, then $\otf(\trace',I') < \otf(\trace,I)$.
		\end{enumerate}
		We do so by case distinction on $\rulfun(n)$.
		Note that $\rulfun(n)$ cannot be (ax) or ($=$R), for otherwise, $\seqfun(n)$ would be valid, contradicting the fact that $I \not\models \seqfun(n)$.
		
		\textbf{Case:} $\rulfun(n)$ is (wk).
		Then $\seqfun(n)$ is of the form $\sregseq$, and $n$ has a single child $m$ such that $\seqfun(m)$ is of the form $\defn, \Gamma' \vdash \Delta'$ with $\Gamma' \subseteq \Gamma$ and $\Delta' \subseteq \Delta$.
		We pick $n' = m$ and $I' = I$.
		Clearly, $I' \not\models \seqfun(n')$.
		Take $\trace, \trace' \in \tset$, and suppose that $\tpair(\trace, \trace')(\seqfun(n), \rulfun(n), \seqfun(n'))=1$.
		Since in particular $(\trace, \seqfun(n)) \in \tval$ and $(\trace', \seqfun(n')) \in \tval$, it follows that $\trace = \trace'$, and hence, $\otf(\trace',I') = \otf(\trace,I)$.
		Note that $\tpair(\trace, \trace')(\seqfun(n), \rulfun(n), \seqfun(n'))$ cannot be equal to $2$, as $\rulfun(n)$ is not (case).
		
		\textbf{Case:} $\rulfun(n)$ is (subst).
		Then $\seqfun(n)$ is of the form $\defn[t/x], \Gamma[t/x] \vdash \Delta[t/x]$, and $n$ has a single child $m$ such that $\seqfun(m) \doteq \defn, \Gamma \vdash \Delta$.
		We pick $n' = m$ and $I' = I[x:t^I]$.
		Then clearly, $I' \not\models \seqfun(n')$.
		Take $\trace, \trace' \in \tset$, and suppose that $\tpair(\trace, \trace')(\seqfun(n), \rulfun(n), \seqfun(n'))=1$.
		Since in particular $(\trace, \seqfun(n)) \in \tval$ and $(\trace', \seqfun(n')) \in \tval$, $\trace'$ must be of the form $(\defn, \form)$ and $\trace$ of the form $(\defn[t/x], \form[t/x])$.
		Furthermore, by definition of $\tpair$, $\defn[t/x]$ must be equal to $\defn$.
		By definition of $\otf$, $\otf(\trace, I)$ is the smallest $i$ such that $\form[t/x]^{\tvstr^I_i} \in \{\true, \false\}$ for a well-founded induction $\wfind$ of $\defn$ in $I|_{\pars{\defn}}$.
		Let $\wfind$ be a well-founded induction of $\defn$ in $I|_{\pars{\defn}}$ that realizes this minimum.
		Since $x^{I'} = t^{I}$ for all $i$, it follows that $x^{\tvstr^{I'}_i} = t^{\tvstr^{I}_i}$ for all $i$, and hence, $\form^{\tvstr^{I'}_i} = \form[t/x]^{\tvstr^I_i}$ for all $i$.
		Thus, $\otf(\trace, I)$ is the smallest $i$ such that $\form^{\tvstr^{I'}_i} \in \{\true, \false\}$.
		Since $\defn = \defn[t/x]$ and $I|_{\pars{\defn[t/x]}} = I'|_{\pars{\defn[t/x]}}$, 
		$\wfind$ is also a well-founded induction of $\defn[t/x]$ in $I'|_{\pars{\defn[t/x]}}$, and therefore, $\otf(\trace',I') \leq \otf(\trace,I)$.
		As before $\tpair(\trace, \trace')(\seqfun(n), \rulfun(n), \seqfun(n'))$ cannot be equal to $2$, since $\rulfun(n)$ is not (case).
		
		\textbf{Case:} $\rulfun(n)$ is (cut).
		Then $\seqfun(n)$ is of the form $\defn, \Gamma \vdash \Delta$, and $n$ has two children $m_1$ and $m_2$ such that $\seqfun(m_1) \doteq \defn, \Gamma \vdash \form, \Delta$ and $\seqfun(m_2) \doteq \defn, \Gamma, \form \vdash \Delta$ for an \fo-formula $\form$.
		If $I$ interprets $\form$, we let $I' = I$.
		Otherwise, we let $I'$ be an arbitrary expansion of $I$ that interprets all non-logical symbols in $\form$.
		If $I' \models \form$, we let $n' = m_2$, and if $I' \not\models \form$, we let $n' = m_1$.
		This way, we ensure that $I' \not\models \seqfun(n')$.
		Take $\trace, \trace' \in \tset$, and suppose that $\tpair(\trace, \trace')(\seqfun(n), \rulfun(n), \seqfun(n'))=1$.
		By definition of $\tpair$, $\trace = \trace'$, and hence, $\otf(\trace',I') = \otf(\trace,I)$.
		
		\textbf{Case:} $\rulfun(n)$ is ($\lnot$L).
		Then $\seqfun(n)$ is of the form $\defn, \Gamma, \lnot \form \vdash \Delta$, and $n$ has a single child $m$ with $\seqfun(m) \doteq \defn, \Gamma \vdash \form, \Delta$.
		We pick $n' = m$ and $I' = I$.
		Clearly, $I' \not\models \seqfun(n')$.
		Take $\trace, \trace' \in \tset$, and suppose that $\tpair(\trace, \trace')(\seqfun(n), \rulfun(n), \seqfun(n'))=1$.
		Then $\trace$ is of the form $(\defn, \formtwo)$ and $\trace'$ of the form $(\defn, \formtwo')$.
		By definition of $\tpair$, either $\formtwo \doteq \formtwo'$ or $\formtwo \doteq \lnot \form$ and $\formtwo' \doteq \form$.
		In both cases, we have that $\otf(\trace',I') = \otf(\trace,I)$.
		Indeed, this is trivial in the first case and in the second case, it follows from the fact that $(\lnot \form)^\tvstr = \true$ iff $\form^\tvstr = \false$ for any three-valued structure $\tvstr$ that interprets $\form$.
		 
		 \textbf{Case:} $\rulfun(n)$ is ($\lnot$R).
		 This case is similar to the previous case.
		 
		 \textbf{Case:} $\rulfun(n)$ is ($\lor$L).
		 Then $\seqfun(n)$ is of the form $\defn, \Gamma, \form \lor \formtwo \vdash \Delta$, and $n$ has a two children $m_1$ and $m_2$ with $\seqfun(m_1) \doteq \defn, \Gamma , \form \vdash \Delta$ and $\seqfun(m_2) \doteq \defn, \Gamma , \formtwo \vdash \Delta$.
		 Since $(\form \lor \formtwo)^{\tvstr} = \true$ implies $\form^{\tvstr} = \true$ or $\formtwo^{\tvstr} = \true$ for all three-valued structures $\tvstr$, it follows that $\otf((\defn, \form), I) \leq \otf((\defn, \form \lor \formtwo), I)$ or $\otf((\defn, \formtwo), I) \leq \otf((\defn, \form \lor \formtwo), I)$.
		 In the first case, we pick $n' = m_1$ and in the second case, we pick $n' = m_2$.
		 In both cases, we pick $I' = I$.
		 Clearly, $I' \not\models \seqfun(n')$.
		 Take $\trace, \trace' \in \tset$, and suppose that $\tpair(\trace, \trace')(\seqfun(n), \rulfun(n), \seqfun(n'))=1$.
		 Then $\trace$ is of the form $(\defn, \formthree)$ and $\trace'$ of the form $(\defn, \formthree')$.

		 By definition of $\tpair$, either $\formthree \doteq \formthree'$, or $\formthree \doteq \form \lor \formtwo$ and $\formthree' \doteq \form$ or $\formthree' \doteq \formtwo$, depending on whether we picked $n' = m_1$ or $n' = m_2$.
		 In the first case, it trivially follows that $\otf(\trace',I') = \otf(\trace,I)$.
		 In the second case, our choice of $n'$ implies that $\otf(\trace',I') \leq \otf(\trace,I)$.
		 Indeed, we picked $n' = m_1$ in the case where $\otf((\defn, \form), I) \leq \otf((\defn, \form \lor \formtwo), I)$ and we picked $n' = m_2$ in the case where $\otf((\defn, \formtwo), I) \leq \otf((\defn, \form \lor \formtwo), I)$.
		 
		 \textbf{Case:} $\rulfun(n)$ is ($\lor$R).
		 Then $\seqfun(n)$ is of the form $\defn, \Gamma \vdash \form \lor \formtwo, \Delta$, and $n$ has one child $m$ with $\seqfun(m) \doteq \defn, \Gamma \vdash \form, \formtwo, \Delta$.
		 We pick $n' = m$ and $I'=I$.
		 Clearly, $I' \not\models \seqfun(n')$.
		 Take $\trace, \trace' \in \tset$, and suppose that $\tpair(\trace, \trace')(\seqfun(n), \rulfun(n), \seqfun(n'))=1$.
		 Then $\trace$ is of the form $(\defn, \formthree)$ and $\trace'$ of the form $(\defn, \formthree')$.
		 By definition of $\tpair$, either $\formthree \doteq \formthree'$, or $\formthree \doteq \form \lor \formtwo$ and $\formthree' \doteq \form$ or $\formthree' \doteq \formtwo$.
		 In the first case, it trivially follows that $\otf(\trace',I') = \otf(\trace,I)$.
		 For the second case, note that $(\form \lor \formtwo)^{\tvstr} = \false$ implies $\form^{\tvstr} = \false$ and $\formtwo^{\tvstr} = \false$ for all three-valued structures $\tvstr$.
		 Therefore, $\otf((\defn, \form), I) \leq \otf((\defn, \form \lor \formtwo), I)$ and $\otf((\defn, \formtwo), I) \leq \otf((\defn, \form \lor \formtwo), I)$.
		 Thus, in both cases, we have that $\otf(\trace',I') \leq \otf(\trace,I)$.
		 
		 \textbf{Case:} $\rulfun(n)$ is ($\land$L).
		 This case is similar to the previous case.
		 
		 \textbf{Case:} $\rulfun(n)$ is ($\land$R).
		 This case is similar to the the case where $\rulfun(n)$ is ($\lor$L).
		 
		 \textbf{Case:} $\rulfun(n)$ is ($\Rightarrow$L).
		 This case is similar to the the case where $\rulfun(n)$ is ($\lor$L).
		 
		 \textbf{Case:} $\rulfun(n)$ is ($\Rightarrow$R).
		 This case is similar to the the case where $\rulfun(n)$ is ($\lor$R).
		 
		 \textbf{Case:} $\rulfun(n)$ is ($\forall$L).
		 Then $\seqfun(n)$ is of the form $\defn, \Gamma, {\forall x: \form} \vdash \Delta$ and $n$ has a single child $m$ such that $\seqfun(m) \doteq \defn, \Gamma, \form[t/x] \vdash \Delta$.
		 We pick $n' = m$.
		 If $I$ interprets $t$, we let $I' = I$.
		 Otherwise, we let $I'$ be an arbitrary expansion of $I$ that interprets $t$.
		 Since $I \models {\forall x: \form}$, we have that $I[x:t^{I'}] \models \form$, and hence, $I' \models \form[t/x]$.
		 This shows that $I' \not\models \seqfun(n')$.
		 Take $\trace, \trace' \in \tset$, and suppose that $\tpair(\trace, \trace')(\seqfun(n), \rulfun(n), \seqfun(n'))=1$.
		 Then $\trace$ is of the form $(\defn, \formtwo)$ and $\trace'$ of the form $(\defn, \formtwo')$.
		 By definition of $\tpair$, either $\formtwo \doteq \formtwo'$, or $\formtwo \doteq \forall x: \form$ and $\formtwo \doteq \form$.
		 In the first case, it trivially follows that $\otf(\trace',I') = \otf(\trace,I)$.
		 For the second case, note that by definition of $\otf$, $\otf((\defn, {\forall x: \form}), I)$ is the minimal $\alpha$ such that ${(\forall x: \form)}^{\tvstr^I_\alpha} = \true$ for a well-founded induction $\wfind$ of $\defn$ in $I|_{\pars{\defn}}$.
		 Let $\wfind$ be a well-founded induction of $\defn$ in $I|_{\pars{\defn}}$ for which this minimum is attained, and fix $\alpha$ to be this minimum.
		 Since ${(\forall x: \form)}^{\tvstr^I_\alpha} = \true$, we have that ${\form}^{\tvstr^I_\alpha[x:t^{I'}]} = \true$, and since $\tvstr^{I'}_\alpha = \tvstr^I_\alpha[x:t^{I'}]$, this means that ${\form}^{\tvstr^{I'}_\alpha} = \true$.
		 This shows that $\otf((\defn, \form),I') \leq \otf((\defn, {\forall x: \form}),I)$.
		 Thus, in both cases, $\otf(\trace',I') \leq \otf(\trace,I)$.
		 
		 \textbf{Case:} $\rulfun(n)$ is ($\forall$R).
		 Then $\seqfun(n)$ is of the form $\defn, \Gamma \vdash {\forall x: \form}, \Delta$ such that $x$ does not occur freely in $\defn$, $\Gamma$ or $\Delta$, and $n$ has a single child $m$ such that $\seqfun(m) \doteq \defn, \Gamma \vdash \form, \Delta$.
		 We pick $n' = m$.
		 Since for all three-valued structures $\tvstr$, ${(\forall x : \form)}^{\tvstr} =  \false$ implies the existence of an $a \in \dom{\tvstr}$ such that ${\form}^{\tvstr[x:a]} =  \false$, there exists an $a \in \dom{\tvstr}$ such that $\otf((\defn, \form),I[x:a]) \leq \otf((\defn, {\forall x: \form}),I)$.
		 We fix such an $a \in \dom{\tvstr}$ and let $I' = I[x:a]$.
		 Since $I[x:a] \not\models \form$, and since $x$ does not occur freely in $\defn$, $\Gamma$ or $\Delta$, we have that $I' \not\models \seqfun(n')$.
		 Take $\trace, \trace' \in \tset$, and suppose that $\tpair(\trace, \trace')(\seqfun(n), \rulfun(n), \seqfun(n'))=1$.
		 By definition of $\tpair$, either $\formtwo \doteq \formtwo'$, or $\formtwo \doteq \forall x: \form$ and $\formtwo \doteq \form$.
		 In both cases, we have that $\otf(\trace',I') = \otf(\trace,I)$.
		 Indeed, in the first case, this is trivial, and in the second case, this follows from our choice of $I'$.
		 
		 \textbf{Case:} $\rulfun(n)$ is ($\exists$L).
		 This case is similar to the previous case.
		 
		 \textbf{Case:} $\rulfun(n)$ is ($\exists$R).
		 This case is similar to the case where $\rulfun(n)$ is ($\forall$L).
		 
		 \textbf{Case:} $\rulfun(n)$ is ($=$L).
		 Then $\seqfun(n)$ is of the form $\defn[t/x, s/y], \Gamma[t/x, s/y], {t=s} \vdash \Delta[t/x, s/y]$, and $n$ has a single child $m$ with $\seqfun(m) \doteq \defn[s/x, t/y], \Gamma[s/x, t/y] \vdash \Delta[s/x, t/y]$.
		 We pick $n' = m$ and $I'=I$.
		 Since $I' \models {t=s}$, it follows that $I' \not\models \seqfun(n')$.
		 Take $\trace, \trace' \in \tset$, and suppose that $\tpair(\trace, \trace')(\seqfun(n), \rulfun(n), \seqfun(n'))=1$.
		 Then $\trace$ is of the form $(\defn[t/x, s/y], \form[t/x, s/y])$ and $\trace'$ of the form $(\defn[t/x, s/y], \form[s/x, t/y])$.
		 Furthermore, $\defn[t/x, s/y] = \defn[s/x, t/y]$.
		 For any well-founded induction $\wfind$ of $\defn[t/x, s/y]$ in $I|_{\pars{\defn}}$ and any $i$, we have that $\form[t/x, s/y]^{\tvstr_i} = \form[s/x, t/y]^{\tvstr_i}$, as $t^{\tvstr_i} = s^{\tvstr_i}$.
		 Consequently, $\otf(\trace',I') = \otf(\trace,I)$.
		 
		 \textbf{Case:} $\rulfun(n)$ is (def R).
		 Then $\seqfun(n)$ is of the form $\defn, \Gamma \vdash P(\bar{t}[\bar{s}/\bar{x}]), \Delta$, and $n$ has a single child $m$ with $\seqfun(m) \doteq \defn, \Gamma \vdash \form[\bar{s}/\bar{x}], \Delta$.
		 We pick $n' = m$ and $I' = I$.
		 By Definitions \ref{def:refinement} and \ref{def:well-founded-induction}, $P(\bar{t}[\bar{s}/\bar{x}])^{\tvstr_i} = \false$ implies $\form[\bar{s}/\bar{x}]^{\tvstr_i} = \false$ for all three-valued structures $\tvstr_i$ in a well-founded induction $\wfind$ of $\defn$ in $I|_{\pars{\defn}}$.
		 This entails that $\otf((\defn, \form[\bar{s}/\bar{x}]), I') \leq \otf((\defn, P(\bar{t}[\bar{s}/\bar{x}])), I)$.
		 In particular, $I' \not\models \form[\bar{s}/\bar{x}]$, and hence, $I' \not\models \seqfun(n')$.
		 Take $\trace, \trace' \in \tset$, and suppose that $\tpair(\trace, \trace')(\seqfun(n), \rulfun(n), \seqfun(n'))=1$.
		 Then $\trace$ is of the form $(\defn, \formtwo)$ and $\trace'$ of the form $(\defn, \formtwo')$.
		 By definition of $\tpair$, either $\formtwo \doteq \formtwo'$ or $\formtwo \doteq P(\bar{t}[\bar{s}/\bar{x}])$ and $\formtwo' \doteq \form[\bar{s}/\bar{x}]$.
		 In this first case, it trivially follows that $\otf(\trace',I') = \otf(\trace,I)$, and in for second case, this follows from what we deduced before.
		 
		 \textbf{Case:} $\rulfun(n)$ is (case).
		 Then $\seqfun(n)$ is of the form $\defn, \Gamma, P(\bar{v}) \vdash \Delta$, and for every definitional rule $\defrul$ of $\defn$ defining $P$, $n$ has a child $m$ with $\seqfun(m) \doteq \defn, \Gamma, \bar{v} = \bar{t}[\bar{y}/\bar{x}], \form[\bar{y}/\bar{x}] \vdash \Delta$, where $\bar{y}$ is a tuple of object symbols with the same length of $\bar{x}$, none of which occurs freely in $\defn$, $\Gamma$, $P(\bar{v})$ or $\Delta$.
		 By definition of $\otf$, $\otf(P(\bar{v}), I)$ is the minimal $\alpha$ such that $P(\bar{v})^{\tvstr^I_\alpha} = \true$ for a well-founded induction $\wfind$ of $\defn$ in $I|_{\pars{\defn}}$.
		 Let $\wfind$ be a well-founded induction of $\defn$ in $I|_{\pars{\defn}}$ for which this minimum is attained, and fix $\alpha$ to be this minimum.

		 By Definitions \ref{def:refinement} and \ref{def:well-founded-induction}, there exists an $\alpha' < \alpha$ such that $\form_P(\bar{v})^{\tvstr^I_{\alpha'}} = \true$.
		 In other words, the truth of $\form_P(\bar{v})$ is derived strictly before the truth of $P(\bar{v})$ in $(\tvstr^I_i)_{0 \leq i \leq \beta}$.
		 By definition of $\form_P$, there exists a definitional rule $\defrul$ in $\defn$ defining $P$ such that $(\exists \bar{x} : {\bar{v} = \bar{t}} \land \form)^{\tvstr^I_{\alpha'}} = \true$.
		 This means that there exists a tuple of object symbols $\bar{a}$ in $\dom{\tvstr^I_{\alpha'}}$ %
		 with the same length as $\bar{x}$ such that $({\bar{v} = \bar{t}} \land \form)^{\tvstr^I_{\alpha'}[\bar{x}:\bar{a}]} = \true$.
		 Since none of the object symbols in $\bar{y}$ occurs in $\bar{v}$, this implies that $({\bar{v} = \bar{t}[\bar{y}/\bar{x}]} \land \form[\bar{y}/\bar{x}])^{\tvstr^I_{\alpha'}[\bar{y}:\bar{a}]} = \true$.
		 In particular, $\form[\bar{y}/\bar{x}]^{\tvstr^I_{\alpha'}[\bar{y}:\bar{a}]} = \true$.
		 Since none of the object symbols in $\bar{y}$ occur freely in $\defn$, none are in $\sym{\defn} = \voc{\tvstr_{\alpha'}}$, and hence $\tvstr^I_{\alpha'}[\bar{y}:\bar{a}] = \tvstr^{I[\bar{y}:\bar{a}]}_{\alpha'}$.
		 Therefore, $({\bar{v} = \bar{t}[\bar{y}/\bar{x}]} \land \form[\bar{y}/\bar{x}])^{\tvstr^{I[\bar{y}:\bar{a}]}_{\alpha'}} = \true$, which shows that $\otf((\defn, \form[\bar{y}/\bar{x}]), I[\bar{y}:\bar{a}]) < \otf((\defn, P(\bar{v})), I)$.
		 Let $n'$ be the child $m$ of $n$ corresponding to this definitional rule, and let $I' = I[\bar{y} : \bar{a}]$.
		 Since none of the object symbols of $\bar{y}$ occurs in $\defn$, $\Gamma$ or $\Delta$, we have that $I' \not\models \seqfun(n')$.
		 Take $\trace, \trace' \in \tset$, and suppose that $\tpair(\trace, \trace')(\seqfun(n), \rulfun(n), \seqfun(n'))=1$.
		 Then $\trace = \trace'$, and hence, $\otf(\trace',I') = \otf(\trace,I)$.
		 Now suppose that $\tpair(\trace, \trace')(\seqfun(n), \rulfun(n), \seqfun(n'))=2$.
		 Then $\trace = (\defn, P(\bar{v}))$ and $\trace' = (\defn, \form[\bar{y}/\bar{x}])$.
		 By our choice of $n'$ and $I'$, $\otf(\trace',I') < \otf(\trace,I)$. %
		
		In summary, we have shown that $\otf$ is indeed an ordinal trace function associated with $\tval$ and $\tpair$.
		
		It remains to show that any \scinf-pre-proof satisfies the generalized trace condition w.r.t.~$\tval$, $\tpair$ and $\otf$ iff it satisfies the trace condition.
		Let $\deriv$ be an \scinf-pre-proof and suppose that it satisfies the generalized trace condition w.r.t.~$\tval$, $\tpair$ and $\otf$.
		Let $\pth = (n_i)_i$ be an infinite branch in $\deriv$.
		Then there exists an infinitely progressing generalized trace $(\defn, \form_i)_i$ along $\pth$.
		This means that $\tpair((\defn, \form_i), (\defn, \form_{i+1}))(\seqfun(n_i), \rulfun(n_i), \seqfun(n_{i+1})) \in \{1, 2\}$ for all $i$.
		By definition of $\tval$, it follows that $(\form_i)_i$ is a trace along $\pth$.
		Furthermore, any progression point of $(\defn, \form_i)_i$ is a progression point of $(\form_i)_i$.
		Thus, $(\form_i)_i$ is an infinitely progressing trace along $\pth$, which shows that $\deriv$ satisfies the trace condition.
		
		Suppose conversely that $\deriv$ satisfies the trace condition, and let $\pth = (n_i)_i$ be an infinite branch in $\deriv$.
		Then there exists an infinitely progressing trace $(\form_i)_i$ along $\pth$.
		Let $\defn$ be the definition in any (and hence, every) sequent in $\deriv$.
		Then $(\defn, \form_i)_i$ is a generalized trace w.r.t.~$\tval$, $\tpair$ and $\otf$ along $\pth$.
		Indeed, by definition of $\tpair$, $\tpair(\form_i, \form_{i+1})(\seqfun(n_i), \rulfun(n_i), \seqfun(n_{i+1})) \in \{1, 2\}$ for all $i$.
		Furthermore, every progression point of $(\defn, \form_i)_i$ is a progression point of $(\form_i)_i$.
		This shows that $\deriv$ satisfies that generalized trace condition w.r.t.~$\tval$, $\tpair$ and $\otf$, and thus finishes the proof.
	\end{proof}

\subsection{Cycle Normalization} %

The cyclic proof $\deriv$ in Example \ref{ex:proof-liar} is not in \emph{cycle normal form}, as the companions $\repfun(n)$ of the buds $n$ are not ancestors of $n$ in $\deriv$.
However, by extending the branches, we can transform $\deriv$ to a proof in cycle normal form.
\emph{Cycle normalization} essentially says that this is possible for any cyclic (pre-)proof.
While this property may seem straightforward, its proof rather intricate \cite{Brotherston2006PhD}.
In this section, we introduce the notions required to the state cycle normalization theorem for \sccyc, which follows as a corollary of Proposition \ref{prop:instance-generic-inf-calculus}.

	\begin{definition}[Cycle Normal Form]
		An \sccyc-pre-proof $(\deriv, \repfun)$ is said to be in \emph{cycle normal form} if for every bud $n$ in $\deriv$, $\repfun(n)$ is an ancestor of $n$ in $\deriv$.
	\end{definition}
	
	Given a partial function $g : A \to B$ and $x, y \in A$, we write $g(x) \peq g(y)$ if $g(x)$ and $g(y)$ are both undefined, or it they are both defined and $g(x) = g(y)$.
	
	\begin{definition}[Derivation Homomorphism]
		Let $\deriv$ and $\deriv'$ be \scinf-derivations, and write $\deriv = (N, E, r, \seqfun, \rulfun)$ and $\deriv' = (N', E', r', \seqfun', \rulfun')$.
		A \emph{derivation homomorphism} from $\deriv$ to $\deriv'$ is a function $f : N \to N'$ such that for all $n, m \in N$: $(n,m) \in E$ iff $(f(n),f(m)) \in E'$, $\seqfun(n) = \seqfun(m)$, and $\rulfun(n) \peq \rulfun'(n)$. 
		We say that a derivation homomorphism $f : N \to N'$ from $\deriv$ to $\deriv'$ is \emph{invertible} if there exists a derivation homomorphism $g : N' \to N$ from $\deriv'$ to $\deriv$ such that $g(f(n)) = n$ for all $n \in N$ and $f(g(n')) = n'$ for all $n' \in N'$.
	\end{definition}
	
	\begin{definition}[Tree Unraveling]
		Let $\prf = (\deriv, \repfun)$ be an \sccyc-pre-proof with $\deriv = (N, E, r, \seqfun, \rulfun)$.
		Let $\branchset{\graph{\prf}}$ denote the set of finite branches in the graph $\graph{\prf}$ of $\prf$, i.e., the set of finite sequences $(n_0, \dots, n_k)$ %
		such that $n_0 = r$ and for all $i \in \{0, \dots, k-1\}$: $(n_i, n_{i+1}) \in E$.
		The \emph{tree unraveling} $\trunf{\prf}$ of $\prf$ is the \scinf-pre-proof $(N', E', r', \seqfun', \rulfun')$ such that:
		\begin{itemize}
			\item $N' = \branchset{\graph{\prf}}$;
			\item $E' = \{ ((n_0, \dots, n_k), (n_0, \dots, n_k, n_{k+1})) \mid (n_k, n_{k+1}) \in E \}$;
			\item $r' = (r)$;
			\item $\seqfun'((n_0, \dots, n_k)) = \seqfun(n_k)$ for all $(n_0, \dots, n_k) \in \branchset{\graph{\prf}}$; and
			\item $\rulfun'((n_0, \dots, n_k)) = \rulfun(n_k)$ for all $(n_0, \dots, n_k) \in \branchset{\graph{\prf}}$.\footnote{
				Note that $\rulfun$ is total on the nodes of $\graph{\prf}$.
			} %
		\end{itemize}
	\end{definition}
	
	\begin{proposition}
		The tree unraveling $\trunf{\prf}$ of an \sccyc-proof $\prf$ is an \scinf-proof.
	\end{proposition}
	
	\begin{definition}
		Let $\prf$ and $\prf'$ be \sccyc-pre-proofs. 
		We say that $\prf$ is \emph{equivalent to} $\prf'$, and write $\prf \eqp \prf'$, if there exists an invertible derivation homomorphism from $\trunf{\prf}$ to $\trunf{\prf'}$. 
	\end{definition}
	
	\begin{theorem}[Cycle Normalization]
		Any \sccyc-(pre-)proof is equivalent to an \sccyc-(pre-)proof in cycle normal form.
		Furthermore, there exists a procedure that transforms any \sccyc-(pre-)proof with $n$ nodes into an \sccyc-(pre-)proof in cycle normal form with no more than $n^{2^{n/2}}$ nodes.
	\end{theorem}

	\subsection{Trace Manifolds}

	Pre-proofs in cycle normal form admit finitary alternatives to the trace condition, formulated in terms of \emph{trace manifolds}.
	Brotherston introduced two equivalent notions of trace manifold: one in terms of \emph{strongly connected subgraphs} and one in terms of an \emph{induction order} \cite{Brotherston2006PhD}.
	In this section, we provide the required notions to rigorously define these finitary soundness conditions for \sccyc-pre-proofs in cycle normal form.
	The fact that they entail soundness, and that they are equivalent to each other, follows as a corollary of Proposition \ref{prop:instance-generic-inf-calculus}.

	\begin{definition}[Basic Cycle]
		Let $\prf = (\deriv, \repfun)$ be an \sccyc-pre-proof in cycle normal form and let $n$ be a bud in $\prf$.
		The \emph{basic cycle} $\bacyc{n}$ in $\graph{\prf}$ is the path in $\graph{\prf}$ obtained from the unique path from $\repfun(n)$ to $n$ in $\deriv$ by replacing the final edge $(m,n)$ with the edge $(m, \repfun(n))$ in $\graph{\prf}$.\footnote{
			The fact that $\prf$ is in cycle normal form guarantees that there is a unique path from $\repfun(n)$ to $n$ in $\deriv$.
		}
	\end{definition}

	\begin{definition}[Structural Connectivity, Sprenger and Dam \cite{fossacs/SprengerD03}]
		Let $\prf = (\deriv, \repfun)$ be an \sccyc-proof in cycle normal form and $\budset$ the set of buds of $\deriv$.
		We define a binary relation $\structord$ on $\budset$ by letting $n_1 \structord n_2$ if $\repfun(n_1)$ appears on the basic cycle $\bacyc{n_2}$ in $\graph{\prf}$.
	\end{definition}
	
	Contrary to what its notation may suggest, the relation $\structord$ is generally neither anti-symmetric nor transitive.
	
	\begin{definition}[Weak $\structord$-Connectivity]
		Let $\prf$ be an \sccyc-proof and $\budset$ the set of buds of $\prf$.
		We say that a subset $S$ of $\budset$ is \emph{weakly $\structord$-connected} if for all $n,m \in \budset$, there exists a finite sequence $n_1, \dots, n_k$ such that $n_1 = n$, $n_k = m$, and for all $i \in \{1, \dots, k-1\}$: $n_i \structord n_{i+1}$ or $n_{i+1} \structord n_{i}$.
	\end{definition}
	
	\begin{definition}[Trace Manifold]
		Let $\prf = (\deriv, \repfun)$ be an \sccyc-pre-proof in cycle normal form and let $\budset = \{n_1, \dots, n_k\}$ be the set of buds of $\deriv$.
		A \emph{trace manifold} for $\prf$ is a set of traces
		\[
		\{
		\trace_{S,i} \mid
		S = \bigcup_{n \in \budset_S} \mathcal{C}_n  \text{  where $\budset_S$ is a weakly $\structord$-connected subset of $\budset$ containing $n_i$}
		\}
		\]
		such that:
		\begin{itemize}
			\item for all $S$ and $i$, $\trace_{S, i}$ is a trace along the basic cycle $\bacyc{n_i}$ that takes the same value at both instances of $\repfun(n_i)$ in $\bacyc{n_i}$;\footnote{
				A trace along a path in a graph $\graph{\prf}$ of an \sccyc-pre-proof is defined similarly as in Definition \ref{def:trace}.
			}
			\item for all $S$: if $n_i, n_j \in \budset_S$ and $n_i \structord n_j$, then $\trace_{S, i}(\repfun(n_i)) = \trace_{S, j}(\repfun(n_i))$; and
			\item for all $S$, there exists an $i$ such that $\trace_{S,i}$ has at least one progression point. %
		\end{itemize}
	\end{definition}
	
	\begin{theorem}
		Let $\prf$ be an \sccyc-pre-proof in cycle normal form.
		If $\prf$ has a trace manifold, then $\prf$ is an \sccyc-proof.
	\end{theorem}

	\begin{definition}[Induction Order]
		Let $\prf$ be an \sccyc-pre-proof in cycle normal form and let $\budset$ be the set of buds of $\prf$.
		A partial order $\indord$ on $\budset$ is said to be an \emph{induction order} for $\prf$ if: 
		\begin{itemize}
			\item $\indord$ is \emph{forest-like}, i.e., if $n \indord m_1$ and $n \indord m_2$, then $m_1 = m_2$ or $m_1 \indord m_2$ or $m_2 \indord m_1$; and
			\item every weakly $\structord$-connected subset of $\budset$ has a greatest element w.r.t.~$\indord$, i.e., an element $n \in \budset$ such that $m \indord n$ for all $m \in \budset$.
		\end{itemize}
	\end{definition}
	
	\begin{definition}[Ordered Trace Manifold]
		Let $\prf = (\deriv, \repfun)$ be an \sccyc-pre-proof in cycle normal form and let $\budset = \{n_1, \dots, n_k\}$ be the set of buds of $\deriv$.
		Let $\indord$ be an induction order for $\prf$.
		An \emph{ordered trace manifold} w.r.t.~$\indord$ is a set of traces $\{ \trace_{i j} \mid n_i \indord n_j \}$ such that for all $i,j \in \{1, \dots, k\}$:
		\begin{itemize}
			\item $\trace_{i j}$ is a trace along the unique path from $\repfun(n_i)$ to $n_i$ in $\deriv$;
			\item $\trace_{i j}(n_i) = \trace_{i j}(\repfun(n_i))$;
			\item %
			if $n_i \indord n_k$, $n_j \indord n_k$ and $n_j \structord n_i$,
			then $\trace_{i k}(\repfun(n_j)) = \trace_{j k}(\repfun(n_j))$; and
			\item $\trace_{i i}$ has at least one progression point. %
		\end{itemize}
	\end{definition}
	
	\begin{proposition}
	Let $\prf$ be an \sccyc-proof in cycle normal form.
	Then $\prf$ has a trace manifold iff $\prf$ has an ordered trace manifold w.r.t.~some induction order $\indord$ for $\prf$.
	\end{proposition}
	
	\begin{theorem}
		Let $\prf$ be an \sccyc-pre-proof in cycle normal form and let $\indord$ be an induction order for $\prf$.
		If $\prf$ has an ordered trace manifold w.r.t.~$\indord$, then $\prf$ is an \sccyc-proof.
	\end{theorem}

\section{Justification Theory} \label{app:justifications}

In Section \ref{sec:well-founded-semantics}, we learned that the semantics of non-monotone inductive definitions can be captured with the notion of well-founded induction.
An alternative way to capture the semantics of non-monotone inductive definitions is with the notion of \emph{justification} \cite{DeneckerS93,DeneckerPhD93}.
Denecker developed justification theory to assign declarative meaning to logic programs by interpreting them as (non-monotone) inductive definitions.
Intuitively, a justification is a graph-like object showing how the truth or falsity of a defined atom 
follows from the rules of the corresponding definition.

The justification theory presented in this appendix can be seen as an instance of the framework by Marynissen \cite{phd/Marynissen22}. %
Marynissen considers three-valued models of various semantics for non-monotonic logics, whereas our purposes only require two-valued well-founded models.

\begin{definition}[Fact] \label{def:fact}
	Let $\vocab$ be a vocabulary and $D$ a set.
	For each $a \in D$, we fix an object symbol $c_a$ that does not occur in $\vocab$, and we denote the set of all these object symbols by $\consts{D}$.
	Let $\vocab^{\consts{D}} \coloneq \vocab \cup \consts{D}$.
	A \emph{fact} w.r.t.~$\vocab$ and $D$ is a tuple $(\form, v)$, where $\form$ is an \fo-formula over $\vocab^{\consts{D}}$ and $v \in \{\true, \false\}$.
	We call a fact $(\form, v)$ \emph{positive} if $v = \true$, \emph{negative} if $v = \false$, and \emph{atomic} if $\form$ is an atom.
	We denote the set of facts w.r.t.~$\vocab$ and $D$ by $\facts{\vocab}{D}$.
\end{definition}

By abuse of notation, we will often denote an object symbol $c_a \in \consts{D}$ simply by $a$.

\begin{definition} \label{def:just-derives}
	Let $\defn$ be a definition and $\cont$ a $\defn$-context (i.e., a $\pars{\defn}$-structure). 
	Let $\vocab = \sym{\defn}$ and $D = \dom{\cont}$.
	Let $\consts{D}$ be as in Definition \ref{def:fact}, and let $\cont^{\consts{D}}$ be the expansion of $\cont$ with vocabulary $\vocab \cup \consts{D}$ that interprets every object symbol $c_a \in \consts{D}$ as $a$.
	Let $\formtwo$ be an \fo-formula over $\vocab^{\consts{D}}$ and $P(\bar{s})$ an atom over $\vocab^{\consts{D}}$ such that $P \in \defp{\defn}$.\footnote{
		An \fo-formula \emph{over} a vocabulary $\vocab$ is an \fo-formula $\form$ such that all non-logical symbols in $\form$ are in $\vocab$.
	}
	We say that $\formtwo$ \emph{derives} $P(\bar{s})$ w.r.t.~$\defn$ and $\cont$ if there exists a rule $\defrul$ in $\defn$ and a tuple of object symbols $\bar{c}$ in $\consts{D}$ such that $\formtwo \doteq \form[\bar{c}/\bar{x}]$ and $\cont^{\consts{D}} \models \, {\bar{s} = \bar{t}[\bar{c}/\bar{x}]}$. 
\end{definition}

\begin{definition}[Justification] \label{def:justification}
	Let $\defn$ be a definition and $\cont$ a $\defn$-context. %
	Let $\vocab = \sym{\defn}$ and $D = \dom{\cont}$, and let $\fact$ be a fact w.r.t.~$\vocab$ and $D$.
	Let $J = (N, E, r, \labfun)$ be a connected, rooted, labeled, directed graph, where $N$ is the set of \emph{nodes}, $E \subseteq N \times N$ the set of \emph{edges}, $r \in N$ the \emph{root} and $\labfun: N \to \facts{\defn}{D}$ the \emph{labeling function}.
	Assume that $\labfun(r) = \fact$, and that for every node $n \in N$:
	\begin{itemize}
		\item if $\labfun(n)$ is an atomic fact $(\form, v)$ such that $\form$ is not a defined atom of $\defn$, 
		then $n$ has no children, and $\cont^{\consts{D}} \models \form$ iff $v=\true$ (where $\cont^{\consts{D}}$ is as in Definition \ref{def:just-derives});
		\item otherwise, $n$ either has no children, in which case it is called a \emph{bud} of $J$, or:
		\begin{itemize}
			\item if $\labfun(n)$ is an atomic fact $(P(\bar{t}), \true)$ with $P \in \defp{\defn}$, then $n$ has one child, with label $(\form, \true)$, such that $\form$ derives $P(\bar{t})$ w.r.t.~$\defn$ and $\cont$;
			\item if $\labfun(n)$ is an atomic fact $(P(\bar{t}), \false)$ with $P \in \defp{\defn}$, then for every $\form$ that derives $P(\bar{t})$ w.r.t.~$\defn$ and $\cont$, $n$ has a child with label $(\form, \false)$;
			\item if $\labfun(n)$ is of the form $(\lnot \form, \true)$, then $n$ has one child, with label $(\form, \false)$; %
			\item if $\labfun(n)$ is of the form $(\lnot \form, \false)$, then $n$ has one child, with label $(\form, \true)$; 
			\item if $\labfun(n)$ is of the form $(\form \land \formtwo, \true)$, then $n$ has two children, with labels $(\form, \true)$ and $(\formtwo, \true)$;
			\item if $\labfun(n)$ is of the form $(\form \land \formtwo, \false)$, then $n$ has one child, with label $(\form, \false)$ or $(\formtwo, \false)$;
			\item if $\labfun(n)$ is of the form $({\forall x: \form}, \true)$, then, for every $a \in D$, $n$ has a child with label $(\form[c_a/x], \true)$ (and no other children);
			\item if $\labfun(n)$ is of the form $({\forall x: \form}, \false)$, then $n$ has one child, with label $(\form[c_a/x], \false)$ for some $a \in D$.
			\item (The rules for $\lor$, $\Rightarrow$ %
			and $\exists$ can be derived from the rules above.)
		\end{itemize}
	\end{itemize}
	Then we call $J$ a \emph{justification} of $\fact$ w.r.t.~$\defn$ and $\cont$.
	If furthermore $J$ has no buds, we say that $J$ is \emph{locally complete}.
\end{definition}

\begin{example} \label{ex:just-even}
	Let $\defn$ be the definition of $\mathit{Even}$ from Example \ref{ex:def_even} and let $\cont$ be the standard $\defn$-context, i.e., $\dom{\cont} = \nat$, $\mathit{zero}^{\str} = 0$, $\mathit{succ}^{\str}(n) = n+1$ for all $n \in \nat$ and $\mathit{Nat}^{\str} = \nat$. 
	The tree in Figure \ref{fig:just-two-even} is a graphical representation of a locally complete justification $J$ of the fact $(\mathit{Even}(2), \true)$ w.r.t.~$\defn$ and $\cont$.
	
	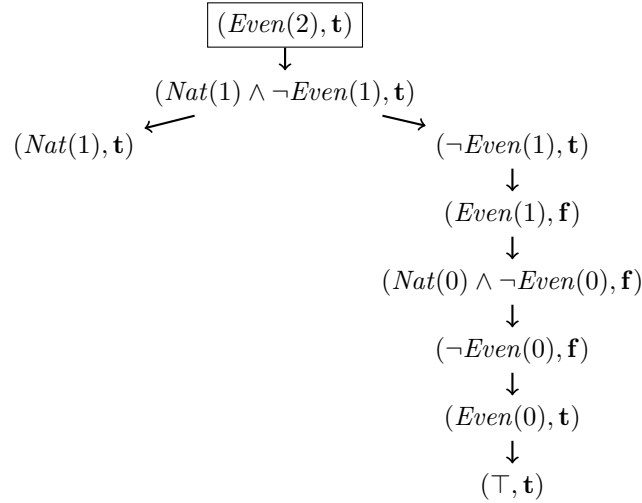
\begin{figure}[h]
		\centering
				\begin{tikzpicture}[ node distance = -3mm and 8mm, ]
					\node[draw] (even2) at (0,0) {$(\mathit{Even}(2), \true)$}; 
					\node (nat1andnoteven1) [below= 3 mm of even2] {$(\mathit{Nat}(1) \land \lnot \mathit{Even}(1), \true)$}; \draw[just] (even2) to (nat1andnoteven1); 
					\node (nat1) [below left= 1mm and 0mm of nat1andnoteven1] {$(\mathit{Nat}(1), \true)$}; \draw[just] (nat1andnoteven1) to (nat1); 
					\node (noteven1) [below right= 1mm and 0mm of nat1andnoteven1] {$(\lnot \mathit{Even}(1), \true)$}; \draw[just] (nat1andnoteven1) to (noteven1); 
					\node (even1) [below= 3mm of noteven1] {$(\mathit{Even}(1), \false)$}; \draw[just] (noteven1) to (even1); 
					\node (nat0andnoteven0) [below= 3mm of even1] {$(\mathit{Nat}(0) \land \lnot \mathit{Even}(0), \false)$}; \draw[just] (even1) to (nat0andnoteven0); 
					\node (noteven0) [below= 3mm of nat0andnoteven0] {$(\lnot \mathit{Even}(0), \false)$}; \draw[just] (nat0andnoteven0) to (noteven0); 
					\node (even0) [below= 3mm of noteven0] {$(\mathit{Even}(0), \true)$}; \draw[just] (noteven0) to (even0); 
					\node (true) [below= 3mm of even0] {$(\top, \true)$}; \draw[just] (even0) to (true); 
				\end{tikzpicture}
			\caption{
					A justification with only finite branches.
					Since both leafs are labeled with facts that do not contain defined atoms of $\defn$, this justification is locally complete.
				} 
			\label{fig:just-two-even}
		\end{figure}
\end{example}
	
	\begin{example} \label{ex:just-reachable}
		Let $\defn$ be the definition
		\[
		\defin{
			\mathit{Reachable}(\mathit{start}) \rul \top \\
			\forall x,y: \mathit{Reachable}(y) \rul \mathit{Reachable}(x) \land \mathit{Edge}(x,y)
		} 
		\]
		and let $\cont$ be the $\defn$-context defined by $\dom{\cont} = \{a, b\}$, $\mathit{start}^{\cont} = a$ and $\mathit{Edge}^{\cont} = \{ (b,b) \}$. %
		The graph in Figure \ref{fig:just-b-not-reachable} represents a locally complete justification of the fact $(\mathit{Reachable}(b), \false)$ w.r.t.~$\defn$ and $\cont$, and the graph in Figure \ref{fig:just-b-reachable} represents a locally complete justification of the fact $(\mathit{Reachable}(b), \true)$ w.r.t.~$\defn$ and $\cont$.
		\begin{figure}[h]
			\centering
			\begin{tikzpicture}[ node distance = -3mm and 8mm, ]
				\node[draw] (rb) at (0,0) {$(\mathit{Reachable}(b), \false)$}; 
				\node (raandeab) [below left= 1mm and 0mm of rb] {$(\mathit{Reachable}(a) \land \mathit{Edge}(a,b), \false)$}; \draw[just] (rb) to (raandeab); 
				\node (rbandebb) [below right= 1mm and 0mm of rb] {$(\mathit{Reachable}(b) \land \mathit{Edge}(b,b), \false)$}; \draw[just] (rb) to (rbandebb); 
				\draw[just] (rbandebb) to (rb); 
				\node (eab) [below= 3mm of raandeab] {$(\mathit{Edge}(a,b), \false)$}; 
				\draw[just] (raandeab) to (eab); 
			\end{tikzpicture}
			\caption{A justification with a negative infinite branch.}
			\label{fig:just-b-not-reachable}
		\end{figure}
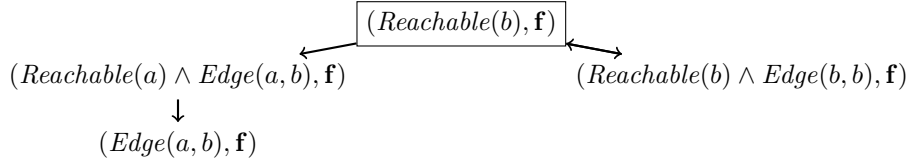
		
		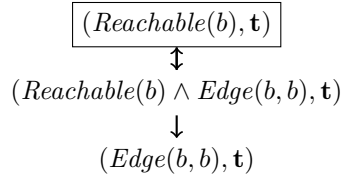
\begin{figure}[h]
			\centering
			\begin{tikzpicture}[ node distance = -3mm and 8mm, ]
				\node[draw] (rb) at (0,0) {$(\mathit{Reachable}(b), \true)$}; 
				\node (rbandebb) [below= 3mm of rb] {$(\mathit{Reachable}(b) \land \mathit{Edge}(b,b), \true)$}; 
				\draw[just] (rb) to (rbandebb); 
				\draw[just] (rbandebb) to (rb); 
				\node (ebb) [below= 3mm of rbandebb] {$(\mathit{Edge}(b,b), \true)$}; 
				\draw[just] (rbandebb) to (ebb); 
			\end{tikzpicture}
			\caption{A justification with a positive infinite branch.}
			\label{fig:just-b-reachable}
		\end{figure}
	\end{example}
	
	Example \ref{ex:just-reachable} shows that not all justifications are acceptable, in the sense that not all justifications derive true facts.
	Justification theory determines which justifications are acceptable by specifying which sorts of \emph{branches} are allowed.
	The latter depends on the considered semantics in general.
	For the well-founded semantics, the allowed branches are precisely the finite and the \emph{negative} infinite branches.
	
	\begin{definition}
		Let $J =(N, E, r, \labfun)$ be a locally complete justification. %
		A \emph{branch} $\branch$ in $J$ is a maximal path $(n_i)_i$ in $(N, E)$ starting at the root $r$.
		We call an infinite branch $(n_i)_{i}$ \emph{positive} resp.~\emph{negative} if there exists an $i_0$ such that for all $i \geq i_0$: if $\labfun(n_i)$ is an atomic fact, it is positive resp.~negative.\footnote{
			Note that any infinite branch in a justification visits infinitely many atomic facts.
		}
		We call an infinite branch \emph{mixed} if it is %
		not positive or negative.
	\end{definition}
	
	\begin{definition}[Good Justification]
		We say that a locally complete justification $J$ is \emph{good} (under the well-founded semantics)
		if every infinite branch in $J$ is negative.
	\end{definition}
	
	\begin{example}
		The justification in Figure \ref{fig:just-two-even} is good as it has no infinite branches.
		The first justification in Figure \ref{fig:just-b-not-reachable} is good as its only infinite branch is negative.
		The second justification in Figure \ref{fig:just-b-reachable} is not good as it has a positive infinite branch. 
	\end{example}
	
	The following proposition is an instance of \cite[Theorem 6.5.12]{phd/Marynissen22}. 
	It shows that the well-founded semantics can be characterized with justifications.
	
	\begin{proposition} \label{prop:wf-sem-just}
		Let $\defn$ be a definition and $\str$ a structure interpreting $\defn$.
		Let $D \coloneq \dom{\str}$ and $\cont \coloneq \str|_{\pars{\defn}}$.
		Then $\str \modelswf \defn$ iff for every defined predicate $P$ of $\defn$ and every $\bar{a} \in D^{\ar{P}}$:
		\begin{itemize}
			\item if $\bar{a} \in P^\str$, then there exists a good justification of $(P(\bar{a}), \true)$ w.r.t.~$\defn$ and $\cont$;
			\item if $\bar{a} \notin P^\str$, then there exists a good justification of $(P(\bar{a}), \false)$ w.r.t.~$\defn$ and $\cont$.
		\end{itemize}
	\end{proposition}

\section{Completeness Proof} \label{app:completeness}

In this appendix, we prove the completeness of \scinf:

\Completeness*
	
	We start the proof by fixing an \foid-sequent $\sregfix$.
	Let $\vocab$ be the vocabulary consisting of all non-logical symbols in $\sregfix$, and $\terms$ the set of all terms composed of function symbols in $\vocab$.
	Let $\varset$ be a countably infinite set of object symbols.
	We will construct a search tree for $\sregfix$ alongside a \emph{schedule}, which is an infinite list of \emph{schedule elements}, corresponding to applications of inference rules.

	\begin{definition}[Schedule]
		A \emph{schedule element} is defined by the following rules:
		\begin{itemize}
			\item If $\form$ and $\formtwo$ are \fo-formulae over $\vocab$, then $\lnot \form$, $\form \land \formtwo$, $\form \lor \formtwo$, $\form \Rightarrow \formtwo$, $\form \Leftrightarrow \formtwo$ and $\cut{\form}$ are schedule elements.
			\item If $\form$ is an \fo-formula over $\vocab$, $x \in \varset$ and $t \in \terms$, then $\langle {\forall x: \form},\, t \rangle$ and $\langle {\exists x: \form},\, t \rangle$ are schedule elements.
			\item If $t$ and $s$ are terms in $\terms$, $x$ and $y$ distinct object symbols in $\varset$, and $\Gamma$ and $\Delta$ finite sets of \fo-formulae over $\vocab$, then $\langle t=s, x, y, \Gamma, \Delta \rangle$ is a schedule element. 
			\item If $\defrul$ is a definitional rule in $\defn$ and $\bar{s}$ a tuple of terms with the same length as $\bar{x}$, then $\langle P(\bar{t}[\bar{s}/\bar{x}]), \form[\bar{s}/\bar{x}] \rangle$ is a schedule element.
		\end{itemize}
		A \emph{schedule} is a sequence $\sched$ of schedule elements in which every schedule element appears infinitely often.
	\end{definition}
	
	Henceforth, we fix a schedule $\sched$, whose existence follows from the fact that there are countably many schedule elements.
	
	\begin{definition}[Search Tree] \label{def:search-tree}
		The \emph{search tree} for $\sregfix$ %
		is an \scinf-derivation of $\sregfix$, which we construct as the limit $\schtree$ of a sequence of \scinf-derivations $(\deriv_i)_{i}$ of $\sregfix$ such that for every $i$, $\deriv_{i+1}$ extends $\deriv_i$.\footnote{
			We say that an \scinf-derivation $\deriv = (N, E, r, \seqfun, \rulfun)$ \emph{extends} an \scinf-derivation $\deriv' = (N', E', r', \seqfun', \rulfun')$ if $N' \subseteq N$, $E' = E \cap (N' \times N')$, $r' = r$, $\seqfun' = \seqfun|_{N'}$ and $\rulfun' = \rulfun|_{N'}$.
		}
		The first derivation $\deriv_0$ consists of a single node with label $\sregfix$.
		Given a derivation $\deriv_i$, we construct $\deriv_{i+1}$ as follows.
		First, we replace every bud in $\deriv_i$ with label $\sregseq$ such that $\Gamma \cap \Delta \neq \emptyset$, $\bot \in \Gamma$, or $\top \in \Delta$ with the following derivation:
		\begin{equation*}
			\begin{prooftree}
				\infer0[\textup{(ax)}]{\sregseq}
			\end{prooftree}
		\end{equation*}
		Similarly, we replace every bud in $\deriv_i$ with label $\sregseq$ such that ${t=t} \in \Delta$ for a term $t$ with the following derivation:
		\begin{equation*}
			\begin{prooftree}
				\infer0[\textup{($=$R)}]{\sregseq}
			\end{prooftree}
		\end{equation*}

		We proceed by case distinction on the $i$-th schedule element $\xi_{i}$ in the fixed schedule $\sched$:
		
		\textbf{Case}: $\xi_{i}$ is of the form $\lnot \form$.
		We replace every bud with label $\sregseq$ such that $\lnot \form \in \Gamma$ with the derivation
		\begin{equation*}
			\begin{prooftree}
				\hypo{\defn, \Gamma \vdash \form, \Delta}
				\infer1[\textup{($\lnot$L)}]{\sregseq}
			\end{prooftree}
		\end{equation*}	
		and every bud with label $\sregseq$ such that $\lnot \form \in \Delta$ with the derivation
		\begin{equation*}
			\begin{prooftree}
				\hypo{\defn, \Gamma, \form \vdash \Delta}
				\infer1[\textup{($\lnot$R)}]{\sregseq}
			\end{prooftree}
		\end{equation*}	
		
		\textbf{Case}: $\xi_{i}$ is of the form $\form \land \formtwo$.
		We replace every bud with label $\sregseq$ such that $\form \land \formtwo \in \Gamma$ with the derivation
		\begin{equation*}
			\begin{prooftree}
				\hypo{\defn, \Gamma, \form, \formtwo \vdash \Delta}
				\infer1[\textup{($\land$L)}]{\sregseq}
			\end{prooftree}
		\end{equation*}	
		and every bud with label $\sregseq$ such that $\form \land \formtwo \in \Delta$ with the derivation
		\begin{equation*}
			\begin{prooftree}
				\hypo{\defn, \Gamma \vdash \form, \Delta}
				\hypo{\defn, \Gamma \vdash \formtwo, \Delta}
				\infer2[\textup{($\land$R)}]{\sregseq}
			\end{prooftree}
		\end{equation*}	
		
		\textbf{Cases}: $\xi_{i}$ is of the form $\form \lor \formtwo$ or $\form \Rightarrow \formtwo$.
		These cases are similar to the previous case.
		
		\textbf{Case}: $\xi_{i}$ is of the form $\cut{\form}$.
		We replace every bud with label $\sregseq$ with the derivation
		\begin{equation*}
			\begin{prooftree}
				\hypo{\defn, \Gamma \vdash \form, \Delta}
				\hypo{\defn, \Gamma, \form \vdash \Delta}
				\infer2[\textup{(cut)}]{\sregseq}
			\end{prooftree}
		\end{equation*}	
		
		\textbf{Case}: $\xi_{i}$ is of the form $\langle {\forall x: \form}, t \rangle$.
		We replace every bud with label $\sregseq$ such that ${\forall x: \form} \in \Gamma$ with the derivation
		\begin{equation*}
			\begin{prooftree}
				\hypo{\defn, \Gamma, \form[t/x] \vdash \Delta}
				\infer1[\textup{($\forall$L)}]{\sregseq}
			\end{prooftree}
		\end{equation*}	
		and every bud with label $\sregseq$ such that ${\forall x: \form} \in \Delta$ with the derivation
		\begin{equation*}
			\begin{prooftree}
				\hypo{\defn, \Gamma \vdash \form[z/x], \Delta}
				\infer1[\textup{($\forall$R)}]{\sregseq}
			\end{prooftree}
		\end{equation*}	
		where $z$ is an object symbol in $\varset$ that does not occur freely in $\defn$, $\Gamma$ or $\Delta$.
		
		\textbf{Case}: $\xi_{i}$ is of the form $\langle {\exists x: \form}, t \rangle$.
		This case is similar to the previous case.
		
		\textbf{Case}: $\xi_{i}$ is of the form $\langle {t=s}, x, y, \Gamma, \Delta \rangle$.
		We replace every bud with label of the form $\defn, \Gamma', {t=s} \vdash \Delta'$ with the derivation\footnote{
			For this to be a correct instance of ($=$L), we assume that $x$ and $y$ do not occur freely in $\defn$, $\Gamma'$ and $\Delta'$.
			If this is not the case, we can replace $x$ and $y$ with other object symbols in $\varset$ that do not occur freely in $\defn$, $\Gamma'$ and $\Delta'$ (which exist since these sets are finite) and modify $\Gamma$ and $\Delta$ accordingly.
		}
		\begin{equation*}
			\begin{prooftree}
				\hypo{\defn, \Gamma', \Gamma''[s/x,t/y], t=s \vdash \Delta', \Delta''[s/x,t/y]}
				\infer1[\textup{($=$L)}]{\defn, \Gamma', t=s \vdash \Delta'}
			\end{prooftree}
		\end{equation*}	
		where $\Gamma''$ consists of all formulae $\form \in \Gamma$ such that $\form[t/x,s/y] \in \Gamma'$ and $\Delta''$ consists of all formulae $\form \in \Delta$ such that $\form[t/x,s/y] \in \Delta'$.
		
		\textbf{Case}: $\xi_{i}$ is of the form $\langle P(\bar{t}[\bar{s}/\bar{x}]), \form[\bar{s}/\bar{x}] \rangle$.
		We replace every bud with label $\sregseq$ such that $P(\bar{t}[\bar{s}/\bar{x}]) \in \Gamma$ with the derivation
		\begin{equation*}
			\begin{prooftree}
				\hypo{\textup{case distinctions}}
				\infer1[\textup{(case)}]{\sregseq}
			\end{prooftree}
		\end{equation*}	
		For every rule $\forall \bar{y}: P(\bar{v}) \rul \formtwo$ in $\defn$ defining $P$, this rule has a premise $\defn, \Gamma, \bar{t}[\bar{s}/\bar{x}] = \bar{v}[\bar{z}/\bar{y}], \formtwo[\bar{z}/\bar{y}] \vdash \Delta$,
		where $\bar{z}$ is a tuple of object symbols in $\varset$ that do not occur freely in $\defn$, $\Gamma$ or $\Delta$, with the same length as $\bar{y}$.
		Furthermore, we replace every bud with label $\sregseq$ such that $P(\bar{t}[\bar{s}/\bar{x}]) \in \Delta$ with the derivation
		\begin{equation*}
			\begin{prooftree}
				\hypo{\defn, \Gamma \vdash \form[\bar{s}/\bar{x}], \Delta}
				\infer1[\textup{(def R)}]{\sregseq}
			\end{prooftree}
		\end{equation*}	
		
		Since for every $i \in \nat$, $\deriv_{i+1}$ extends $\deriv_i$, the sequence $(\deriv_i)_i$ has a well-defined limit $\schtree$,\footnote{
			The existence of a limit follows from a generalization of the Knaster-Tarski fixpoint theorem \cite{au/Markowsky76}, as the extension relation defines a \emph{chain-complete partial order} on the set of \scinf-derivations.
		} 
		which is also an \scinf-derivation of $\sregfix$.
		We call $\schtree$ the \emph{search tree} for $\sregfix$. %
	\end{definition}
	
	By construction, the search tree $\schtree = (N, E, s, \seqfun, \rulfun)$ for $\sregseq$ %
	has the property that for all $(n,n') \in E$: if $\seqfun(n) \doteq \sregseq$ and $\seqfun(n') \doteq \defn, \Gamma' \vdash \Delta'$, then $\Gamma \subseteq \Gamma'$ and $\Delta \subseteq \Delta'$.
	Furthermore, $\schtree$ has no buds, meaning that it is an \scinf-pre-proof.
	If $\schtree$ is moreover an \scinf-proof, then there is nothing left to prove.
	Thus, we can assume for the remainder of the proof that $\schtree$ is not an \scinf-proof.
	This means that there exists a branch in $\schtree$ along which no infinitely progressing trace exists.
	We fix such a branch $(n_i)_i$ and refer to it as the \emph{untraceable branch} $\untbranch$.
	For every $i$, we write $\seqfun(n_i) \doteq \defn, \Gamma_i \vdash \Delta_i$.
	We define the sets $\Gamma_\omega \coloneq \bigcup_{i} \Gamma_i$ and $\Delta_\omega \coloneq \bigcup_{i} \Delta_i$, and refer to $\limseq$ as the \emph{limit sequent}.\footnote{
		Since $\Gamma_\omega$ and $\Delta_\omega$ contain infinitely many formulae, the limit sequent is strictly speaking not a sequent.
	}	
	The sets $\Gamma_\omega$ and $\Delta_\omega$ form a partitioning of the set of \fo-formulae over $\vocab$:
	
	\begin{proposition} \label{prop:limit-sequent-partition}
		For every \fo-formula $\form$ over $\vocab$, either $\form \in \Gamma_\omega$ or $\form \in \Delta_\omega$, but not both.
	\end{proposition}
	
	\begin{proof}
		Let $\form$ be an \fo-formula over $\vocab$ 
		There exists a $j$ such that $\cut{\form}$ occurs as the $j$-th schedule element $\xi_j$ in $\sched$.
		Write $\untbranch =  (n_i)_i$ and $\seqfun(n_i) = \defn, \Gamma_i \vdash \Delta_i$ for every $i$. %
		By construction of $\schtree$, the sequent $\seqfun(n_{j+1})$ is of the form $\defn, \Gamma_j \vdash \form, \Delta_j$ or of the form $\defn, \Gamma_j, \form \vdash \Delta_j$.
		Thus, $\form \in \Gamma_{j+1}$ or $\form \in \Delta_{j+1}$, and hence, $\form \in \Gamma_\omega$ or $\form \in \Delta_\omega$.
		The formula $\form$ cannot belong to both $\Gamma_\omega$ and $\Delta_\omega$, for otherwise, there would be an $i$ such that $\form \in \Gamma_i \cap \Delta_i$.
		However, by construction of $\schtree$, this would imply that $\untbranch$ were finite, which is a contradiction.
	\end{proof}
	
	We proceed by constructing a countermodel $\cntmod$ of $\limseq$, which is also a countermodel of $\sregfix$, as $\Gamma_0 \subseteq \Gamma_\omega$ and $\Delta_0 \subseteq \Delta_\omega$.
	The domain of $\cntmod$ consists of all equivalence classes $[t]$ of terms $t \in \terms$ under an equivalence relation $\eqrel$ over $\terms$:
	
	\begin{definition} \label{def:eqrel}
		For all $t,s \in \terms$, we let $t \eqrel s$ iff ${t=s} \in \Gamma_\omega$.
	\end{definition}
	
	\begin{lemma} \label{lem:interchange-equivalent-terms}
		Let $x$ and $y$ be distinct object symbols in $\varset$ and $\form$ an \fo-formula over $\vocab$.
		Let $t$ and $s$ be terms such that $t \eqrel s$.
		If $\form[t/x, s/y] \in \Gamma_\omega$, then $\form[s/x, t/y] \in \Gamma_\omega$, and if $\form[t/x, s/y] \in \Delta_\omega$, then $\form[s/x, t/y] \in \Delta_\omega$.
	\end{lemma}
	
	\begin{proof}
		Suppose that $\form[t/x, s/y] \in \Gamma_\omega$.
		Since $t \eqrel s$, we have that ${t = s} \in \Gamma_\omega$, and hence, there exists an $i_0$ such that $t=s \in \Gamma_i$ for all $i \geq i_0$.
		Consider the schedule element $\langle {t = s}, x, y, \{\form\}, \emptyset \rangle$.
		Since it occurs infinitely often on the schedule $\sched$, there exists an $i \geq i_0$ such that $\xi_i = \langle {t = s}, x, y, \{\form\}, \emptyset \rangle$.
		By construction of $\schtree$, the $i$-th transition on the untraceable branch $\untbranch$ is of the following form:
		\begin{equation*}
			\begin{prooftree}
				\hypo{\defn, \Gamma, \form[s/x, t/y], {t=s} \vdash \Delta}
				\infer1[($=$L)]{\defn, \Gamma, {t=s} \vdash \Delta}
			\end{prooftree}
		\end{equation*}
		Thus, $\form[s/x, t/y] \in \Gamma_{i+1}$, which implies that $\form[s/x, t/y] \in \Gamma_\omega$.
		The proof for $\Delta_\omega$ is similar.
	\end{proof}
	
	\begin{proposition} \label{prop:equiv-relation}
		For all $t,s,u, t_1, \dots, t_k, s_1, \dots, s_k \in \terms$ and all $k$-ary function symbols $f$ in $\vocab$:
		\begin{enumerate}%
			\item $t \eqrel t$; \label{it:refl} %
			\item if $t \eqrel s$, then $s \eqrel t$; \label{it:symm} %
			\item if $t \eqrel s$ and $s \eqrel u$, then $t \eqrel u$; \label{it:trans} %
			\item if $t_1 \eqrel s_1, \dots, t_k \eqrel s_k$, then $f(t_1, \dots, t_k) \eqrel f(s_1, \dots, s_k)$. \label{it:cong} %
		\end{enumerate}
	\end{proposition}
	
	\begin{proof}
		Let $t,s,u, t_1, \dots, t_k, s_1, \dots, s_k \in \terms$ and $f$ a $k$-ary function symbol in $\vocab$.
		\begin{enumerate}
			\item By construction of $\schtree$, ${t=t}$ cannot be in $\Delta_\omega$, for otherwise, $\untbranch$ would be finite.
			Therefore, by Proposition \ref{prop:limit-sequent-partition}, ${t=t} \in \Gamma_\omega$, and by Definition \ref{def:eqrel}, $t \eqrel t$.
			\item Suppose that $t \eqrel s$.
			Then ${t=s} \in \Gamma_\omega$.
			By Lemma \ref{lem:interchange-equivalent-terms}, ${s=t} \in \Gamma_\omega$, and hence,
			$s \eqrel t$.
			\item Suppose that $t \eqrel s$ and $s \eqrel u$.
			Then ${t=s} \in \Gamma_\omega$. %
			By %
			Lemma \ref{lem:interchange-equivalent-terms}, ${t=u} \in \Gamma_\omega$.
			Thus, $t \eqrel u$.
			\item Suppose that $t_1 \eqrel s_1, \dots, t_k \eqrel s_k$. %
			By rule \ref{it:refl}, ${f(t_1, \dots, t_k) = f(t_1, \dots, t_k)} \in \Gamma_\omega$.
			By repeated application of Lemma \ref{lem:interchange-equivalent-terms}, ${f(t_1, \dots, t_k) = f(s_1, \dots, s_k)} \in \Gamma_\omega$.
			Hence, $f(t_1, \dots, t_k) \eqrel f(s_1, \dots, s_k)$. \qedhere
		\end{enumerate}
	\end{proof}
	
	By rules \ref{it:refl}-\ref{it:trans} of Proposition \ref{prop:equiv-relation}, $\eqrel$ is an equivalence relation on $\terms$.
	We denote the equivalence class of a term $t \in \terms$ under $\eqrel$ by $[t]$.
	We write $(t_1, \dots, t_k) \eqrel (s_1, \dots, s_k)$ for $t_1 \eqrel s_1, \dots, t_k \eqrel s_k$, 
	and $[(t_1, \dots, t_k)]$ for $([t_1], \dots, [t_k])$.
	
	\begin{definition}%
		\label{def:countermodel}
		We define $\cntmod$ as the structure with $\voc{\cntmod} = \vocab$, $\dom{\cntmod} = \{ [t] \mid t \in \terms \}$, 
		and such that:
		\begin{itemize}
			\item for every $k$-ary function symbol $f$ in $\vocab$: $f^{\cntmod}([t_1], \dots, [t_k]) = [f(t_1, \dots, t_k)]$ for all $t_1, \dots, t_k \in \terms$;\footnote{
				Note that for every function symbol $f$ in $\vocab$, $f^{\cntmod}$ is well-defined by rule \ref{it:cong} in Proposition \ref{prop:equiv-relation}.
						} 
			\item for every $k$-ary predicate symbol $P$ in $\vocab$: $P^{\cntmod} = \{ ([t_1], \dots, [t_k]) \mid P(t_1, \dots, t_k) \in \Gamma_\omega \}$.
		\end{itemize}
	\end{definition}
	
	In what follows, we will show that $\cntmod$ is a countermodel of $\limseq$.
	In Proposition \ref{prop:cnt-FO-forms}, we show that $\cntmod$ satisfies every formula in $\Gamma_\omega$ and no formula in $\Delta_\omega$, and in Proposition \ref{prop:cnt-models-defs}, we show that $\cntmod$ satisfies the definition $\defn$.
	
	\begin{proposition} \label{prop:t^I=[t]}
		$t^{\cntmod} = [t]$ for all $t \in \terms$.
	\end{proposition}
	
	\begin{proof}
		We show that $t^{\cntmod} = [t]$ for all $t \in \terms$ by structural induction on $t$.
		Since terms are composed exclusively of function symbols, the only case is where $t$ is of the form $f(t_1, \dots, t_k)$ for a $k$-ary function symbol $f$ and terms $t_1, \dots, t_k$.\footnote{
			We do not need a separate case for object symbols, since we defined them as $0$-ary function symbols.
		}
		By induction hypothesis, $t_i^{\cntmod} = [t_i]$ for every $i \in \{ 1, \dots, k \}$.
		Using Definition \ref{def:countermodel}, we obtain that $f(t_1, \dots, t_k)^{\cntmod} = f^{\cntmod}(t_1^{\cntmod}, \dots, t_k^{\cntmod}) = f^{\cntmod}([t_1], \dots, [t_k]) = [f(t_1, \dots, t_k)]$, which finishes the proof.
	\end{proof}
	
	\begin{proposition} \label{prop:cnt-FO-forms}
		Let $\form$ be an \fo-formula over $\vocab$.
		If $\form \in \Gamma_\omega$, then $\cntmod \models \form$, and if $\form \in \Delta_\omega$, then $\cntmod \not\models \form$.
	\end{proposition}
	
	\begin{proof}
		We prove this property by structural induction on $\form$.
		
		\textbf{Case}: $\form$ is of the form $t=s$.
		Suppose that $t=s \in \Gamma_\omega$. 
		Then $t \eqrel s$ by Definition \ref{def:eqrel}. 
		Consequently, $t^{\cntmod} = [t] = [s] = s^{\cntmod}$, and hence $\cntmod \models t=s$.
		
		Now suppose that $t=s \in \Delta_\omega$. 
		Then $t=s \notin \Gamma_\omega$ by Proposition \ref{prop:limit-sequent-partition}, and therefore, $t \not\eqrel s$ by Definition \ref{def:eqrel}. 
		Consequently, $t^{\cntmod} = [t] \neq [s] = s^{\cntmod}$, and hence $\cntmod \not\models t=s$.
		
		\textbf{Case}: $\form$ is of the form $P(\bar{t})$.
		Suppose that $P(\bar{t}) \in \Gamma_\omega$. 
		Then $[\bar{t}] \in P^{\cntmod}$ by Definition \ref{def:countermodel}, and hence, $\cntmod \models P(\bar{t})$.
		
		Now suppose that $P(\bar{t}) \in \Delta_\omega$, and suppose for contradiction that $\cntmod \models P(\bar{t})$.
		Then $[\bar{t}] \in P^{\cntmod}$ by Definition \ref{def:countermodel}.
		This means that there exists an $\ar{P}$-tuple of terms $\bar{s}$ such that $\bar{s} \eqrel \bar{t}$ and $P(\bar{s}) \in \Gamma_\omega$.
		By repeated application of Lemma \ref{lem:interchange-equivalent-terms}, $P(\bar{t}) \in \Gamma_\omega$.
		However, this contradicts Proposition \ref{prop:limit-sequent-partition}.
		
		\textbf{Case}: $\form$ is of the form $\lnot \formtwo$.
		Suppose that $\lnot \formtwo \in \Gamma_\omega$. 
		By construction of $\schtree$ and definition of the limit sequent $\limseq$, it follows that $\formtwo \in \Delta_\omega$.
		Indeed, as $\lnot \formtwo \in \Gamma_\omega$, there exists a $j$ such that $\lnot \formtwo \in \Gamma_j$.
		Since $\lnot \formtwo$ occurs infinitely often on the schedule $\sched$, it occurs as the $i$-th schedule element $\xi_i$ for some $i \geq j$.
		By construction of $\schtree$, %
		the set $\Delta_{i+1}$ contains $\formtwo$, and therefore, we find that indeed $\formtwo \in \Delta_\omega$.
		By induction hypothesis, it follows that $\cntmod \not\models \formtwo$.
		Consequently, $\cntmod \models \lnot \formtwo$.
		
		Now suppose that $\lnot \formtwo \in \Delta_\omega$.
		Then by a similar reasoning as before, $\formtwo \in \Gamma_\omega$. %
		By induction hypothesis, it follows that $\cntmod \models \formtwo$, and hence, $\cntmod \not\models \lnot \formtwo$.
		
		\textbf{Case}: $\form$ is of the form $\formtwo \land \formthree$.
		Suppose that $\formtwo \land \formthree \in \Gamma_\omega$. 
		Then $\formtwo,\formthree \in \Gamma_\omega$ by construction of $\schtree$. 
		By induction hypothesis, $\cntmod \models \formtwo$ and $\cntmod \models \formthree$.
		Hence, $\cntmod \models \formtwo \land \formthree$.
		
		Now suppose that $\formtwo \land \formthree \in \Delta_\omega$.
		Then $\formtwo \in \Delta_\omega$ or $\formthree \in \Delta_\omega$ by construction of $\schtree$. 
		By induction hypothesis, $\cntmod \not\models \formtwo$ or $\cntmod \not\models \formthree$.
		Hence, $\cntmod \not\models \formtwo \land \formthree$.
		
		\textbf{Cases}: $\form$ is of the form $\formtwo \lor \formthree$ or $\formtwo \Rightarrow \formthree$. %
		These cases are similar to the previous case.
		
		\textbf{Case}: $\form$ is of the form $\forall x: \formtwo$.
		Suppose that ${\forall x: \formtwo} \in \Gamma_\omega$, and pick an arbitrary term $t \in \terms$.
		By construction of $\schtree$, we have that $\formtwo[t/x] \in \Gamma_\omega$.
		Indeed, since $\forall x: \formtwo \in \Gamma_\omega$, there exists a $j$ %
		such that $\forall x: \formtwo \in \Gamma_j$. Since $\langle \forall x: \formtwo, t \rangle$ occurs infinitely often on the schedule $\sched$, it occurs as the $i$-th schedule element of $\sched$ for some $i \geq j$.
		By construction of $\schtree$, %
		the set $\Gamma_{i+1}$ contains $\formtwo[t/x]$.
		Thus, we indeed find that $\formtwo[t/x] \in \Gamma_\omega$.
		By induction hypothesis, $\cntmod \models \formtwo[t/x]$. %
		Since $t^{\cntmod} = [t]$ by Proposition \ref{prop:t^I=[t]}, this implies that $\cntmod[x:[t]] \models \formtwo$.
		Since we picked $t \in \terms$ arbitrarily, and since any domain element of $\cntmod$ is of the form $[t]$ for some term $t$, this shows that $\cntmod \models {\forall x: \formtwo}$.
		
		Now suppose that ${\forall x: \formtwo} \in \Delta_\omega$.
		By construction of $\schtree$, the set $\Delta_\omega$ contains $\formtwo[z/x]$ for some object symbol $z \in \varset$. 
		By induction hypothesis, $\cntmod \not\models \formtwo[z/x]$.
		Since $z^{\cntmod} = [z]$, %
		this implies that $\cntmod[x:[z]] \not\models \formtwo$.
		Consequently, $\cntmod \not\models {\forall x: \formtwo}$.
		
		\textbf{Case}: $\form$ is of the form $\exists x: \formtwo$. 
		This case is similar to the previous case.
	\end{proof}
	
	\begin{proposition} \label{prop:cnt-models-defs}
		$\cntmod \modelswf \defn$.
	\end{proposition}
	
	\begin{proof}
		Let $\vocab \coloneq \sym{\defn}$, $D \coloneq \dom{\cntmod}$ and $\cont \coloneq \cntmod|_{\pars{\defn}}$, and let  $\consts{D}$ be as in Definition \ref{def:fact}.
		Using Proposition \ref{prop:wf-sem-just}, we show that $\cntmod \modelswf \defn$ by constructing a good justification w.r.t.~$\defn$ and $\cont$ of every atomic fact $(P(\bar{a}), v)$ such that $P(\bar{a})$ is a defined atom of $\defn$ and $v = \true$ iff $\bar{a} \in P^{\cntmod}$.
		To accomplish this, we more generally construct a good justification w.r.t.~$\defn$ and $\cont$ of any fact $(\form, v)$ w.r.t.~$\vocab$ and $D$ such that %
		$v = \true$ iff $\cntmod^{\consts{D}} \models \form$.
		Here $\consts{D}$ is a set of distinguished constant symbols $c_a$ for every $a \in D$, as defined in Definition \ref{def:fact}, and $\cntmod^{\consts{D}}$ is the expansion of $\cntmod$ that interprets every object symbol $c_a \in \consts{D}$ as $a$, defined similarly to $\cont^{\consts{D}}$ in Definition \ref{def:just-derives}.
		
		Fix a fact $(\form_0, v_0)$ of the above form.
		We construct a good justification of $(\form_0, v_0)$ w.r.t.~$\defn$ and $\cont$ as the limit $J_\omega$ of a sequence of justifications $(J_i)_{i}$ of $(\form_0, v_0)$ w.r.t.~$\defn$ and $\cont$ such that for every $i$, $J_{i+1}$ extends $J_i$.\footnote{
			We say that a justification $J = (N, E, r, \labfun)$ \emph{extends} a justification $J' = (N', E', r', \labfun')$ if $N' \subseteq N$, $E' = E \cap (N' \times N')$, $r' = r$ and $\labfun' = \labfun|_{N'}$.
		}
		By abuse of notation, we will often identify \fo-formulae over $\vocab$ with \fo-formulae over $\vocab^{\consts{D}}$ by identifying terms $t$ over $\vocab$ with object symbols $c_{[t]}$ in $\vocab^{\consts{D}}$.
		To guarantee that $J_\omega$ is a good justification, we also construct, for any branch $\branch$ in any of the justifications $J_i$, a corresponding finite subpath $\subpthcor{\branch}$ of the untraceable branch $\untbranch = (n_i)_i$, as well as a trace $\tsubpthcor{\branch}$ along $\subpthcor{\branch}$.
		We construct these subject to the following invariant:
		\begin{enumerate}
			\item If a branch $\branch$ ends with a fact $(\form, \true)$, then $\trace^\branch_{e_\branch} = \form$ and $\form \in \Gamma_{e_\branch}$.
			\item If a branch $\branch$ ends with a fact $(\form, \false)$, then $\trace^\branch_{e_\branch} = \form$ and $\form \in \Delta_{e_\branch}$. \label{it:const-neg}
		\end{enumerate}
		We will use these paths and traces to construct, for every infinite branch $\branch$ in $J_\omega$, a corresponding trace $(\trace_i)_i$ along a tail of $\untbranch$, and use the fact that $(\trace_i)_i$ has only finitely many progression points to argue that $\branch$ is negative.
		
		The first justification $J_0$ of the sequence $(J_i)_{i}$ consists of a single node $n$ with label $(\form_0, v_0)$.
		This justification has a single branch $\branch$ consisting of the single node $n$, for which we need to construct a subpath $\subpthcor{\branch}$ of $\untbranch$ and a trace $\tsubpthcor{\branch}$ along $\subpthcor{\branch}$.
		Since, by assumption, $v_0 = \true$ iff $\cntmod^{\consts{D}} \models \form_0$, it follows that if $v_0 = \true$, then $\form_0 \in \Gamma_\omega$, and if $v_0 = \false$, then $\form_0 \in \Delta_\omega$ (using Proposition \ref{prop:cnt-FO-forms}).
		Let $j$ be such that $\form_0 \in \Gamma_j \cup \Delta_j$.
		We let $\subpthcor{\branch}$ be the single-noded subpath $(n_j)$ of $\untbranch$ (stated differently, we let $s_\branch = e_\branch = j$), 
		and we let $\subpthcor{\pthjust}$ be the trace $(\form_0)$ along $(n_j)$.
		These choices entail that the invariant holds for $J_0$.
		
		Suppose we have constructed the justification $J_i$ for a certain $i$, as well as subpaths $\subpthcor{\branch}$ of $\untbranch$ and traces $\tsubpthcor{\branch}$ along $\subpthcor{\branch}$ corresponding to every branch $\branch$ in $J_i$.
		Then we construct $J_{i+1}$ by ``expanding'' all buds of $J_i$, and extend the paths $\subpthcor{\branch}$ and traces $\tsubpthcor{\branch}$ accordingly.
		More concretely, for every bud $n$ in $J_i$, we do the following. Write $\labfun(n) = (\form, v)$. 
		Let $\branch$ be the unique branch in $J_i$ that ends with $n$. 
		Let $\subpth$ be the corresponding subpath of $\untbranch$ and $\tsubpth$ the trace along $\subpth$ corresponding to $\branch$.
		We proceed based on the form of $(\form, v)$:
		\begin{itemize}
			\item Suppose that $\form$ is an atom and $v = \true$.
			Since $n$ is a bud, $\form$ must be a defined atom $P(\bar{t})$ of $\defn$ (by Definition \ref{def:justification}). 
			By the invariant on $J_i$, we have that $P(\bar{t}) \in \Gamma_{e}$.
			Let $r \geq e$ be minimal such that the $r$-th schedule element $\xi_r$ is of the form $\langle P(\bar{t}), \formtwo \rangle$.
			Then, by construction of $\schtree$, the $r$-th transition in $\untbranch$ corresponds to an instance of (case) with active formula $P(\bar{t})$.
			The $(r+1)$-th node $n_{r+1}$ in $\untbranch$ is labeled with a premise $\defn, \Gamma_r, \formthree \vdash \Delta_r$ of this instance, where $\formthree$ derives $P(\bar{t})$ w.r.t.~$\defn$ and $\cont$.
			We add a child $m$ to $n$ with label $(\formthree, \true)$. 
			Furthermore, we extend the subpath $\subpth$ of $\untbranch$ by concatenating it with the subpath $\concsubpth$ of $\untbranch$, and we extend the corresponding trace $\tsubpth$ by concatenating it with $(\trace_{e}, \dots, \trace_{e}, \formthree)$.\footnote{
				Since we took $r \geq e$ minimal such that the $r$-th schedule element $\xi_r$ is of the form $\langle P(\bar{t}), \formtwo \rangle$, $(\trace_{e}, \dots, \trace_{e}, \formthree)$ is a trace along $\tsubpth$.
			}
			By construction, this operation preserves the invariant.
			
			\item Suppose that $\form$ is an atom and $v = \false$.
			Again, since $n$ is a bud, $\form$ must be a defined atom $P(\bar{t})$ of $\defn$. 
			For every formula $\formtwo$ that derives $P(\bar{t})$ w.r.t.~$\defn$ and $\cont$, we add a child $m$ with label $(\formtwo, \false)$ to $n$.
			By the invariant on $J_i$, $P(\bar{t}) \in \Delta_{e}$.
			Fix a formula $\formtwo$ that derives $P(\bar{t})$ w.r.t.~$\defn$ and $\cont$, and let $m$ be the corresponding child of $n$.
			Let $\branch'$ be the unique branch in $J_{i+1}$ that ends with $m$.
			We obtain a subpath of $\untbranch$ and a trace along this subpath corresponding to $\branch'$ by extending the path $\subpth$ and the trace $\tsubpth$, respectively.
			Let $r \geq e$ be minimal such that the $r$-th schedule element $\xi_r$ is $\langle P(\bar{t}), \formtwo \rangle$.
			Then, by construction of $\schtree$, the $r$-th transition in $\untbranch$ corresponds to an instance of (def R) with active formula $P(\bar{t})$ and auxiliary formula $\formtwo$.
			This entails that $\formtwo \in \Delta_{r+1}$.
			We extend $\subpth$ by concatenating it with $\concsubpth$ and extend $\tsubpth$ by concatenating it with $(P(\bar{t}), \dots, P(\bar{t}), \formtwo)$.
			By construction, this operation preserves the invariant.
			
			\item Suppose that $\form$ is of the form $\lnot \formtwo$ and $v = \true$.
			Then we add a child $m$ to $n$ with label $(\formtwo, \false)$.
			By the invariant on $J_i$, we have that $\lnot \formtwo \in \Gamma_{e}$.
			Let $\branch'$ be the unique branch in $J_{i+1}$ that ends with $m$.
			We obtain a subpath of $\untbranch$ and a trace along this subpath corresponding to $\branch'$ by extending the path $\subpth$ and the trace $\tsubpth$, respectively.
			Let $r \geq e$ be minimal such that $\lnot \formtwo$ occurs as the $r$-th schedule element $\xi_{r}$.
			Then the $r$-th transition in $\untbranch$ corresponds to an instance of ($\lnot$L) with active formula $\lnot \formtwo$ and auxiliary formula $\formtwo$.
			Hence, $\formtwo \in \Delta_{r+1}$.
			We extend $\subpth$ by concatenating it with $\concsubpth$ and extend $\tsubpth$ by concatenating it with $(\lnot \formtwo, \dots, \lnot \formtwo, \formtwo)$.
			By construction, this operation preserves the invariant. 
			
			\item Suppose that $\form$ is of the form $\lnot \formtwo$ and $v = \false$.
			Then we add a child $m$ to $n$ with label $(\formtwo, \true)$.
			By the invariant on $J_i$, we have that $\lnot \formtwo \in \Delta_{e}$.
			Let $\branch'$ be the unique branch in $J_{i+1}$ that ends with $m$.
			We obtain a subpath of $\untbranch$ and a trace along this subpath corresponding to $\branch'$ by extending the path $\subpth$ and the trace $\tsubpth$, respectively.
			Let $r \geq e$ be minimal such that $\lnot \formtwo$ occurs as the $r$-th schedule element $\xi_{r}$.
			Then the $r$-th transition in $\untbranch$ corresponds to an instance of ($\lnot$R) with active formula $\lnot \formtwo$ and auxiliary formula $\formtwo$.
			Hence, $\formtwo \in \Gamma_{r+1}$.
			We extend $\subpth$ by concatenating it with $\concsubpth$ and extend $\tsubpth$ by concatenating it with $(\lnot \formtwo, \dots, \lnot \formtwo, \formtwo)$.
			By construction, this operation preserves the invariant. 
	
			\item Suppose that $\form$ is of the form $\formtwo \land \formthree$ and $v = \true$. 
			Then we add two children $m_1$ and $m_2$ to $n$ with the respective labels $(\formtwo, \true)$ and $(\formthree, \true)$.
			By the invariant on $J_i$, we have that $\formtwo \land \formthree \in \Gamma_{e}$.
			Let $\branch_1$ be the unique branch in $J_{i+1}$ that ends with $m_1$, and let $\branch_2$ the unique branch in $J_{i+1}$ that ends with $m_2$.
			We obtain subpaths of $\untbranch$ and traces along these subpaths corresponding to $\branch_1$ and $\branch_2$ by extending the path $\subpth$ and the trace $\tsubpth$, respectively.
			Let $r \geq e$ be minimal such that $\formtwo \land \formthree$ occurs as the $r$-th schedule element $\xi_{r}$.
			Then the $r$-th transition in $\untbranch$ corresponds to an instance of ($\land$L) with active formula $\formtwo \land \formthree$ and auxiliary formulae $\formtwo$ and $\formthree$.
			Hence, $\formtwo \in \Gamma_{r+1}$ and $\formthree \in \Gamma_{r+1}$.
			For $\branch_1$, we extend $\subpth$ by concatenating it with $\concsubpth$ and extend $\tsubpth$ by concatenating it with $(\formtwo \land \formthree, \dots, \formtwo \land \formthree, \formtwo)$.
			Similarly, for $\branch_2$, we extend $\subpth$ by concatenating it with $\concsubpth$ and extend $\tsubpth$ by concatenating it with $(\formtwo \land \formthree, \dots, \formtwo \land \formthree, \formthree)$.
			By construction, this operation preserves the invariant. 
			
			\item Suppose that $\form$ is of the form $\formtwo \land \formthree$ and $v = \false$.
			By the invariant on $J_i$, we have that $\formtwo \land \formthree \in \Delta_{e}$.
			Let $r \geq e$ be minimal such that $\formtwo \land \formthree$ occurs as the $r$-th schedule element $\xi_{r}$ in $\sched$.
			Then the $r$-th transition in $\untbranch$ corresponds to an instance of ($\land$R) with active formula $\formtwo \land \formthree$.
			Then $\formtwo \in \Delta_{r+1}$ or $\formthree \in \Delta_{r+1}$.
			In the first case, we add a child $m$ to $n$ with label $(\formtwo, \false)$ and in the second case, we add a child $m$ to $n$ with label $(\formthree, \false)$.
			Let $\branch'$ be the unique branch in $J_{i+1}$ that ends with $m$.
			We obtain a subpath of $\untbranch$ and a trace along this subpath corresponding to $\branch'$ by extending the path $\subpth$ and the trace $\tsubpth$, respectively.
			We do so by concatenating $\subpth$ with $\concsubpth$ and $\tsubpth$ with $(\formtwo \land \formthree, \dots, \formtwo \land \formthree, \formtwo)$ or $(\formtwo \land \formthree, \dots, \formtwo \land \formthree, \formthree)$, depending on the case.
			By construction, this operation preserves the invariant. 
			
			\item Suppose that $\form$ is of the form $\forall x: \formtwo$ and $v = \true$.
			For every $[t] \in \dom{\cntmod}$, we add a child $m$ to $n$ with label $(\formtwo[t/x], \true)$.
			By the invariant on $J_i$, we have that $\forall x: \formtwo \in \Gamma_{e}$.
			Fix a $[t] \in \dom{\cntmod}$, and let $m$ be the child of $n$ with label $\formtwo[t/x]$.
			Let $\branch'$ be the unique branch in $J_{i+1}$ that ends with $m$.
			We obtain a subpath of $\untbranch$ and a trace along this subpath corresponding to $\branch'$ by extending the path $\subpth$ and the trace $\tsubpth$, respectively.
			Let $r \geq e$ be minimal such that $\langle \forall x: \formtwo, t \rangle$ occurs as the $r$-th schedule element $\xi_{r}$ in $\sched$.
			Then the $r$-th transition in $\untbranch$ corresponds to an instance of ($\forall$L) with active formula $\forall x: \formtwo$ and auxiliary formula $\formtwo[t/x]$.
			It follows that $\formtwo[t/x] \in \Gamma_{r+1}$.
			We extend $\subpth$ by concatenating it with $\concsubpth$ and extend $\tsubpth$ by concatenating it with $({\forall x: \formtwo}, \dots,{ \forall x: \formtwo}, \formtwo[t/x])$.
			By construction, this operation preserves the invariant. 
			
			\item Suppose that $\form$ is of the form $\forall x: \formtwo$ and $v = \false$.
			By the invariant on $J_i$, we have that $\forall x: \formtwo \in \Delta_{e}$.
			Let $r \geq e$ be minimal such that the $r$-th schedule element $\xi_{r}$ in $\sched$ is of the form $\langle \forall x: \formtwo, t \rangle$.
			Then the $r$-th transition in $\untbranch$ corresponds to an instance of ($\forall$R) with active formula $\forall x: \formtwo$ and auxiliary formula of the form $\formtwo[z/x]$.
			It follows that $\formtwo[z/x] \in \Delta_{r+1}$.
			We add a child $m$ to $n$ with label $(\formtwo[z/x], \false)$.
			Let $\branch'$ be the unique branch in $J_{i+1}$ that ends with $m$.
			We obtain a subpath of $\untbranch$ and a trace along this subpath corresponding to $\branch'$ by extending the path $\subpth$ and the trace $\tsubpth$, respectively.
			We do so by concatenating $\subpth$ with $\concsubpth$ and $\tsubpth$ with $({\forall x: \formtwo}, \dots, {\forall x: \formtwo}, \formtwo[z/x])$.
			By construction, this operation preserves the invariant.
			
			\item (The cases for $\lor$, $\Rightarrow$ %
			and $\exists$ can be derived from the cases above.)
		\end{itemize}
		
		Let $J_\omega$ be the limit of the sequence $(J_i)_i$.\footnote{
			As with the construction of the search tree $\schtree$ in Definition \ref{def:search-tree}, the existence of the limit $J_\omega$ follows from a generalization of the Knaster-Tarski fixpoint theorem \cite{au/Markowsky76}, since the extension relation defines a chain-complete partial order on the set of justifications.
		}
		Then $J_\omega$ is a justification of $(\form_0, v_0)$ w.r.t.~$\defn$ and $\cont$.
		Note that the atomic facts $(\form, v)$ labeled by leaves in $J_\omega$ have the property that $\cont^{\consts{D}} \models \form$ iff $v = \true$, which follows from the %
		invariant on the justifications $J_i$
		and Proposition \ref{prop:cnt-FO-forms}.
		By construction, $J_\omega$ has no buds.
		Hence, it is a locally complete justification.
		
		It remains to argue that $J_\omega$ is good.
		Let $\branch$ be an infinite branch in $J_\omega$, and denote by $\branch_j$ the subpath of $\branch$ in $J_\omega$ consisting of its first $j$ elements.
		We associate with $\branch$ an (infinite) tail $\pth$ of $\untbranch$ and a trace $(\trace_i)_i$ along $\pth$ obtained as the limits of the paths %
		$\subpthcor{\branch_j}$ and the traces $\tsubpthcor{\branch_j}$, respectively, for $j$ going to infinity.\footnote{
			Note that we constructed the subpaths %
			and traces 
			corresponding to the branches in $J_i$
			in such a way that for every $j$, $\subpthcor{\branch_{j+1}}$ extends $\subpthcor{\branch_{j}}$ and $\tsubpthcor{\branch_{j+1}}$ extends $\tsubpthcor{\branch_{j}}$.
			Hence, the existence of these limits follows from a generalization of the Knaster-Tarski fixpoint theorem \cite{au/Markowsky76}, since the extension relation defines a chain-complete partial order on the set of sequences.
		}
		Suppose for contradiction that $\branch$ is positive or mixed.
		Then $\branch$ contains infinitely many positive atomic facts $(P(\bar{t}), \true)$ with $P \in \defp{\defn}$.
		By construction, the trace $(\trace_i)_i$ %
		corresponding to $\branch$ has infinitely many progression points.
		However, this contradicts the defining property of $\untbranch$.
	\end{proof}
	
	We are now ready to prove Theorem \ref{thm:completeness}.
	
	\begin{proof}[Proof of Theorem \ref{thm:completeness}]
		Suppose that $\seqshort$ is a valid \foid-sequent.
		Let $\schtree$ be a search tree for $\seqshort$ (as defined in Definition \ref{def:search-tree}).
		If $\schtree$ had an untraceable branch, then there would exist a structure $\cntmod$ (as defined in Definition \ref{def:countermodel}) that is a countermodel of $\seqshort$ (as shown in Propositions \ref{prop:cnt-FO-forms} and \ref{prop:cnt-models-defs}).
		However, this would contradict the validity of $\seqshort$.
		Therefore, every infinite branch in $\schtree$ admits %
		an infinitely progressing trace, 
		which means that $\schtree$ is an \scinf-proof of $\seqshort$.
	\end{proof}

\section{Relation between \texorpdfstring{\scfoid}{SCFO(ID)} and \texorpdfstring{\sccyc}{SCFO(ID)-cyc}} \label{app:relation-calculi}

In this appendix, we show that every \scfoid-theorem is also an \sccyc-theorem, by extending Brotherston and Simpson's proof of the fact that every \lkid-theorem is a \clkid-theorem \cite{jlc/BrotherstonS11}.
The sequent calculus \scfoid \cite{lpnmr/VandenEedeVD24,arxiv} is defined similarly as \scinf, with the difference that proofs in \scfoid are {finite} trees, and that the left introduction rule for defined atoms is the \emph{induction rule} (ind) instead of the case distinction rule (case):
\begin{equation*}
	\begin{prooftree}
		\hypo{\text{minor premises}}
		\hypo{\defn, \Gamma, \ihf{P}{\bar{t}} \vdash \Delta}
		\infer2[(ind)]{\defn, \Gamma, P(\bar{t}) \vdash \Delta}
	\end{prooftree}
\end{equation*}
Here %
$P(\bar{t})$ is a defined atom of $\defn$.
As explained in \cite{arxiv}, an application of (ind) involves the selection of a subset $\Pi$ of $\defp{\defn}$ containing $P$, and for every $Q \in \Pi$ an \emph{induction hypothesis} $\ih_Q$ associated with $Q$.
An induction hypothesis $\ih_Q$ associated with $Q$ is a \emph{set expression} $\setexpg$, where $\bar{z}$ is an $\ar{Q}$-tuple of object symbols and $\form$ an \fo-formula.
Informally, $\ih_Q$ must be chosen such that it describes an upper bound on $Q$. %
Given an induction hypothesis $\ih_Q \doteq \setexpg$ and an $\ar{Q}$-tuple of terms $\bar{s}$, we write $\ihf{Q}{\bar{s}}$ to denote the formula $\form[\bar{s}/\bar{z}]$.
For every definitional rule $\forall \bar{y}: Q(\bar{s}) \rul \formtwo$ in $\defn$ with $Q \in \Pi$, the inference rule (def L) has a \emph{minor premise}
\[\defn, \Gamma, \formtwo\substposihpi \vdash \ihf{Q}{\bar{s}}, \Delta . \]
Here $\formtwo\substposihpi$ denotes the formula obtained from $\formtwo$ by replacing all positive occurrences of atoms $R(\bar{u})$ with $R \in \Pi$ by $\ihf{R}{\bar{u}}$.
This inference rule requires that for all definitional rules $\forall \bar{y}: Q(\bar{s}) \rul \formtwo$ in $\defn$ with $Q \in \Pi$, no object symbol in $\bar{y}$ occurs freely in $\defn$, $\Delta$ or $\Gamma$. 
The premise $\defn, \Gamma, \ihf{P}{\bar{t}} \vdash \Delta$ is called the \emph{major premise} of (ind).

The induction rule serves as a formalization of the \emph{principle of mathematical induction}, in which the elements of inductively defined sets are shown to satisfy a certain property by finding a suitable induction hypothesis that entails the property and that is preserved under the definitional rules. 
The former is formalized by the major premise, and the latter by the minor premises.
The set $\Pi$ contains the defined predicates involved in the induction argument.
To guarantee soundness for non-monotone definitions, we only replace positive occurrences of defined atoms in the minor premises.

To prove Theorem \ref{thm:ind-thm-is-cyc-thm}, we show that every instance of (ind) is derivable in \sccyc.

\IndDerivableInCyc*

\begin{proof}
	Consider an instance of (ind):
	\begin{equation} \label{eq:instance-ind}
		\begin{prooftree}
			\hypo{\text{minor premises}}
			\hypo{\defn, \Gamma, \ihf{P}{\bar{v}} \vdash \Delta}
			\infer2[(ind)]{\defn, \Gamma, P(\bar{v}) \vdash \Delta}
		\end{prooftree}
	\end{equation}
	with the corresponding set $\Pi \subseteq \defp{\defn}$ and induction hypotheses $\ihs{Q}$ for all $Q \in \Pi$.
	For every definitional rule $\forall \bar{y}: Q(\bar{s}) \rul \formtwo$ in $\defn$ with $Q \in \Pi$, this instance has a minor premise $\defn, \Gamma, \formtwo\substposihpi \vdash \ihf{Q}{\bar{s}}, \Delta$.
	
	Let $\mset$ be the set of formulae $\forall \bar{y}: \formtwo\substposihpi \Rightarrow \ihf{Q}{\bar{s}}$ corresponding to the definitional rules $\forall \bar{y}: Q(\bar{s}) \rul \formtwo$ in $\defn$ with $Q \in \Pi$.
	For each $Q \in \Pi$, we fix an $\ar{Q}$-tuple of object symbols $\bar{z}_Q$.
	Consider the following \sccyc-derivation:
	\begin{equation} \label{eq:cyc-derivation-ind-rule}
		\begin{prooftree}
			\hypo{\defn, P(\bar{z}_P), \mset \vdash \ihf{P}{\bar{z}_P}}
			\infer1[(subst)]{\defn, P(\bar{v}), \mset \vdash \ihf{P}{\bar{v}}}
			\infer1[($\land$L)]{\vdots}
			\infer1[($\land$L)]{\defn, P(\bar{v}), \bigwedge \mset \vdash \ihf{P}{\bar{v}}}
			\infer1[($\Rightarrow$R)]{\defn, P(\bar{v}) \vdash \bigwedge \mset \Rightarrow \ihf{P}{\bar{v}}}
			\infer1[(wk)]{\defn, \Gamma, P(\bar{v}) \vdash \bigwedge \mset \Rightarrow \ihf{P}{\bar{v}}, \Delta}
			
			\hypo{\defn, \Gamma \vdash \bigwedge \mset, \Delta}
			\hypo{\defn, \Gamma, \ihf{P}{\bar{v}} \vdash \Delta}
			\infer2[($\Rightarrow$L)]{\defn, \Gamma, \bigwedge \mset \Rightarrow \ihf{P}{\bar{v}} \vdash \Delta}
			\infer1[(wk)]{\defn, \Gamma, P(\bar{v}), \bigwedge \mset \Rightarrow \ihf{P}{\bar{v}} \vdash \Delta}
			
			\infer2[(cut)]{\defn, \Gamma, P(\bar{v}) \vdash \Delta}
		\end{prooftree}
	\end{equation}
	Here $\bigwedge \mset$ denotes the conjunction of all formulae in $\mset$. %
	By repeated application of ($\land$R), the sequent $\defn, \Gamma \vdash \bigwedge \mset, \Delta$ can derived from sequents of the form $\defn, \Gamma \vdash \forall \bar{y}: \formtwo\substposihpi \Rightarrow \ihf{Q}{\bar{s}}, \Delta$, corresponding to the definitional rules $\forall \bar{y}: Q(\bar{s}) \rul \formtwo$ in $\defn$ with $Q \in \Pi$.
	For each such definitional rule, consider the following \sccyc-derivation:
	\begin{equation*}
		\begin{prooftree}
			\hypo{\defn, \Gamma, \formtwo\substposihpi \vdash \ihf{Q}{\bar{s}}, \Delta}
			\infer1[($\Rightarrow$R)]{\defn, \Gamma \vdash \formtwo\substposihpi \Rightarrow \ihf{Q}{\bar{s}}, \Delta}
			\infer1[($\forall$R)]{\vdots}
			\infer1[($\forall$R)]{\defn, \Gamma \vdash \forall \bar{y}: \formtwo\substposihpi \Rightarrow \ihf{Q}{\bar{s}}, \Delta}
		\end{prooftree}
	\end{equation*}
	Extending (\ref{eq:cyc-derivation-ind-rule}) with these derivations, we obtain an \sccyc-derivation of $\defn, \Gamma, P(\bar{v}) \vdash \Delta$ with buds labeled by the premises of (\ref{eq:instance-ind}), as well as $\defn, P(\bar{z}_P), \mset \vdash \ihf{P}{\bar{z}_P}$.
	Thus, it suffices to find an \sccyc-proof of $\defn, P(\bar{z}_P), \mset \vdash \ihf{P}{\bar{z}_P}$. 
	In the remainder of this proof, we will construct an \sccyc-proof $\prf = (\deriv, \repfun)$ of $\defn, P(\bar{z}_P), \mset \vdash \ihf{P}{\bar{z}_P}$.
	More specifically, we will construct $\deriv$ as the final derivation $\deriv_k$ of a finite sequence $\deriv_0, \dots, \deriv_k$ of \sccyc-derivations of $\defn, P(\bar{z}_P), \mset \vdash \ihf{P}{\bar{z}_P}$ and construct the repetition function $\repfun$ alongside the sequence $\deriv_0, \dots, \deriv_k$.
	Furthermore, we will construct the derivations $\deriv_0, \dots, \deriv_k$ subject to the invariant that every bud of $\deriv_i$ is labeled with a sequent of the form $\defn, \mset, \form \vdash \form\substposihpi$ or $\defn, \mset, \form\substnegihpi \vdash \form$, for an \fo-formula $\form$. %
	Here $\form\substnegihpi$ denotes the formula obtained from $\form$ by replacing all negative occurrences of atoms $R(\bar{u})$ with $R \in \Pi$ by $\ihf{R}{\bar{u}}$.
	
	The first derivation $\deriv_0$ derives the sequent $\defn, \mset, P(\bar{z}_P) \vdash \ihf{P}{\bar{z}_P}$ via (case) from the sequents $\defn, \mset, \bar{z}_P=\bar{t}, \form \vdash \ihf{P}{\bar{z}_P}$ corresponding to the definitional rules $\defrul$ in $\defn$ defining $P$, and every such sequent is derived as follows:
	\begin{equation*}
		\begin{prooftree}
			\hypo{\defn, \mset, \form \vdash \form\substposihpi}
			\infer1[(wk)]{\defn, \mset, \form \vdash \form\substposihpi, \ihf{P}{\bar{t}}}
			\infer0[(ax)]{\defn, \mset, \ihf{P}{\bar{t}}, \form \vdash \ihf{P}{\bar{t}}}
			\infer2[($\Rightarrow$L)]{\defn, \mset, \form\substposihpi \Rightarrow \ihf{P}{\bar{t}}, \form \vdash \ihf{P}{\bar{t}}}
			\infer1[($\forall$L)]{\vdots}
			\infer1[($\forall$L)]{\defn, \mset, \form \vdash \ihf{P}{\bar{t}}}
			\infer1[($=$L)]{\vdots}
			\infer1[($=$L)]{\defn, \mset, \bar{z}_P=\bar{t}, \form \vdash \ihf{P}{\bar{z}_P}}
		\end{prooftree}
	\end{equation*}
	In the final application of ($\forall$L), we used the fact that $\mset$ contains the formula $\forall \bar{x}: \form\substposihpi \Rightarrow \ihf{P}{\bar{t}}$.
	Clearly, $\deriv_0$ satisfies the invariant.
	
	Suppose we have constructed the derivation $\deriv_i$ for some $i$.
	If $\deriv_i$ has no bud $n$ without a companion $\repfun(n)$, then it is the final derivation $\deriv_k$ in the sequence $\deriv_0, \dots, \deriv_k$.
	Otherwise, we construct a derivation $\deriv_{i+1}$ from $\deriv_i$ as follows.
	For every bud $n$ in $\deriv_i$ without a companion $\repfun(n)$, labeled with a sequent of the form $\defn, \mset, \form \vdash \form\substposihpi$, we do the following:
	\begin{itemize}
		\item If $\form$ does not contain any predicate of $\Pi$, then $\form\substposihpi \doteq \form$, and we replace $n$ with the proof 
		\begin{equation*}
			\begin{prooftree}
				\infer0[(ax)]{\defn, \mset, \form \vdash \form\substposihpi}
			\end{prooftree}
		\end{equation*}
		\item If $\form$ is of the form $Q(\bar{s})$ such that $Q \in \Pi$ and $\bar{s} \neq \bar{z}_Q$, then we replace $n$ with the derivation
		\begin{equation*}
			\begin{prooftree}
				\hypo{\defn, \mset, Q(\bar{z}_Q) \vdash \ihf{Q}{\bar{z}_Q}}
				\infer1[(subst)]{\defn, \mset, Q(\bar{s}) \vdash \ihf{Q}{\bar{s}}}
			\end{prooftree}
		\end{equation*}
		Denote the top node by $m$.
		If $\defn, \mset, Q(\bar{z}_Q) \vdash \ihf{Q}{\bar{z}_Q}$ already occurs as the label of an ancestor of $m$, then we pick this ancestor as the companion $\repfun(m)$ of $m$.
		Otherwise, we replace $m$ with a derivation analogous to the derivation $\deriv_0$ of $\defn, \mset, P(\bar{z}_P) \vdash \ihf{P}{\bar{z}_P}$. 
		\item If $\form$ is of the form $\lnot \formtwo$, then we replace $n$ with the derivation
		\begin{equation*}
			\begin{prooftree}
				\hypo{\defn, \mset, \formtwo\substnegihpi \vdash \formtwo}
				\infer1[($\lnot$R)]{\defn, \mset \vdash \lnot \formtwo\substnegihpi, \formtwo}
				\infer1[($\lnot$L)]{\defn, \mset, \lnot \formtwo \vdash (\lnot \formtwo)\substposihpi}
			\end{prooftree}
		\end{equation*}
		Here we used the fact that $(\lnot \formtwo)\substposihpi \doteq \lnot \formtwo\substnegihpi$.
		\item If $\form$ is of the form $\formtwo \land \formthree$, then we replace $n$ with the derivation
		\begin{equation*}
			\begin{prooftree}
				\hypo{\defn, \mset, \formtwo \vdash \formtwo\substposihpi}
				\infer1[(wk)]{\defn, \mset, \formtwo, \formthree \vdash \formtwo\substposihpi}
				\hypo{\defn, \mset, \formthree \vdash \formthree\substposihpi}
				\infer1[(wk)]{\defn, \mset, \formtwo, \formthree \vdash \formthree\substposihpi}
				\infer2[($\land$R)]{\defn, \mset, \formtwo, \formthree \vdash \formtwo\substposihpi \land \formthree\substposihpi}
				\infer1[($\land$L)]{\defn, \mset, \formtwo \land \formthree \vdash (\formtwo \land \formthree)\substposihpi}
			\end{prooftree}
		\end{equation*}
		Here we used the fact that $(\formtwo \land \formthree)\substposihpi \doteq \formtwo\substposihpi \land \formthree\substposihpi$.
		\item If $\form$ is of the form $\forall x: \formtwo$, then we replace $n$ with the derivation
		\begin{equation*}
			\begin{prooftree}
				\hypo{\defn, \mset, \formtwo \vdash \formtwo\substposihpi}
				\infer1[($\forall$L)]{\defn, \mset, \forall x: \formtwo \vdash \formtwo\substposihpi}
				\infer1[($\forall$R)]{\defn, \mset, \forall x: \formtwo \vdash (\forall x: \formtwo)\substposihpi}
			\end{prooftree}
		\end{equation*}
		Here we used the fact that $(\forall x: \formtwo)\substposihpi \doteq \forall x: \formtwo\substposihpi$, and that $x$ does not occur freely in $\defn$ or $\mset$.\footnote{
			The fact that $x$ does not occur freely in $\defn$ or $\mset$ follows from the fact that, by construction, the formulae $\form$ in the sequents $\defn, \mset, \form \vdash \form\substposihpi$ and $\defn, \mset, \form\substnegihpi \vdash \form$ that occur as labels of buds in $\deriv_i$ are subformulae of bodies of definitional rules in $\defn$.
		}
		\item The cases for $\lor$, $\Rightarrow$ %
		and $\exists$ are similar to the cases above.
	\end{itemize}
	Note that the invariant is preserved in all cases.
	The procedure for sequents of the form $\defn, \mset, \form\substnegihpi \vdash \form$ is similar to the procedure for sequents of the form $\defn, \mset, \form \vdash \form\substposihpi$.
	
	The above procedure terminates after finitely many steps. 
	This is because at every step, a bud $n$ without companion $\repfun(n)$ and with label $\seqshort$ of the form $\defn, \mset, \form \vdash \form\substposihpi$ or $\defn, \mset, \form\substnegihpi \vdash \form$ is replaced with a derivation $\deriv$ of $\seqshort$ such that either: 
	\begin{itemize}
		\item $\deriv$ is a (non-cyclic) proof of $\seqshort$;
		\item $\deriv$ has a single bud $m$ with a companion $\repfun(m)$ to an ancestor;
		\item $\deriv$ contains a sequent of the form $\defn, \mset, Q(\bar{z}_Q) \vdash \ihf{Q}{\bar{z}_Q}$ for a predicate symbol $Q \in \Pi$ that did not yet occur as the label of an ancestor; or
		\item the buds in $\deriv$ are labeled with sequents of the form $\defn, \mset, \formtwo \vdash \formtwo\substposihpi$ or\\ $\defn, \mset, \formtwo\substnegihpi \vdash \formtwo$, where $\formtwo$ is a strict subformula of $\form$. 
	\end{itemize}
	Since the subformula order is well-founded, a branch in the constructed derivation repeatedly encounters sequents of the form $\defn, \mset, Q(\bar{z}_Q) \vdash \ihf{Q}{\bar{z}_Q}$.
	Since $\Pi$ is finite, a branch must eventually encounter a sequent of this form twice, in which case the first occurrence is picked as the companion of the second occurrence and the procedure terminates on this branch.
	
	Since every bud $n$ in $\deriv_k$ has a companion $\repfun(n)$, $\prf = (\deriv_k, \repfun)$ is an \sccyc-pre-proof.
	It is furthermore an \sccyc-proof, as $\graph{\deriv_k}$ satisfies the global trace condition.
	This is because, by construction, every node  in $\graph{\deriv_k}$ with label of the form $\defn, \mset, Q(\bar{z}_Q) \vdash \ihf{Q}{\bar{z}_Q}$ is derived by %
	(case), and therefore, every infinite branch in $\graph{\deriv_k}$ is infinitely progressing. %
\end{proof}
	
\IndThmIsCycThm*

	\begin{proof}
		Let $\seqshort$ be an \scfoid-theorem.
		Then there exists an \scfoid-proof $\deriv$ of $\seqshort$.
		By Lemma \ref{lem:ind-derivable-in-cyc}, we can replace every instance of (ind) in $\deriv$ with an \sccyc-derivation satisfying the global trace condition, resulting in an \sccyc-pre-proof $\deriv'$.
		Since every infinite branch in $\deriv'$ ends up in one of the substituted \sccyc-derivations that satisfy the global trace condition, $\deriv'$ also satisfies the global trace condition.
		Thus, $\deriv'$ is an \sccyc-proof of $\seqshort$.
	\end{proof}

\end{document}